\documentclass{article}

\usepackage{arxiv}

\usepackage[utf8]{inputenc} 
\usepackage[T1]{fontenc}    
\usepackage{hyperref}       
\usepackage{url}            
\usepackage{booktabs}       
\usepackage{nicefrac}       
\usepackage{microtype}      

\usepackage{amsmath}
\usepackage{amsthm}
\usepackage{amssymb}
\usepackage{amsfonts}
\usepackage{dsfont}
\usepackage{accents}
\usepackage{mathtools}
\usepackage{upgreek}
\usepackage{bm}

\usepackage{tabularx} 
\usepackage{caption}
\usepackage{multirow}
\usepackage{booktabs}
\usepackage[skins]{tcolorbox}
\usepackage{graphicx}
\usepackage{graphics}
\graphicspath{{Figures/}}
\usepackage{subfigure}
\usepackage{tikz,pgfplots}
\usepgfplotslibrary{groupplots}
\usepackage{adjustbox}
\usetikzlibrary{intersections,calc,arrows,matrix,spy}
 \pgfplotsset{
        table/search path={Figures/},
    }
\usepackage{color}
\usetikzlibrary{arrows.meta}

\definecolor{darkred}{rgb}{0.6,0,0}
\definecolor{darkgreen}{rgb}{0,0.5,0}
\definecolor{darkblue}{rgb}{0,0,0.5}
\pgfplotsset{compat=1.5.1}
\usepackage{algorithm}
\usepackage{algorithmic}
\usepackage[squaren,Gray]{SIunits}

\newcommand{\fd}{\mathbf{f}}
\newcommand{\kd}{\mathbf{k}}
\newcommand{\xd}{\mathbf{x}}
\newcommand{\zd}{\mathbf{z}}
\newcommand{\md}{\mathbf{m}}
\newcommand{\qd}{\mathbf{q}}
\newcommand{\yd}{\mathbf{y}}

\newcommand{\sd}{\mathbf{s}}
\newcommand{\ud}{\mathbf{u}}
\newcommand{\uind}{\mathbf{u}^\mathrm{in}}
\newcommand{\omegad}{\bm{\upomega}}
\newcommand{\Gtild}{\tilde{G}}
\newcommand{\Gtildd}{\mathbf{\Gtild}}
\newcommand{\Gd}{\mathbf{G}}
\newcommand{\Id}{\mathbf{I}}
\newcommand{\kb}{k_\mathrm{b}}

\newcommand{\Nf}{N}
\newcommand{\Nmeas}{M}

\newcommand{\Pd}{\mathbf{P}}
\newcommand{\wl}{\lambda}
\newcommand{\nb}{\eta_\mathrm{b}}
\newcommand{\uin}{u^\mathrm{in}}

\newcommand{\usc}{u^\mathrm{sc}}
\def\LS{LS{} model}

\newcommand{\argmin}[2]{\arg\,\underset{#1}{\min} \, #2}

\newcommand{\dd}{\mathrm{d}}
\newcommand{\diag}[1]{\mathbf{diag}(#1)}
\newcommand{\ii}{\mathrm{j}}

\def\C{\mathbb{C}}
\def\R{\mathbb{R}}
\def\N{\mathbb{N}}
\def\Z{\mathbb{Z}}
\newcommand{\disp}{\displaystyle}                                           

\newcommand{\ie}{\textit{i.e., }}                                            
\newcommand{\eg}{\textit{e.g., }}                                            

\newtheorem{remark}{Remark}[section]

\newtheorem{lemma}{Lemma}[section]
\newtheorem{theorem}[lemma]{Theorem}
\newtheorem{proposition}[lemma]{Proposition}

\usepackage{mathtools}
\usepackage{cuted}

\usepackage[noadjust]{cite}

\begin{document}
\title{Three-Dimensional Optical Diffraction Tomography with Lippmann-Schwinger Model}

\author{Thanh-an~Pham$^1$,
Emmanuel~Soubies$^2$,
Ahmed Ayoub$^3$,
Joowon~Lim$^3$,
Demetri~Psaltis$^3$, and
Michael~Unser$^1$\\
$^1$ Biomedical Imaging Group, \'Ecole polytechnique f\'ed\'erale de Lausanne,  Lausanne, Switzerland.\\
$^2$ IRIT, Université de Toulouse, CNRS, Toulouse, France.\\
$^3$ Optics Laboratory, \'Ecole polytechnique f\'ed\'erale de Lausanne, Lausanne, Switzerland.\\
e-mail: thanh-an.pham@epfl.ch.
}


\maketitle

\begin{abstract}
 A broad class of imaging modalities involve the resolution of an inverse-scattering problem. Among them, three-dimensional optical diffraction tomography (ODT) comes with its own challenges. These include a limited range of views, a large size of the sample with respect to the illumination wavelength, and optical aberrations that are inherent to the system itself. In this work, we present an accurate and efficient implementation of the forward model. It relies on the exact (nonlinear) Lippmann-Schwinger equation. We address several crucial issues such as the discretization of the Green function, the computation of the far field, and the estimation of the incident field. We then deploy this model in a regularized variational-reconstruction framework and show on both simulated and real data that it leads to substantially better reconstructions than the approximate models that are traditionally used in ODT.
\end{abstract}



%

\section{Introduction}

Optical diffraction tomography (ODT) is a noninvasive quantitative imaging modality~\cite{wolf1969three, jin2017tomographic}.
This label-free technique allows one to determine a three-dimensional map of the refractive index (RI) of samples, which is of particular interest for applications that range from biology~\cite{liu2016cell} to  nanotechnologies~\cite{zhang2016far}.
The acquisition setup sequentially illuminates the sample from different angles.
For each illumination, the outgoing complex wave field (\ie{} the scattered field) is recorded by a digital-holography microscope~\cite{goodman1967digital,kim2010principle}.
Then, from this set of measurements, the RI of the sample can be reconstructed by solving an inverse-scattering problem.
However, its resolution is very challenging due to the nonlinear nature of the interaction between the light and the sample.

\subsection{Related Works}\label{sec:relworks}

To simplify the reconstruction problem, pioneering works focused on linearized models. These include Born~\cite{wolf1969three} and Rytov~\cite{devaney1981inverse} approximations, which are valid for weakly scattering samples~\cite{chen1998validity}. Although originally used to deploy direct inversion methods, these linearized models have been later
combined with iterative regularization techniques to improve their robustness to noise and to alleviate the missing-cone problem~\cite{sung2009optical, lim2015comparative}.

Nonlinear models that adhere more closely to the physic  of the acquisition are needed to recover samples with higher variations of their refractive index. For instance, beam-propagation methods (BPM)~\cite{kamilov2015learning, kamilov2016optical,lim2017assessment,lim2019high} rely on a slice-by-slice propagation model that accounts for multiple scatterings within the direction of propagation (no reflection). Other nonlinear models include the contrast source-inversion method~\cite{Abubakar2002} or the recursive Born approximation~\cite{Kamilov2016}. Although more accurate, all these models come at the price of a large computational cost. 

The theory of scalar diffraction recognizes the Lippmann-Schwinger (LS) model to be the most faithful. It accounts for multiple scatterings, both in transmission and reflection. Iterative forward models that solve the LS equation have been successfully used in~\cite{liu2017seagle,soubies2017efficient,ma2018accelerated} to reconstruct two-dimensional samples from data acquired in the radio-frequency regime. An alternative approach is known as the discrete dipole approximation (DDA) which, in addition, can account for polarized light~\cite{draine1994discrete,girard2010nanometric,zhang2016far}.

Finally, it is noteworthy to mention that the aforementioned approaches have been extended to the phaseless (\ie intensity-only) inverse-scattering problem~\cite{maiden2012ptychographic,tian20153d,horstmeyer2016diffraction,pham2018versatile,unger2019versatile}.

\subsection{Challenges in Three-Dimensional ODT}

  So far, the use of the more sophisticated LS model and DDA has been mostly limited to microwave imaging~\cite{chaumet2009three,abubakar2012application,maire2013high} (see also the numerous references listed in~\cite{litman2009testing}).
  Although led by the same underlying physics, ODT differs from microwave imaging on several aspects that further increases the difficulty of the reconstruction problem.

    \begin{itemize}
        \item The direction of propagation of the incident wave is restricted to a small cone around the optical axis (see Figure~\ref{fig:ODT}). This  lack of measurements leads to the well-known missing-cone problem~\cite{lim2015comparative}.
        \item In typical ODT applications such as biology, the size of the sample is significantly larger (\eg $100 \times$) than the wavelength of the incident wave. This requires a fine discretization that entails very large memory requirements.
        \item The large size of the detector  leads to numerical challenges for the computation of the far-field.  
        \item The benefit of a theoretical expression of the incident wave field, as used in microwave imaging~\cite{litman2009testing}, is made unlikely in ODT due to unknown distortions that are inherent to the system.
    \end{itemize}
      These challenges hindered the adoption of sophisticated models in ODT, with notable exceptions~\cite{zhang2016far,girard2010nanometric} that focused on the reflective mode and considered relatively simple non-biological samples.

\subsection{Contributions and Outline}\label{sec:contributions}

  This paper builds upon the prior works~\cite{liu2017seagle,soubies2017efficient,ma2018accelerated} that are dedicated to the resolution of the 2D inverse scattering problem using an iterative LS forward model. We propose to  extend these works to the 3D ODT problem. Our main contribution is the development of an accurate and efficient implementation of the forward model in 3D. This is crucial to obtain good reconstructions while keeping the computational burden of the method reasonable for large-scale volumes. 
  
  More precisely, we provide a description on how to implement the iterative LS forward model by tackling three challenging difficulties.
  \begin{itemize}
      \item \textit{Discretization of the Green function (Section~\ref{sec:vainikko}).}
      Following an idea proposed by Vainikko~\cite{vainikko1997fast}, we derive an accurate discretization of the Green function and analyze the  errors that are produced when convolving it with a given vector (Theorem~\ref{th:DiscG}). Moreover, we propose a new way of building the discrete Green kernel that avoids a large memory overhead (Proposition~\ref{th:ReducedMem}).
      \item \textit{Computation of the far field (Sections~\ref{sec:far-field} and~\ref{sec:defocus}).}
      We combine the convolutional nature of the model with the fact that the measurements lie on a plane  to derive an efficient method to evaluate the far field.
      \item \textit{Estimation of the incident field (Section~\ref{sec:uin3D}).}
      We build the volume of the incident field through numerical propagation of a real acquisition of it at the detector plane. In particular, we propose a strategy that results in significantly reduced numerical errors.
  \end{itemize}
  Let us emphasize that, to the best of our knowledge, the present paper is the first to provide practical details (\eg discretization, speedup, and memory-saving strategies) concerning the implementation of the iterative LS model in ODT.
  
  Finally, to deal with the missing-cone problem, we deploy a regularized variational reconstruction approach (Section~\ref{sec:probform}). We then present in Section~\ref{sec:results} reconstructions of biological samples for both simulated and real data, and compare them to those of baselines methods.
  
\subsection{Notations}\label{sec:notations}
Scalar and continuously defined functions are denoted by italic letter~(\eg $\eta \in \R$, $g \in L_2(\R)$).
Vectors and matrices are denoted by bold lowercase and bold uppercase letters, respectively (\eg $\fd \in \R^\Nf$, $\Gd \in \C^{\Nf \times \Nf}$).
For a vector $\fd \in \R^N$, $\|\fd\|$ stands for its $\ell_2$-norm. Other $p$-norms will be specified with an index (\ie $\|\cdot \|_p$). The $n$th element of a vector is denoted as $\fd[n]$.  
Then, we denote by $\mathbf{F}$ the discrete Fourier transform (DFT) defined in 1D by $(\mathbf{F}\mathbf{v})[k] = \sum_{n=-N/2+1}^{N/2} \mathbf{v}[n] \mathrm{e}^{\frac{-2\ii \pi}{N}nk}$. (The higher-dimension DFT follows by recursive application of the 1D DFT along each dimension.) The notations $\hat{f}$ and $\hat{\fd}$ refer to the continuous Fourier transform of $f$ and the discrete Fourier transform of $\fd$, respectively. Finally, $\odot$ stands for the Hadamard product and $[\![1;N]\!] := [1\ldots N]$.

\tikzset{
pics/cone/.style args={#1}{
  code={
    \draw [dotted,fill=#1!12,thick,join=round,opacity=0.5](0,0) -- (2,-.7) -- (2,.7) --cycle;
    \draw [fill=#1!12,dashed,thin,opacity=0.5](2,0) ellipse (.4 and .7);
    \draw [dashed,thin,opacity=0.5](1.2,0) ellipse (.24 and .42);
    \draw [dashed,thick] (-2.8,0) -- (2,0);
    \draw [red,opacity=0.7,-latex,thin] (2,0.7) -- (1.2,0.42);
    \draw [red,opacity=0.7,-latex,thin] (2,-0.7) -- (1.2,-0.42);
    \draw [red,-latex,thin] (1.6,0) -- (1,0);
    \draw [red,opacity=0.5,-latex,thin] (2.35,-0.25) -- (1.4,-0.15);
    \draw [red,opacity=0.5,-latex,thin] (2.35,0.25) -- (1.4,0.15);
    \draw [red,-latex,thin] (1.65,-0.35) -- (1,-0.21);
    \draw [red,-latex,thin] (1.65,0.35) -- (1,0.21);

  }
}
}

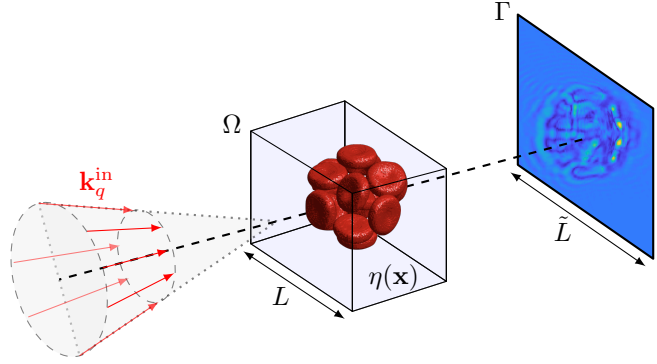
\begin{figure}
    \centering
\begin{tikzpicture}
\draw [fill stretch image={ytot_RBC.png},thick,join=round](3.2,0.75) -- (5,-0.5) -- (5,1.5) -- (3.2,2.75) --cycle;
%
\node at (3,2.8) {$\Gamma$};
\draw [fill=blue!10,opacity=0.5,join=round](0.9,0.1) -- (-0.35,-0.3) -- (1,-1.2)  -- (2.25,-0.8)  --cycle;
\draw [fill=blue!10,opacity=0.5,join=round](0.9,0.1) -- (2.25,-0.8) -- (2.25,0.7) -- (0.9,1.6) --cycle;
\draw [join=round](0.9,0.1) -- (2.25,-0.8) -- (2.25,0.7) -- (0.9,1.6) --cycle;
\draw [fill=blue!10,opacity=0.5,join=round](0.9,0.1) -- (0.9,1.6) -- (-0.35,1.2) -- (-0.35,-0.3) --cycle;
\draw [join=round](0.9,0.1) -- (0.9,1.6) -- (-0.35,1.2) -- (-0.35,-0.3) --cycle;
\path (0,0) pic [rotate=-165,scale=1.5]{cone=gray};
\node at (1,0.3)
{\includegraphics[scale=0.075]{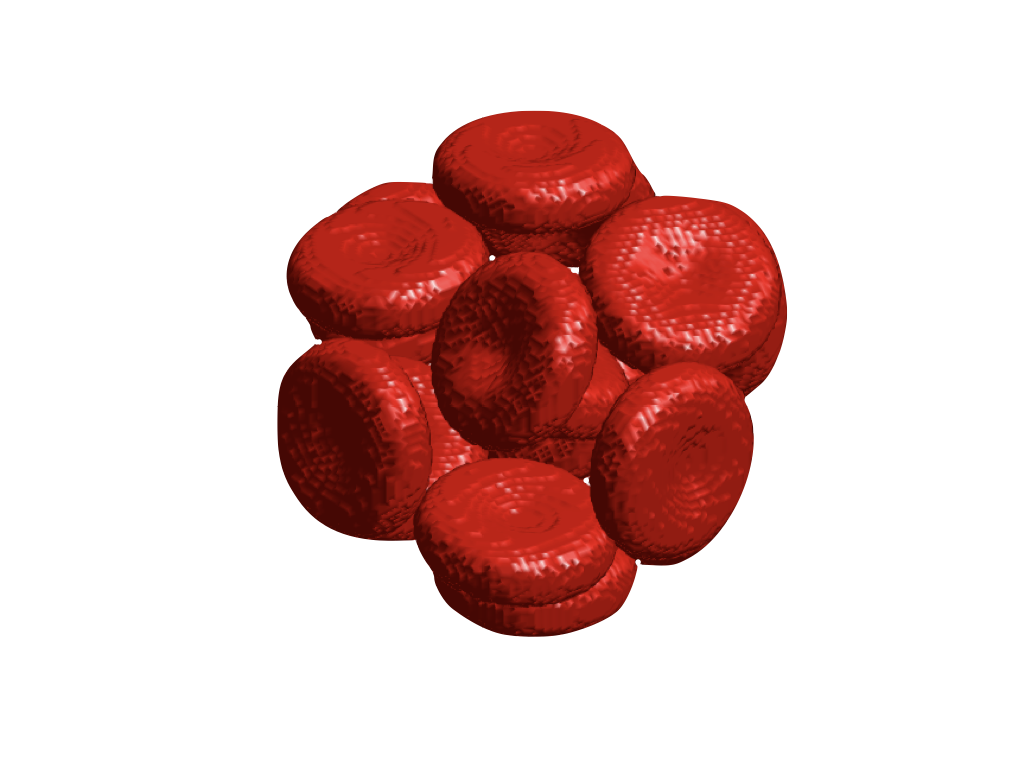}};
\draw [join=round](-0.35,-0.3) -- (1,-1.2) -- (1,0.37) -- (-0.35,1.2) --cycle;
\draw [join=round] (1,-1.2) -- (1,0.37) --  (2.25,0.7)  --(2.25,-0.8) --cycle;
\node at (-0.6,1.3) {$\Omega$};
\draw [latex-latex](-0.45,-0.4) -- (0.9,-1.3)   ;
\node at (0.05,-1) {$L$};
\draw [latex-latex](3.1,0.65) -- (4.9,-0.6)  ;
\node at (3.8,-0.1) {$\tilde{L}$};
\node at (1.55,-0.75) {$\eta(\xd)$};
\node[red] at (-2.4,0.5) {$\kd_q^\mathrm{in}$};
\end{tikzpicture} 
    \caption{Principle of optical diffraction tomography. The arrows represent the wave vectors $\{\mathbf{k}^\mathrm{in}_q \}_{q=1}^Q \in \R^3$ of the $Q$ incident plane waves $\{u^\mathrm{in}_q\}_{q=1}^Q$ which are limited to a cone around the optical axis.}
    \label{fig:ODT}
\end{figure}

\section{Physical Model}\label{sec:physical}
\subsection{Continuous-Domain Formulation}
Let $\eta : \Omega \rightarrow \R$ denotes the continuously-defined refractive index of a sample whose support is assumed to be included in the region of interest $\Omega \subset \R^3$. Without loss of generality and to simplify the presentation, let us consider that $\Omega = [-L/2,L/2]^3$ for $L>0$.
The interaction of the sample with a monochromatic incident field~$\uin : \Omega \rightarrow \C$ of wavelength~$\wl$ produces a scattered field $\usc : \Omega \rightarrow \C$. The resulting total field $u = \usc + \uin$  is governed by the Lippmann-Schwinger  equation
\begin{equation}
\label{eq:lipp}
    u(\xd) = \uin(\xd) + \int_\Omega g(\xd - \zd) f(\zd) u(\zd) \, \dd{\zd},
\end{equation}
where $f(\xd) = \kb^2 \left({\eta(\xd)^2}/{\nb^2} - 1\right)$ is the scattering potential. Here, $\kb = {2\pi \nb}/{\wl}$ is the wavenumber in the surrounding medium and $\nb$ the  corresponding refractive index. Finally, $g: \Omega \rightarrow \C$ is the free-space Green function which, under Sommerfeld’s radiation condition,  is given by~\cite{Schmalz2010}
\begin{equation}
\label{def:Green}
    g(\xd) = \frac{\exp\left( \ii \kb \|\xd \| \right)}{4\pi \|\xd \|}.
\end{equation}
Equation~\eqref{eq:lipp} completely characterizes the image formation model in ODT. Using an interferometric setup,   the total field $u$ is recorded at the focal plane $\Gamma = [-\tilde{L}/2,\tilde{L}/2]^2$, $\tilde{L} \geq L$, of the camera. This focal plane lies outside $\Omega$ at a distance denoted by $x_\Gamma>0$. Finally, we denote by  $M=m^2$ the number of pixels of the detector.

\subsection{Discrete Formulation}\label{sec:discreteForm}

To numerically solve the ODT inverse problem,~\eqref{eq:lipp} has to be properly discretized. To do so, we  first discretize $\Omega$ into $\Nf = n^3$ voxels\footnote{The generalization to the case where there is a different number of points in each dimension is straightforward.}.
Then, the computation of the scattered field~$\yd^\mathrm{sc} \in \C^\Nmeas$ at the camera plane $\Gamma$ follows a two-step process~\cite{liu2017seagle,soubies2017efficient},
\begin{align}\label{eq:lipp1}
    \ud &= \left(\Id - \Gd \, \diag{\fd}\right)^{-1}\uind\\
    \yd^\mathrm{sc} &= \Pd \Gtildd \, \diag{\fd}\ud,\label{eq:lipp2}
\end{align}
where $\Id \in \R^{\Nf \times \Nf}$ is the identity matrix,
\mbox{$\diag{\fd} \in \R^{\Nf \times \Nf}$} is a diagonal matrix formed out of the entries of~$\fd$, and
$\fd \in \R^N$, $\uind \in \C^\Nf$, and $\ud \in \C^\Nf$ are sampled version of $f$, $\uin$, and $u$ within~$\Omega$, respectively. The  matrix $\Gd \in \C^{\Nf \times \Nf}$ is the discrete counterpart of the continuous convolution with the Green function in~\eqref{eq:lipp}~(see Section~\ref{sec:vainikko}). Similarly,  $\Gtildd \in \C^{\Nmeas \times \Nf}$ is a matrix that, given $\ud$ and $\fd$ inside $\Omega$, gives the total field at the measurement plane $\Gamma$ (see Section~\ref{sec:far-field}). Finally, $\Pd \in \C^{\Nmeas \times \Nmeas}$ models the effect of the pupil function of the microscope and can also encode the contribution of a free-space propagation to account for an optical refocus of the measurements.

One will have noticed that~\eqref{eq:lipp1} requires the resolution of a linear system.
This can be efficiently performed using a conjugate-gradient method~\cite{soubies2017efficient} or a biconjugate-gradient stabilized method~\cite{horst1992BiCGSTAB}.
Yet,~\eqref{eq:lipp1} carries the main computational complexity of the forward process~\eqref{eq:lipp1}-\eqref{eq:lipp2}.
To obtain the scattered field at the camera plane $\Gamma$, a naive approach would be to compute the total field $\ud$ in~\eqref{eq:lipp1} on a large region that includes $\Gamma$. Here, the introduction of  $\Gtildd$ allows one to restrict the computation of  $\ud$ to the smaller region $\Omega$ as soon as it fully contains the support of the sample~\cite{liu2017seagle,soubies2017efficient}. This significantly reduces the computational burden of the forward process.

Needless to say, the matrices $\Gd$, $\Gtildd$, and $\Pd$ are never explicitly built. Instead, we exploit the fact that the application of the corresponding linear operators can be efficiently performed using the fast Fourier transform (FFT).

\section{Accurate and Efficient Implementation of the Forward Model}

\subsection{Green's Function Discretization for the Volume: $\Gd$}
\label{sec:vainikko}

Because of the singularity of the Green function~\eqref{def:Green} as well as of its Fourier transform (\ie $\hat{g}(\omegad) ={1}/({\kb^2 - \|\omegad\|^2})$), $\Gd$ in~\eqref{eq:lipp1} cannot be defined through a naive discretization of $g$.
In this section, we describe how $\Gd$ has to be defined in order to minimize the approximation error with respect to the continuous model~\eqref{eq:lipp}.

   First, let us recall that we aim at computing the total field $u$ only inside $\Omega$ and that the support of $f$ is itself assumed to be included in~$\Omega$. Hence,~\eqref{eq:lipp} can be equivalently written as, $\forall \xd \in \Omega$,
   \begin{equation}\label{eq:lippTrunc}
    u(\xd) = \uin(\xd) + \int_\Omega g_{\mathrm{t}}(\xd - \zd) f(\zd) u(\zd) \, \dd{\zd},
\end{equation}
where $g_{\mathrm{t}}$ is a truncated version of the Green function. More precisely, $g_{\mathrm{t}}$ is defined by
\begin{equation}
    g_{\mathrm{t}}(\xd) = \mathrm{rect}\left(\frac{\| \xd \|}{2\sqrt{3}L}\right) g(\xd),
\end{equation}
where $\mathrm{rect}(x) = \{ 1, |x|\leq 1/2; \, 0, \text{ otherwise}\}$. With this definition, one easily gets the equivalence between~\eqref{eq:lipp} and~\eqref{eq:lippTrunc}, as illustrated in Figure~\ref{fig:LippTrunc}.
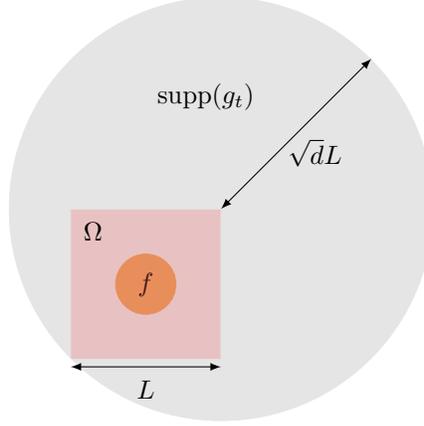
\begin{figure}
    \centering
      \begin{tikzpicture}
      \draw[orange!60,fill=orange!60] (1,1) circle (0.4);
      \node[black] at (1,1) {$f$};
      \draw[white,fill=gray,opacity=0.2] (2,2) circle (2.8284);
      \draw[red!20,fill=red,opacity=0.15] (0,0) rectangle (2,2);
      \node[black] at (0.3,1.7) {$\Omega$};
      \draw[black,latex-latex] (0,-0.1) -- (2,-0.1);
      \node[black] at (1,-0.4) {$L$};
      \draw[black,latex-latex] (2,2) -- (4,4);
      \node[black] at (3.25,2.75) {$\sqrt{d}L$};
      \node[black] at (1.8,3.5) {$\mathrm{supp}(g_t)$};
      \end{tikzpicture}
    \caption{Illustration in dimension two (\ie $d=2$) of the equivalence between~\eqref{eq:lipp} and~\eqref{eq:lippTrunc}.}
    \label{fig:LippTrunc}
\end{figure}

To the best of our knowledge, this observation has to be attributed to Vainikko~\cite{vainikko1997fast} but has then been revitalized by Vico \textit{et al.}~\cite{vico2016fast}. It is essential to a proper discretization of the Lippmann-Schwinger equation~\eqref{eq:lipp}. Specifically, we have that
\begin{equation}
\label{eq:GreenL}
    \widehat{g_{\mathrm{t}}}(\omegad) = \frac{1}{\|\omegad\|^2 - \kb^2}\bigg(1 - \mathrm{e}^{\ii \sqrt{3} L \kb}( \cos(\sqrt{3}L \|\omegad\|) + \ii \kb \sqrt{3}L \,\mathrm{sinc}( \sqrt{3}L\|\omegad\|) ) \bigg)
\end{equation}
for $\|\omegad\| \neq \kb$, which can be extended by continuity as
\begin{equation}\label{eq:GreenL2}
    \widehat{g_{\mathrm{t}}}(\omegad) =\ii \left(\frac{\sqrt{3}L}{2\kb} - \frac{\mathrm{e}^{\ii \sqrt{3} L \kb}}{2\kb^2}\sin(\sqrt{3} L\kb) \right)
\end{equation}
when $\|\omegad\| = \kb$. The practical outcome is that~\eqref{eq:lippTrunc} can now be discretized in the Fourier domain since $\widehat{g_{\mathrm{t}}}$ is a smooth function. 

We now show how  $g_\mathrm{t} \ast v$, for $v \in L_2(\R^3)$, can be numerically evaluated using FFTs and we provide error bounds on the approximation. The proof is provided in Appendix~\ref{proof:thDiscG}.

\begin{figure}[ht]
\centering
\begin{tikzpicture}[anchor=north]
\newcommand{\miewidth}{0.2}
\node[anchor=north] (Mie3D) at (0,0) {\includegraphics[width=\miewidth\textwidth,trim ={3.5cm 1.5cm 5cm 2cm}, clip]{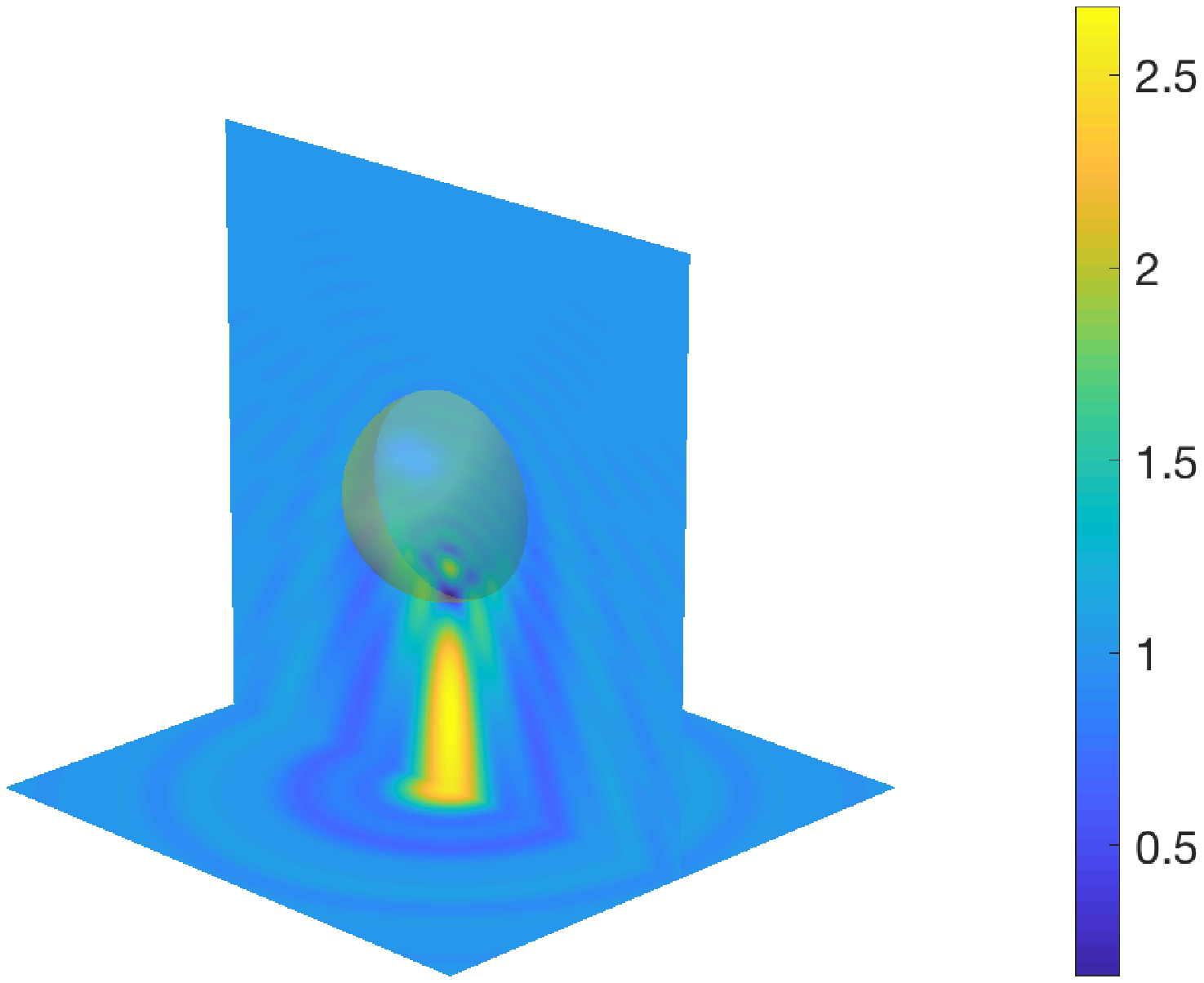}};
\node[anchor=base, baseline = center] (labelMie) at ($(Mie3D.north) - (0,0.5em)$) {Analytical solution};

\node[anchor=west] (Born3D) at ($(Mie3D.east) + (0.5cm,0)$) {\includegraphics[width=\miewidth\textwidth,trim ={3.5cm 1.5cm 5cm 2cm}, clip]{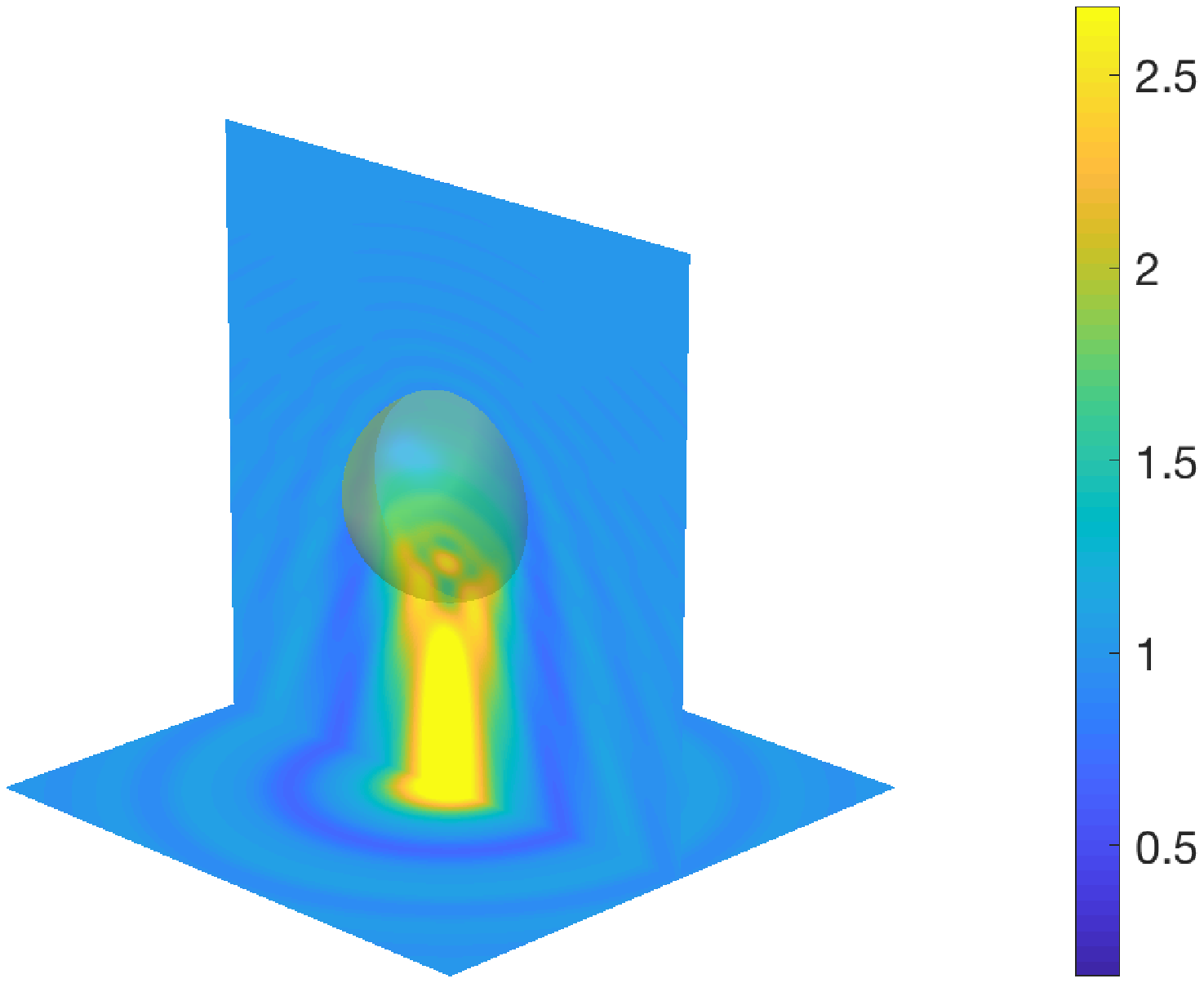}};
\node[anchor=base, baseline = center] (labelBorn) at ($(Born3D.north) - (0,0.5em)$) {Born};

\node[anchor=north] (BPM3D) at ($(Mie3D.south) - (0,0.05cm)$) {\includegraphics[width=\miewidth\textwidth,trim ={3.5cm 1.5cm 5cm 2cm}, clip]{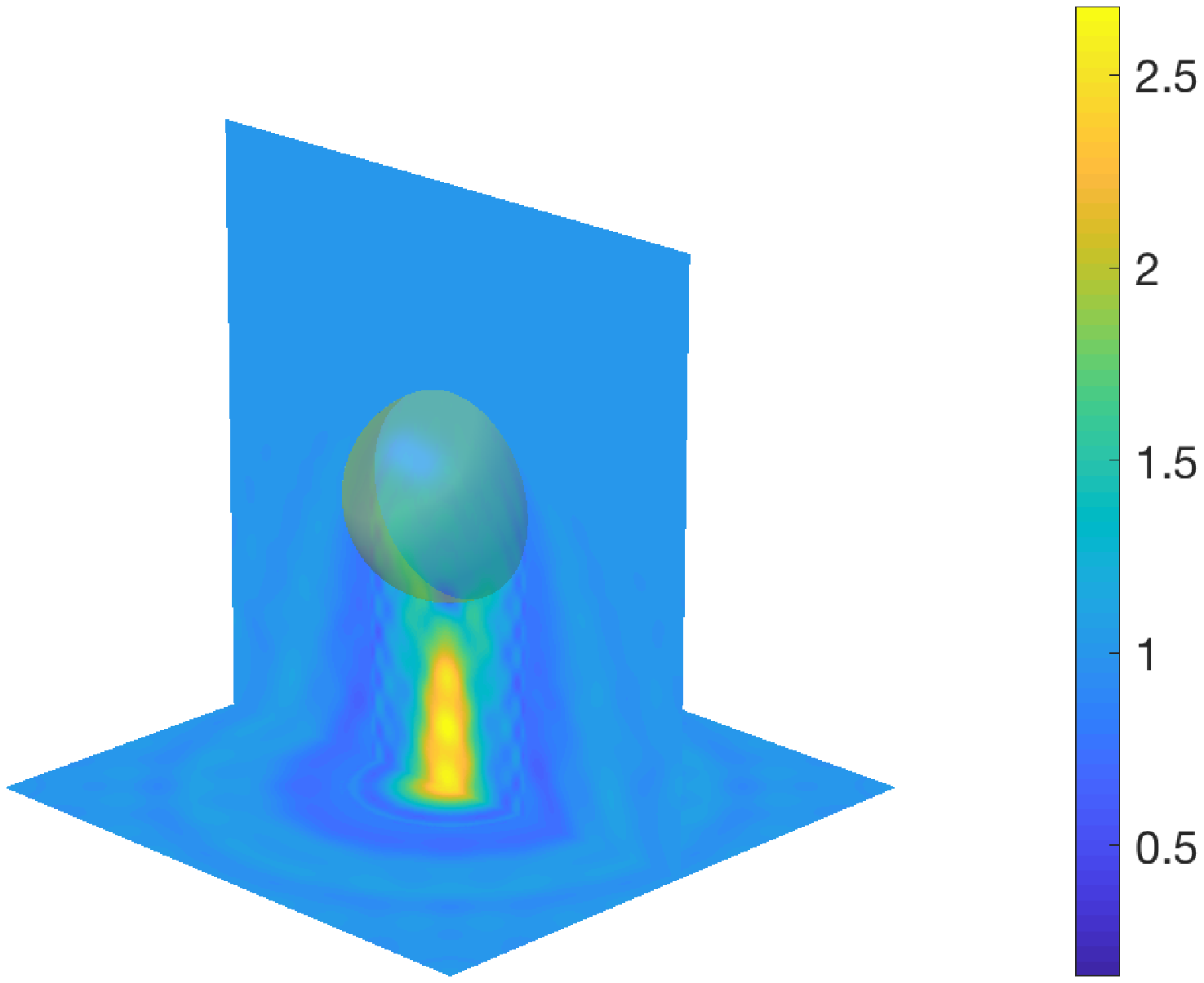}};
\node[anchor=base, baseline = center] (labelBPM) at ($(BPM3D.north) - (0,0.5em)$) {BPM};

\node[anchor=west] (Lipp3D) at ($(BPM3D.east) + (0.5cm,0)$) {\includegraphics[width=\miewidth\textwidth,trim ={3.5cm 1.5cm 5cm 2cm}, clip]{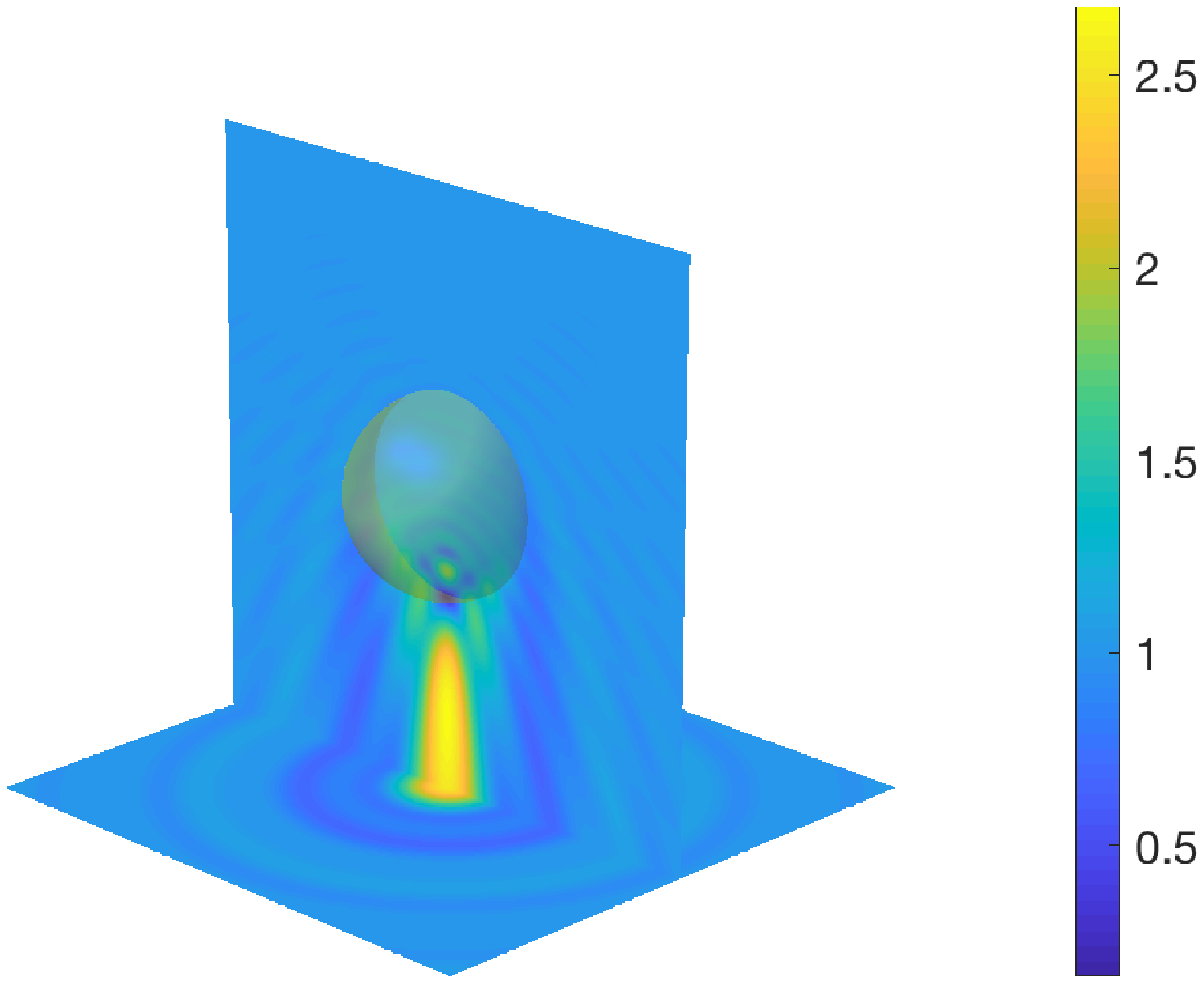}};
\node[anchor=base, baseline = center] (labelLipp) at ($(Lipp3D.north) - (0,0.5em)$) {LS model};
\end{tikzpicture}

\caption{Simulated scattering of a monochromatic wave~($\lambda=532$nm) by a bead embedded in water~($\nb = 1.3388$).
The bead has a diameter of $3\lambda$ and a refractive index of $1.4388$.
The reported total fields are obtained through the analytical solution, the Born model, and the Lippmann-Schwinger iterative forward model for $p=4$ and $h = \lambda/16$ (\ie $n=144$).
}
\label{fig:Mie}
\end{figure}

\begin{theorem}\label{th:DiscG}
   Let $v \in L_2\left([-\frac{L}{2},\frac{L}{2}]^3\right)$ and $\mathbf{v} \in \C^N$ be the sampled version of $v$ using $n>\kb L / \pi$ sampling points in each dimension ($N=n^3$). Let $\mathbf{v}_p$ be the $p$-times zero-padded version of $\mathbf{v}$. Define $h=L/n$ and $\delta = 2\pi/(Lp)$. Then, $\forall \kd \in  [\![\frac{-n}{2}+1;\frac{n}{2}]\!]^3$
     \begin{equation}\label{th:DiscG-eq}
         (\mathbf{G} \mathbf{v})[\kd] = \big(\mathbf{F}^{-1}(\widehat{\mathbf{g}_{\mathrm{t}}} \odot \widehat{\mathbf{v}_p}) \big)[\kd],
    \end{equation}
where  $\widehat{\mathbf{g}_{\mathrm{t}}} = (\widehat{g_{\mathrm{t}}}(\delta \qd))_{\qd \in [\![\frac{-np}{2}+1;\frac{np}{2}]\!]^3}$ and $\widehat{\mathbf{v}_p} = \mathbf{F}\mathbf{v}_p$.

Moreover, if $v$ has ($q-1$) continuous derivatives for $q\geq 3$ and a qth derivative of bounded variations, we have the  error bound
\begin{equation}\label{th:DiscG-error}
    |(g_\mathrm{t} \ast v)(h\kd) - (\mathbf{G} \mathbf{v})[\kd]| \leq \frac{C^\mathrm{tr}}{n^q} + \frac{C^\mathrm{al}}{n^{q-2}} + \frac{C^\mathrm{tp}}{p^2},
\end{equation}
where  $C^\mathrm{al}$, $C^\mathrm{tr}$, and $C^\mathrm{tp}$ are positive constants that are associated to the errors due to the aliasing in $\mathbf{v}$, the truncation of the Fourier integral, and the trapezoidal quadrature rule used to approximate this integral, respectively. 
\end{theorem}

\begin{remark}
   Equation~\eqref{th:DiscG-eq} is hiding a cropping operation. Indeed, the result of $\mathbf{F}^{-1}(\widehat{\mathbf{g}_{\mathrm{t}}} \odot \widehat{\mathbf{v}_p})$ is defined on the grid $[\![\frac{-np}{2}+1;\frac{np}{2}]\!]^3$ but we only retain the elements that belong to $ [\![\frac{-n}{2}+1;\frac{n}{2}]\!]^3$.  
\end{remark}

\begin{remark}
   The assumption that $n>\kb L / \pi \Leftrightarrow \kb < \pi /h$ ensures that the ``peaks'' of $|\hat{g_{\mathrm{t}}}(\omegad)|$ for $\|\omegad\| = \kb$ are included in the frequency domain associated to the DFT (\ie $[-\pi/h,\pi/h]^3$).  This is a natural and minimal requirement to reduce the approximation error. 
\end{remark}

  From Theorem~\ref{th:DiscG}, one sees that the number of sampling points $n$ controls both the aliasing error and the error due to the truncation of the Fourier integral. It is noteworthy that these bounds decrease with the smoothness of $v$ (\ie $q$).   On the other hand, the padding factor $p$ controls the error that results from the trapezoidal quadrature rule. 
  
  \begin{remark}\label{remark:p4}
       A simple argument suggests that the padding factor should be at least $p=4$  to properly capture the oscillations of $\widehat{g_{\mathrm{t}}}$. Indeed, in the spatial domain, the diameter of the support of $g_{\mathrm{t}}$ is $2\sqrt{3}L \approx 3.4 L$. Hence, in order to satisfy the Shannon-Nyquist criterion, the considered spatial domain should be at least of size $4L$, which corresponds to a padding factor  $p=4$.
  \end{remark}

To assess the practical accuracy of the implementation of $\Gd$ provided by Theorem~\ref{th:DiscG}, we consider  the interaction of a plane wave with a bead. For this particular setting, an analytical expression of the scattered field is known~\cite{devaney2012mathematical}. In Figure~\ref{fig:Mie} we compare the theoretical scattered field with the scattered field obtained through the  Born linearization, the BPM, and the resolution of~\eqref{eq:lipp1} with $\Gd$ implemented according to Theorem~\ref{th:DiscG}. One can appreciate the gain in accuracy that the proposed method brings over the standard approximations used in~ODT.

\subsubsection*{Memory Savings}  
 According to Theorem~\ref{th:DiscG}, an accurate computation of the field inside $\Omega$ requires one to zero-pad the volume~$\mathbf{v}$. From Remark~\ref{remark:p4}, we should set at least  $p=4$. This can lead to severe computational and memory issues for the reconstruction of large 3D volumes. Fortunately, as mentioned in~\cite{vico2016fast}, this computation can be reformulated as a discrete convolution with a modified kernel that only involves the twofold padding $p=2$. We summarize this result in Proposition~\ref{th:ReducedMem} and provide a detailed proof in Appendix~\ref{proof:thReducedMem}. Moreover, we provide an expression of the modified kernel that reveals how one can build it directly on the grid $[\![-n+1;n]\!]^3$.
 
\begin{proposition}\label{th:ReducedMem}
   Let $p \in 2 \N \setminus \{0\}$. Then, $\forall \kd \in [\![\frac{-n}{2}+1;\frac{n}{2}]\!]^3$, we have that
   \begin{equation}
    \big(\mathbf{F}^{-1}(\widehat{\mathbf{g}_{\mathrm{t}}} \odot \widehat{\mathbf{v}_{\mathrm{p}}}) \big) [\kd] = \big(\mathbf{F}^{-1}(\widehat{{\mathbf{g}}^\mathrm{m}_{\mathrm{t}}} \odot \widehat{\mathbf{v}_2}) \big) [\kd] ,
\end{equation}
where ${\mathbf{v}_2}$ is a twofold zero-padded version of $\mathbf{v}$, and ${\mathbf{g}}^\mathrm{m}_{\mathrm{t}}$ is the modified kernel
\begin{equation}
      \mathbf{g}_{\mathrm{t}}^\mathrm{m}[\kd]= \frac{8}{p^3} \mkern-10mu \sum_{\sd \in [\![0; \frac{p}{2}-1]\!]^3} \mkern-10mu \mathbf{F}^{-1}(\widehat{\mathbf{g}_{\mathrm{t}}}[\textstyle \frac{p}{2}\cdot-\sd] )[\kd] \,  \mathrm{e}^{ \frac{-2 \ii \pi}{np} \kd^T \sd} ,
\end{equation}
\end{proposition}

\subsection{Green's Function Discretization for the Measurements:~\texorpdfstring{$\Gtildd$}{$\mathbf{G}$}} \label{sec:far-field}

In  works dedicated to the 2D ODT problem, $\Gtildd \in \C^{M \times N}$ is sometimes accessible explicitly~\cite{liu2017seagle,soubies2017efficient,ma2018accelerated}. By contrast, the scale of the 3D ODT problem prevents this in the present work. 
Fortunately,  we are only interested in the evaluation of the total field at the $M$ voxels of the camera plane. By exploiting this planarity, we can significantly reduce the memory and the computational burden of the evaluation of $\Gtildd \mathbf{v}$.

Let $x_{\Gamma} >0$ be the axial position of the measurement plane $\Gamma$ (\ie $\forall \xd \in \Gamma$, $x_3 = x_\Gamma$). Then, letting $v=f \cdot  u$ and expressing the integral in~\eqref{eq:lipp} using a numerical quadrature along the third dimension, we get, $\forall \xd =(x_1,x_2,x_\Gamma) \in \Gamma$,
\begin{equation}\label{eq:Gtild-1}
     (g \ast v)(\xd) = \mkern-15mu  \sum_{k=-\frac{n}{2}+1}^{\frac{n}{2}} h  \int_{[\frac{-L}{2},\frac{L}{2}]^2}  \mkern-5mu  g (\xd - \zd_k) v(\zd_k) \,\dd{z_{k_1}} \dd{z_{k_2}},
\end{equation}
where $\zd_k = (z_{k_1},z_{k_2},kh)$.

From~\eqref{eq:Gtild-1}, $g \ast v$ is computed as a sum of 2D aperiodic convolutions. Considering that the sampling step at the camera plane $\Gamma$ is identical to that of the volume $\Omega$, the 2D convolutions in~\eqref{eq:Gtild-1} is evaluated in the same way as described in Theorem~\ref{th:DiscG}. 
This strategy reduces the computational complexity of the application of $\Gtildd$ to $\mathcal{O}(nM\log(M))$.
Note that, if the sampling step at the camera plane  is $q$ times that of the volume (\ie $h'= qh$, $q \in \N$), one can simply downsample the result of the above procedure by $q$.

\begin{figure}[t]
\centering 
\begin{adjustbox}{width=0.475\textwidth}
\begin{tikzpicture}[scale = 1,font = \huge]
\filldraw[fill=lightgray] (-2,-2.5) -- (-2,1) --  (-0.5,3) -- (-0.5,-0.5) -- (-2,-2.5);
\draw (-0.3,3) node[anchor=south] {$y^\mathrm{in}=a(\cdot)\mathrm{e}^{\ii  (\cdot)^T\tilde{\kd}^\mathrm{in} }$};

\draw[-latex, line width=2pt] (0,0) node[anchor=east]{} -- (3.5,0) node[anchor=west]{};
\draw (1.75, -1.2) node[anchor=south] {Propagation};

\filldraw[fill=lightgray,name] (4,-2.5) -- (4,1) --  (5.5,3) -- (5.5,-0.5) -- (4,-2.5);
\draw (5,3) node[anchor=south] (target) {$u^\mathrm{in}_\mathrm{prop}$};

\node[inner sep=0pt, anchor = west] (ideal) at (6.5,0) {\includegraphics[trim = {3.93 cm} {9.08 cm} {4.88 cm} {8.45 cm}, clip = true, width=0.4\textwidth]{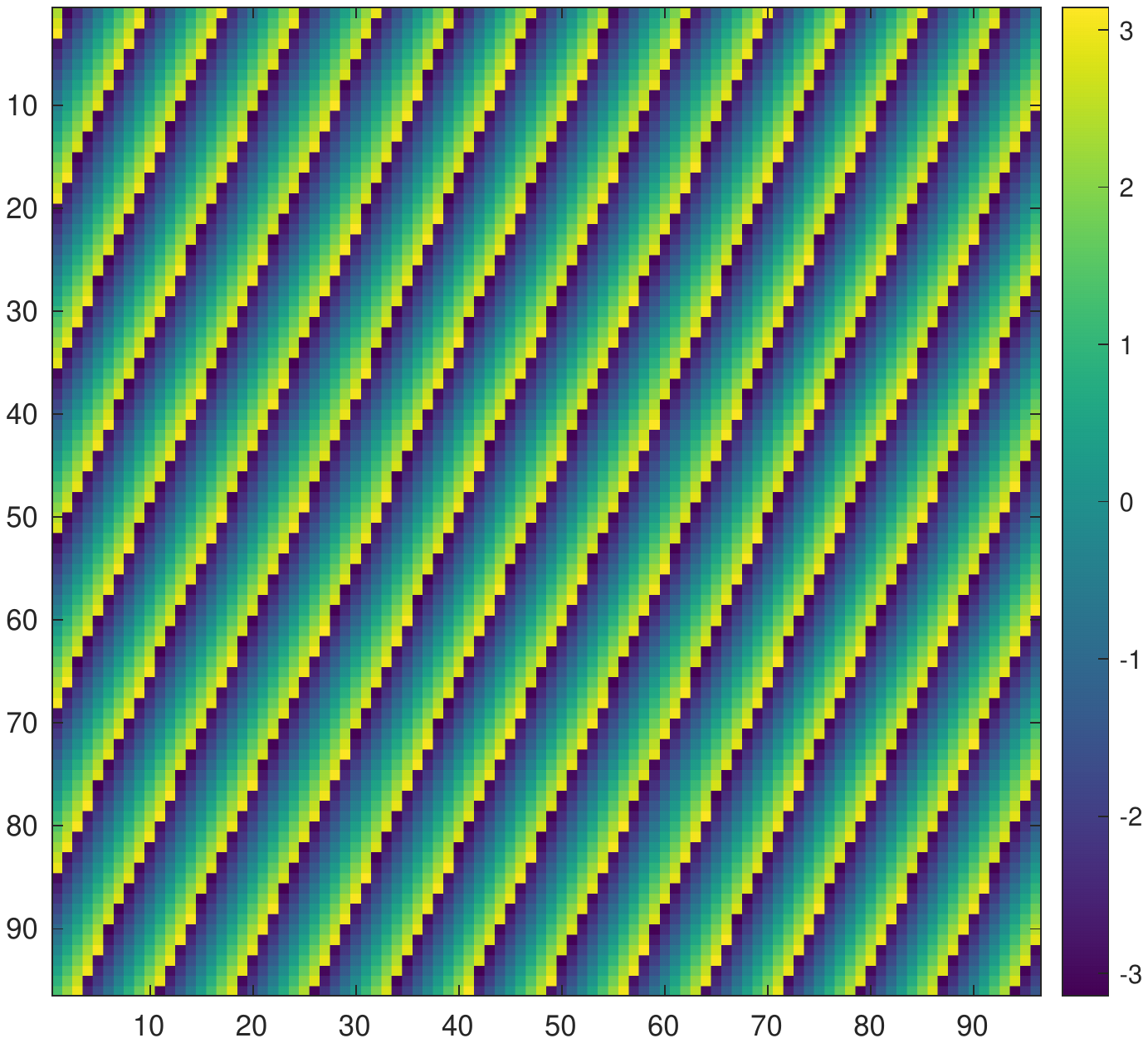} };
\draw ($(ideal.north)+(0,0.1)$) node[align=center,anchor=south] {$\mathrm{Arg}(u^\mathrm{in}_\mathrm{true})$};


\node[inner sep=0pt, anchor = west] (diff1) at ($(-2.5,-9.3)$) {\includegraphics[trim = {3.93 cm} {9.08 cm} {4.88 cm} {8.45 cm}, clip = true, width=0.4\textwidth]{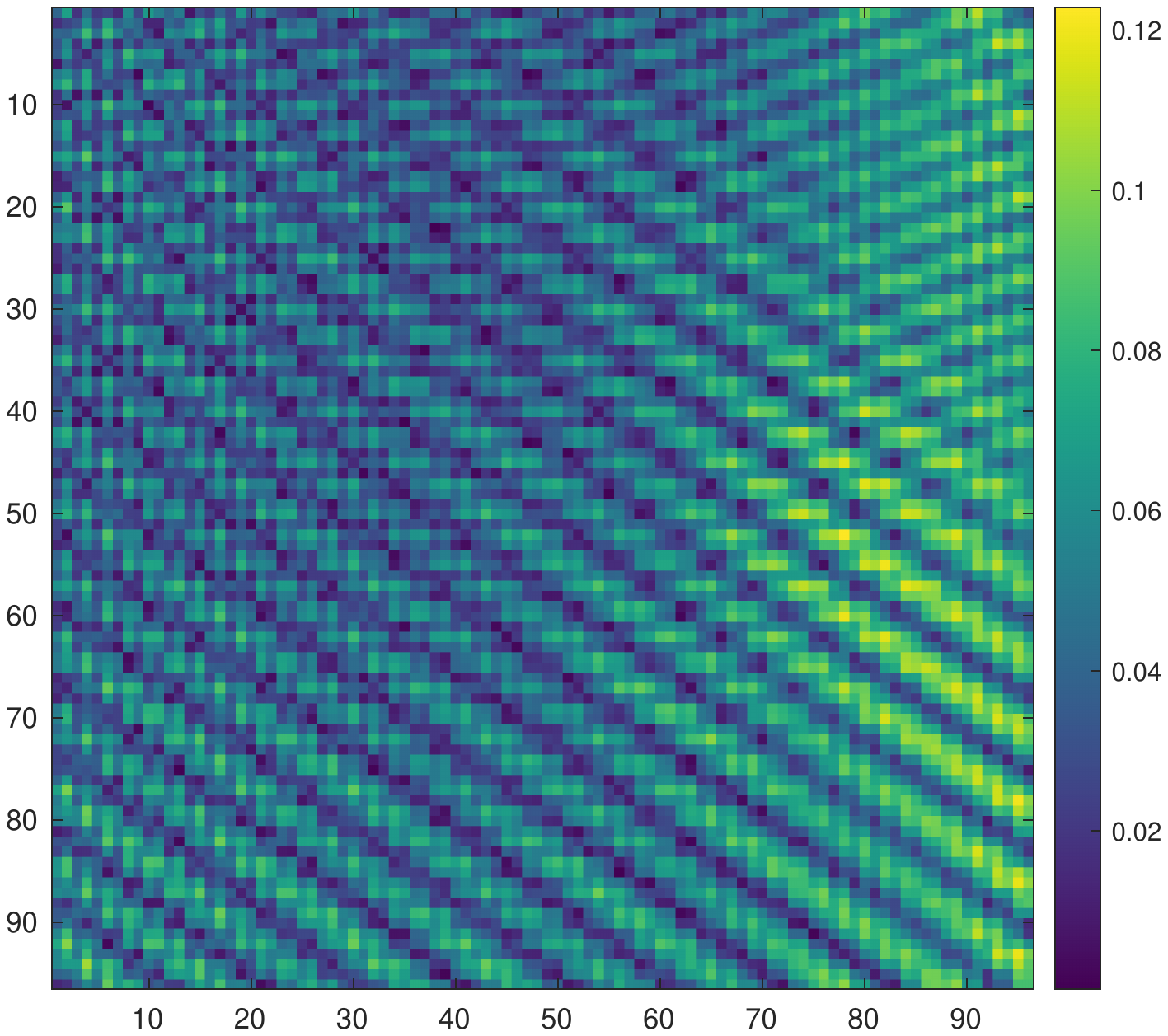} };
\draw ($(diff1.north) + (0,0.2)$) node[anchor=south,align=center] (diff1text) {AS without tilt transfer};

\node[inner sep=0pt, anchor = west] (diff2) at ($(diff1.east) + (1.7,0)$) {\includegraphics[trim = {3.93 cm} {9.08 cm} {4.88 cm} {8.45 cm}, clip = true,
width = 0.4\textwidth]{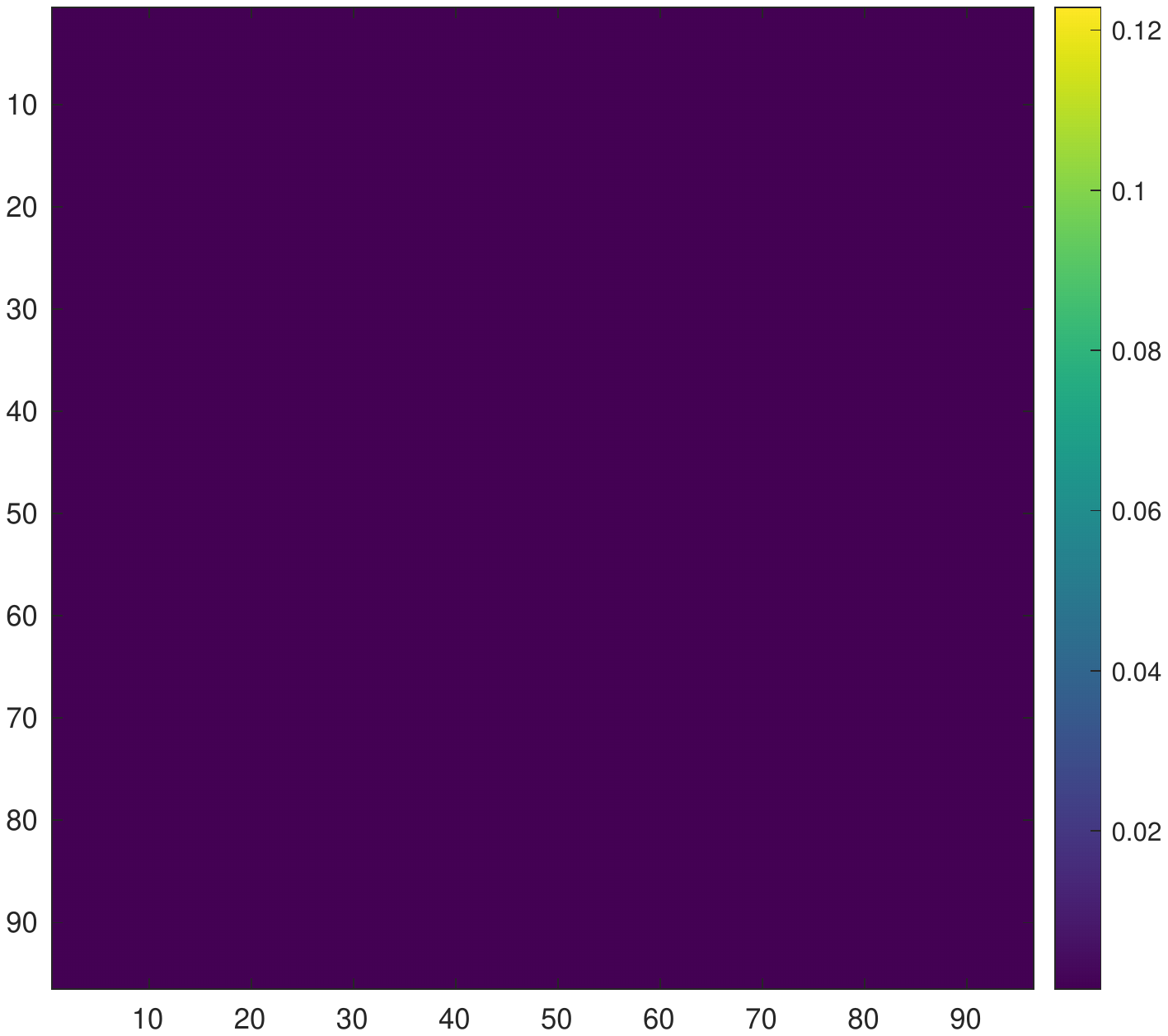} };
\draw ($(diff2.north) + (0,0.2)$) node[anchor=south,align=center] (diff2text) {AS with tilt transfer};
\begin{axis}[at = {($(diff2.north west) + (5,0)$)}, anchor = north west,
    hide axis,
    colormap/viridis,
    colorbar,
    point meta min=0,
    point meta max=0.1228,
    colorbar style={
        height=7.2cm,
        width=0.25cm,
        ytick={0,0.04,...,0.14},
        tick label style={/pgf/number format/fixed},
        /pgf/number format/precision=3
    }]
    \addplot [draw=none] coordinates {(0,0)};
\end{axis}
\end{tikzpicture}
\end{adjustbox}
\caption{Propagation of the incident field.
Top: Scheme of the numerical experiment (left) and phase of the expected propagated field (right).
Bottom: Error map  $| u^\mathrm{in}_\mathrm{true} - u^\mathrm{in}_\mathrm{prop}|$.
\label{fig:3dinc}}
\end{figure}

\subsection{Free-Space Propagation and Pupil Function: $\Pd$}\label{sec:defocus}

  The last matrix to describe in~\eqref{eq:lipp2} is $\Pd$. It models the lowpass filtering behavior of the microscope and can also be used to perform a free-space propagation of the field. For instance, this is required for the acquisition setup described in Section~\ref{sec:ReconsReal}. Hence, $\mathbf{P}$ corresponds to the discrete convolution operator associated to the continuously defined kernel  $p \in L_2(\R^2)$ that depends on the point-spread function (PSF) of the system as well as the considered propagation kernel. Although the output of $\Gtildd$ (scattered field on $\Gamma$) is not compactly supported, it enjoys fast decay, which allows us to apply $\Pd$  via a FFT with suitable padding.
  

\subsection{Computation of the 3D Incident Field: $\uind$}\label{sec:uin3D}

The evaluation of the forward model~\eqref{eq:lipp1} and~\eqref{eq:lipp2} at a given point $\fd \in \R^\Nf$ requires the knowledge of the 3D incident field $\uind \in \C^N$. Here, we propose to build this volume through the free-space propagation of the 2D measurement $\yd^\mathrm{in}\in \C^M$ of this field at the detector plane~$\Gamma$. This is possible as the area of~$\Gamma$ is assumed to be larger than that of a face of the volume~$\Omega$ since $L \leq \tilde{L}$.

Let us denote by $y^\mathrm{in} : \Gamma \rightarrow \C$ the continuous version of $\yd^\mathrm{in}$ to simplify the presentation. Then, we get from the angular spectrum method~\cite{goodman2005introduction} that, $\forall \xd = (x_1,x_2,x_3) \in \Omega$,
\begin{equation}\label{eq:propUin}
    u^\mathrm{in}(\xd) =( p_{x_3} \ast y^\mathrm{in})(x_1,x_2).
\end{equation}
There,  $p_{x_3}$ is the propagation kernel that is defined in the Fourier domain by 
\begin{equation}
    \hat{p}_z(\omegad)= \exp \left(-\ii (x_{\Gamma}-z)\sqrt{ \kb -(\omega_1^2+\omega_2^2)} \right),
\end{equation}
where $x_{\Gamma}$ denotes the position of the measurement plane $\Gamma$.

Because both the propagation kernel and the measured incident field are not compactly supported, a naive computation of the aperiodic convolution in~\eqref{eq:propUin} would introduce significant errors within the estimated volume $u^\mathrm{in}$. The difficulty lies in the way of properly extending the measured field $y^\mathrm{in}$ outside $\Gamma$ to ensure that the result of the convolution inside $\Omega$ is valid. For instance, a zero padding or a simple periodization are not satisfactory as they would introduce large discontinuities in the amplitude and/or the phase of $y^\mathrm{in}$.

Instead, let us inject in~\eqref{eq:propUin} the expression of $y^\mathrm{in}(\xd) = a(\xd) \exp(\ii \xd^T\tilde{\kd}^\mathrm{in})$, where $a : \Gamma \rightarrow \C$ is the complex amplitude of the field and $\tilde{\kd}^\mathrm{in} = (k^\mathrm{in}_{1},k^\mathrm{in}_{2})$ corresponds to the restriction of the wave vector $\kd^\mathrm{in} \in \R^3$ to its first two components, leading to
\begin{align}
    u^\mathrm{in}(\xd) &= \left(p_{x_3} \ast a(\cdot) \mathrm{e}^{\ii (\cdot)^T  \tilde{\kd}^\mathrm{in}}\right)(\tilde{\xd}) \notag \\
    &= \frac{1}{(2\pi)^2}\int_{\R^2} \widehat{p_{x_3}}(\omegad)\hat{a}(\omegad-\tilde{\kd}^\mathrm{in})\mathrm{e}^{\ii \omegad^T \tilde{\xd}}\dd\omegad \notag \\
    &=\frac{\mathrm{e}^{\ii \tilde{\xd}^T \tilde{\kd}^\mathrm{in}}}{(2\pi)^2}\int_{\R^d} \widehat{p_{x_3}}(\omegad + \tilde{\kd}^\mathrm{in})
    \hat{a}(\omegad)\mathrm{e}^{\ii \omegad^T \tilde{\xd}}\dd\omegad \notag \\
    &=   \mathrm{e}^{\ii \tilde{\xd}^T \tilde{\kd}^\mathrm{in} } \left( a \ast p_{x_3}(\cdot)\mathrm{e}^{-\ii (\cdot)^T  \tilde{\kd}^\mathrm{in}}   \right)(\tilde{\xd}),\label{eq:propUin-2}
\end{align}
with $\tilde{\xd}=(x_1,x_2)$ and $\omegad = (\omega_1,\omega_2) \in \R^2$.
Hence,~\eqref{eq:propUin} can be equivalently expressed as a 2D aperiodic convolution of the complex amplitude $a$ with the kernel $p_{x_3}(\cdot)\mathrm{e}^{-\ii (\cdot)^T \tilde{\kd}^\mathrm{in}}$, followed by a modulation in the space domain.
This approach is called tilt transfer because the shift of $y^\mathrm{in}$ in the Fourier domain is transferred to the propagation kernel~\cite{ritter2014modified,guo2014diffraction}.
The advantage of this formulation is that, by contrast to $y^\mathrm{in}$, the complex amplitude $a$ is not far from a constant signal, up to some noise and optical aberrations. Hence,we compute~\eqref{eq:propUin-2} using a  periodic convolution with minor discretization artifacts. 

The advantage of this approach is illustrated in Figure~\ref{fig:3dinc} where we propagate a slice of an ideal tilted plane wave $y^\mathrm{in}$ using the angular spectrum method with and without tilt transfer. The difference between the expected incident field~$u^\mathrm{in}_\mathrm{true}$ and the propagated field $u^\mathrm{in}_\mathrm{prop}$ is depicted in the bottom panel. Clearly, the tilt transfer allows one to significantly reduce the discretization errors and attenuate the aliasing artifacts.

\section{Reconstruction Framework}\label{sec:probform}
\subsection{Problem Formulation}

We adopt a standard variational formulation to recover the scattering potential $\fd$ from the $Q$ scattered fields $\{\yd^\mathrm{sc}_q\}_{q=1}^Q$ that are recorded when the sample is impinged with the incident fields $\{\uind_q\}_{q=1}^Q$. Specifically, the reconstructed $\fd^\ast$ is specified~as
\begin{equation}\label{eq:prob}
     \fd^\ast  \in   \bigg\lbrace \argmin{\fd \in \R^N}{  \bigg( \sum_{q=1}^Q \frac{1}{2\|\yd_q^\mathrm{sc}\|^2}\|\mathbf{H}_q(\fd) - \yd_q^\mathrm{sc} \|^2 + \tau\mathcal{R}(\fd) + i_{\geq 0}(\fd) \bigg)} \bigg\rbrace.
\end{equation}
In~\eqref{eq:prob}, $\mathbf{H}_q : \R^N \rightarrow \C^M$ denotes the forward model described by~\eqref{eq:lipp1} and \eqref{eq:lipp2} for the $q$th incident wave $\uind_q$,  $\mathcal{R}:\R^N \mapsto \R_{\geq 0}$ is a regularization functional, and  $\tau >0$ balances between data fidelity and regularization. The term $i_{\geq 0}(\fd) = \{ 0, \fd \in (\R_{\geq 0})^\Nf; \, +\infty, \text{ otherwise}\}$ is a nonnegativity  constraint that is suitable for our applications. For other applications that involve inverse scattering, this term is modified to constrain the scattering potential to a given range of values. Such priors have  been shown to significantly improve the quality of the reconstruction~\cite{lim2015comparative, sung2009optical}.  
Finally, we consider as regularizer $\mathcal{R}$ either the total-variation seminorm~\cite{Rudin1992} or the Hessian-Schatten norm~\cite{Lefkimmiatis2013}.

\subsection{Optimization} \label{sec:opti}

Following~\cite{liu2017seagle,soubies2017efficient,ma2018accelerated}, we deploy an accelerated forward-backward splitting (FBS) algorithm~\cite{beck2009fast,Nesterov2013} to solve the optimization problem~\eqref{eq:prob}. The iterates are summarized in Algorithm~\ref{Algo:FISTA}, with some further details below.

\begin{algorithm}[t]
		\caption{Accelerated FBS~\cite{beck2009fast,Nesterov2013} for solving \eqref{eq:prob} }\label{Algo:FISTA}
	\begin{algorithmic}[1]
		\REQUIRE  $\fd^0 \in \R^N$, $(\gamma_k >0)_{k \in \N \setminus \{0\}} $
		\STATE $\mathbf{v}^1=\fd^0$
		\STATE $\alpha_1=1$
		\STATE $k=1$
		\WHILE{(not converged)}
		    \STATE Select a subset $\mathcal{Q} \subset [1 \ldots Q]$ \label{alg:stochastic}
			\STATE$\disp\mathbf{d}^{k}= \sum_{q \in \mathcal{Q}} \frac{1}{\|\yd_q^\mathrm{sc}\|^2}\Re\left( \mathbf{J}_{\mathbf{H}_q}^\ast(\fd^k)  \left(\mathbf{H}_q(\fd^k) - \yd_q^\mathrm{sc} \right) \right)$ \label{alg:grad}
			\STATE$\fd^{k}=\mathrm{prox}_{\gamma_k \tau \mathcal{R} + i_{\geq 0}}\left( \mathbf{v}^k - \gamma_k \mathbf{d}^k \right)$ \label{alg:prox}
			\STATE $\displaystyle \alpha_{k+1} \leftarrow \frac{1 + \sqrt{1+4\alpha_k^2}}{2}$
			\STATE $\displaystyle  \mathbf{v}^{k+1} = \fd^k + \left(\frac{\alpha_k-1}{\alpha_{k+1} }\right)(\fd^k - \fd^{k-1}) $
			\STATE $k \leftarrow k+1$
 		\ENDWHILE
 	\end{algorithmic}
	\end{algorithm}

\begin{itemize}
    \item As in~\cite{soubies2017efficient}, we implemented a stochastic-gradient version of the algorithm by selecting a subset of of the measurements $\{\yd^\mathrm{sc}_q\}_{q=1}^Q$ at each iteration  (Line~\ref{alg:stochastic}). This allows us to reduce the computational burden of the method. 
    \item Line~\ref{alg:grad}  corresponds to the evaluation of the gradient of $\frac{1}{2\|\yd_q^\mathrm{sc}\|^2} \sum_{q \in \mathcal{Q}} \|\mathbf{H}_q(\cdot) - \yd^\mathrm{sc}_q\|^2$. An explicit expression of the Jacobian matrix  $\mathbf{J}_{\mathbf{H}_q}(\fd^k)$ of  $\mathbf{H}_q$ can be found in~\cite{soubies2017efficient,ma2018accelerated}. Similarly to the forward model~\eqref{eq:lipp1}, the application of this Jacobian matrix to a given vector of $\C^M$ requires the inversion of $(\Id - \diag{\fd} \Gtildd^\ast)$. Again, this inversion is performed using a conjugate-gradient-based algorithm.
    \item   For both the TV and Hessian-Schatten-norm regularizers, no known closed-form expression exists for the proximity operator of $\gamma_k \tau \mathcal{R} + i_{\geq 0}$ (Line~\ref{alg:prox}). However, there exist efficient algorithms to evaluate them. Specifically, we use the fast gradient-projection method for TV~\cite{beck2009fast2} and its extension to the Hessian-Schatten-norm regularizer~\cite{Lefkimmiatis2013}.
    \item We set the sequence of step sizes to $\gamma_k = \gamma_0 /\sqrt{k}$ for $\gamma_0>0$. This is standard and ensures the convergence of incremental proximal-gradient methods~\cite{bertsekas2011incremental}.
\end{itemize}

The whole reconstruction pipeline is implemented within the framework of the GlobalBioIm library\footnote{http://bigwww.epfl.ch/algorithms/globalbioim/}~\cite{soubies2018pocket} and will be made available online.

\section{Numerical Results}\label{sec:results}

 In this section , we present two types of experiments. First we validate our computational pipeline on simulated data. Then, we deploy the proposed approach on some real data. For both cases, we provide comparison with existing algorithms.
\subsection{Simulated Data}

\subsubsection{Simulation Setting}

We simulated red blood cells (RBCs) with a maximal RI of~$1.05$~(see Figure~\ref{fig:sim_slice} top row).
This sample is immersed in air~($\nb = 1$) and is illuminated by tilted plane waves with wavelength~$\lambda = 600$\nano\meter.
To simulate the ODT measurements, we used the discrete dipole approximation model on a grid with a resolution of~$75$\nano\meter.
To probe the sample, we generated 40 views within a cone of illumination whose half-angle is~$45^\degree$. This corresponds to severely restricted  angles of view and makes the reconstruction problem very challenging. 
Each view has $512^2$ measurements (resolution of $150$\nano\meter).
Finally, we have simulated, independently for each view, an acquisition of the incident field on $\Gamma$. 

\newcommand{\szsub}{0.145}
\newcommand{\szax}{3*\szsub}
\newcommand{\legjet}{0.9*\szsub}
\pgfplotsset{
        colormap={parula}{
            rgb255=(53,42,135)
            rgb255=(15,92,221)
            rgb255=(18,125,216)
            rgb255=(7,156,207)
            rgb255=(21,177,180)
            rgb255=(89,189,140)
            rgb255=(165,190,107)
            rgb255=(225,185,82)
            rgb255=(252,206,46)
            rgb255=(249,251,14)
        },
    }

\begin{table}[t]
\caption{Relative error of the RBCs reconstructions.}
    \centering
    \begin{tabular}{l|c|c|c}
        \toprule
        \toprule
        Method & Rytov & BPM & \LS{}
        \\
        \midrule
        $\frac{\|\hat{\mathbf{n}} - \mathbf{n}_\mathrm{gt} \|^2}{\| \mathbf{n}_\mathrm{gt} \|^2}$ & $1.8231 \times 10^{-4}$ & $2.4585 \times 10^{-5}$ & $\mathbf{9.0120\times 10^{-6}}$\\
        \bottomrule
        \bottomrule
    \end{tabular}
    \label{tab:relerr}
\end{table}

\renewcommand{\szsub}{0.145}
\renewcommand{\szax}{3*\szsub}
\renewcommand{\legjet}{0.9*\szsub}

\begin{figure}[t]
    \centering
    
    \begin{tikzpicture}
    \begin{axis}[at={(0,0)},anchor = south west,
    ylabel = Ground truth,
    xmin = 0,xmax = 216,ymin = 0,ymax = 72,
    width={\szax*\textwidth}, 
        scale only axis,
        enlargelimits=false,
        axis line style={draw=none},
        tick style={draw=none},
        axis equal image,
        xticklabels={,,},yticklabels={,,},
        ]
    \node[inner sep=0pt, anchor = south west] (gt_xy_1) at (0,0) {\includegraphics[ width=\szsub\textwidth]{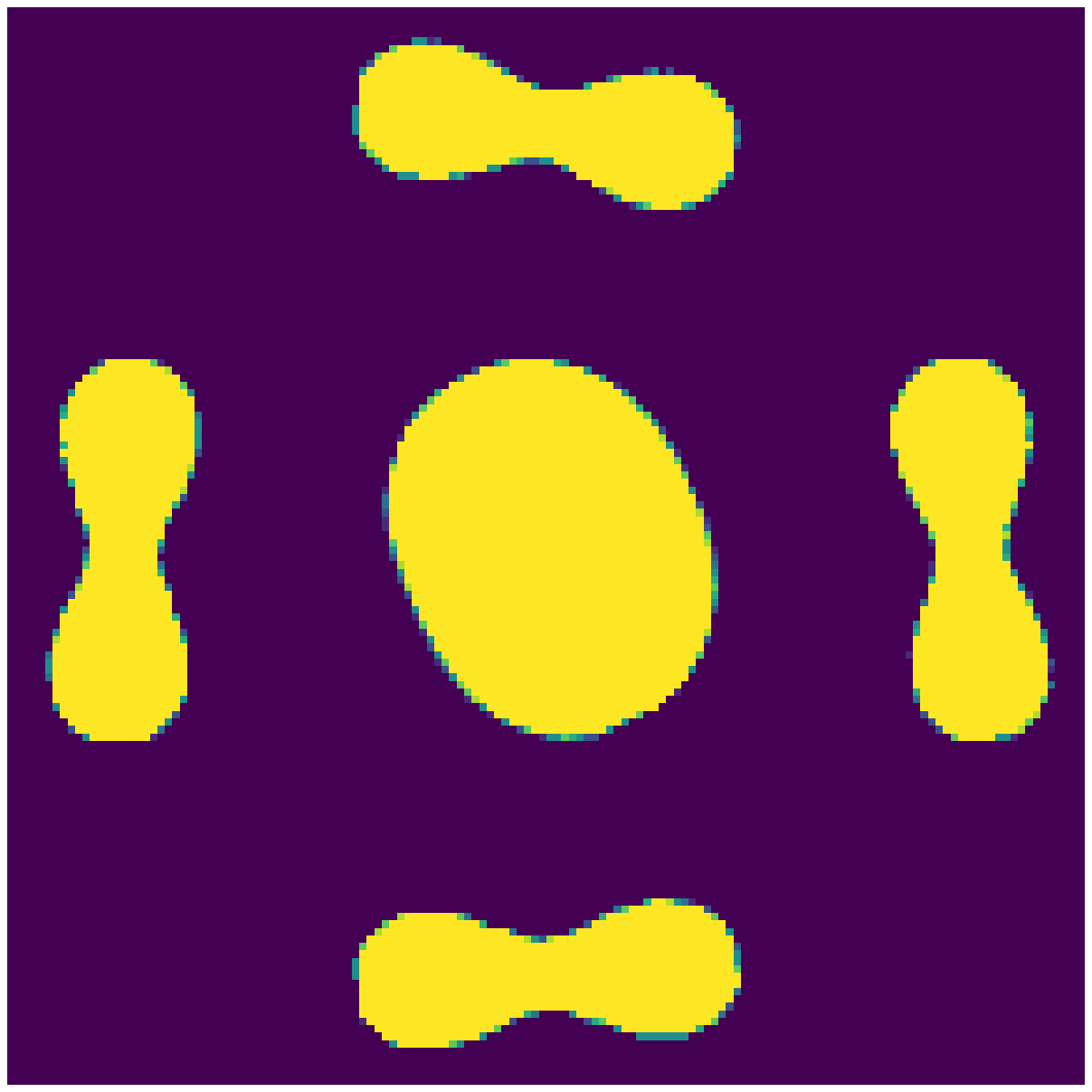}};
    \node[anchor = north east, white] at (gt_xy_1.north east) {XY};
    \draw[white,dotted,thick] (0,72- 30)--(72,72 - 30);
    \draw[white,dotted,thick] (30,0)--(30,72);
    
    \node[inner sep=0pt, anchor = west] (gt_xy_2) at (gt_xy_1.east) {\includegraphics[ width=\szsub\textwidth]{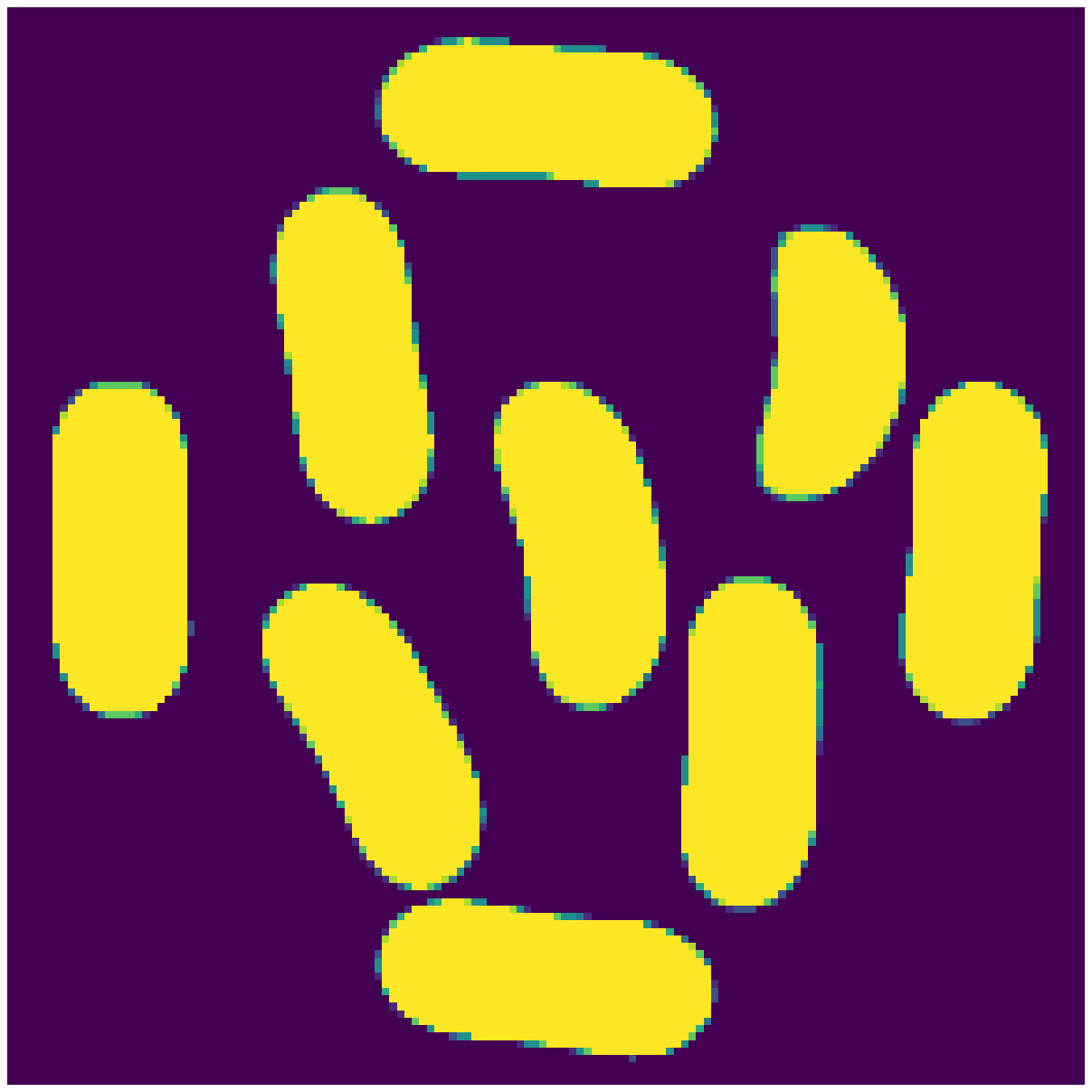}};
    \node[anchor = north east, white] at (gt_xy_2.north east) {XZ};
    \draw[white,dotted,thick] (72 ,72- 30)--(72 + 72,72 - 30);
    \draw[white,dotted,thick] (72 + 36,0)--(72 + 36,72);
    \draw[-{latex[length = 0.05mm]}, thick,white] (72 + 67, 57)-- (72 +  61, 54);
    
    \node[inner sep=0pt, anchor = west] (gt_xy_3) at (gt_xy_2.east) {\scalebox{1}[-1]{\includegraphics[ width=\szsub\textwidth,]{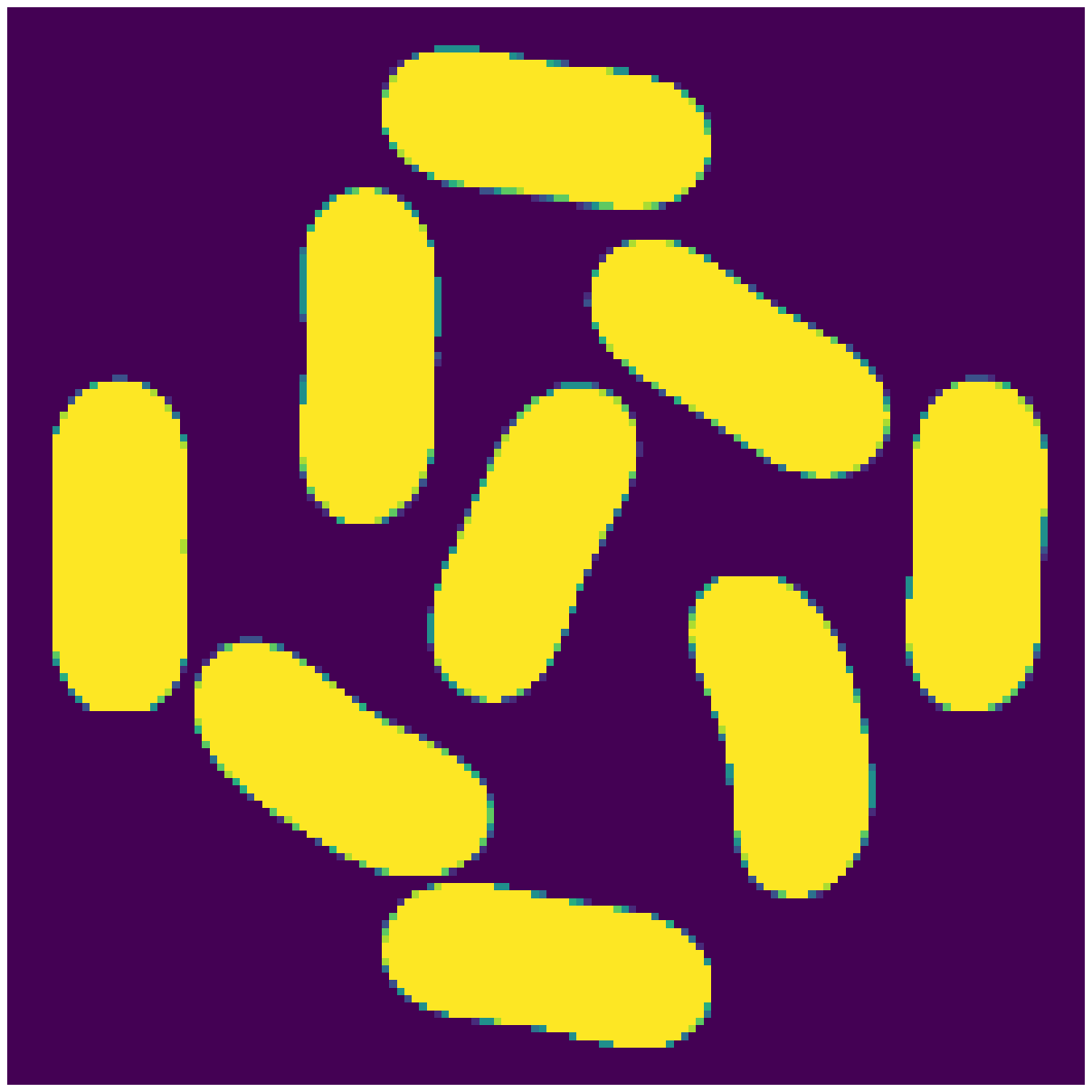}}};
     \node[anchor = north east, white] at (gt_xy_3.north east) {YZ};
     \draw[white,dotted,thick] (72*2 + 72 - 36,0)--(72*2 + 72 - 36,72);
    \draw[white,dotted,thick] (72*2+0,30)--(72*2+72,30);
    \draw[-{latex[length = 0.05mm]}, thick,white] (2*72 + 35, 22)-- (2*72 +  31, 27);
    \end{axis}
    
     \begin{axis}[at={(gt_xy_1.south west)},anchor = north west, ylabel = Rytov,
    xmin = 0,xmax = 216,ymin = 0,ymax = 72, width=\szax\textwidth,
        scale only axis,
        enlargelimits=false,
        axis line style={draw=none},
        tick style={draw=none},
        axis equal image,
        xticklabels={,,},yticklabels={,,},
        ]

    \node[inner sep=0pt, anchor = south west] (ryt_xy_1) at (0,0) {\includegraphics[ width=\szsub\textwidth]{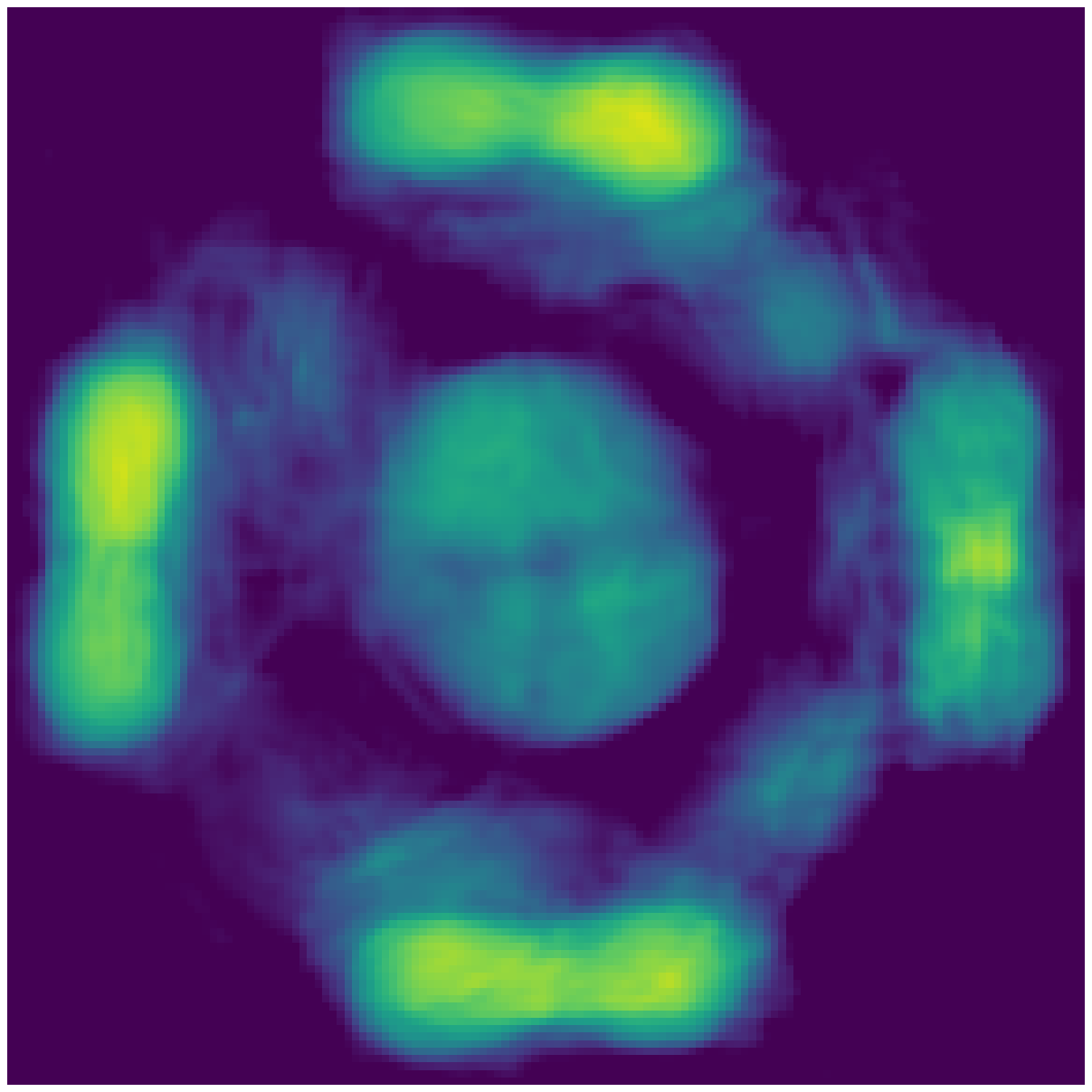}};
    
    \node[inner sep=0pt, anchor = west] (ryt_xy_2) at (ryt_xy_1.east) {\includegraphics[ width=\szsub\textwidth]{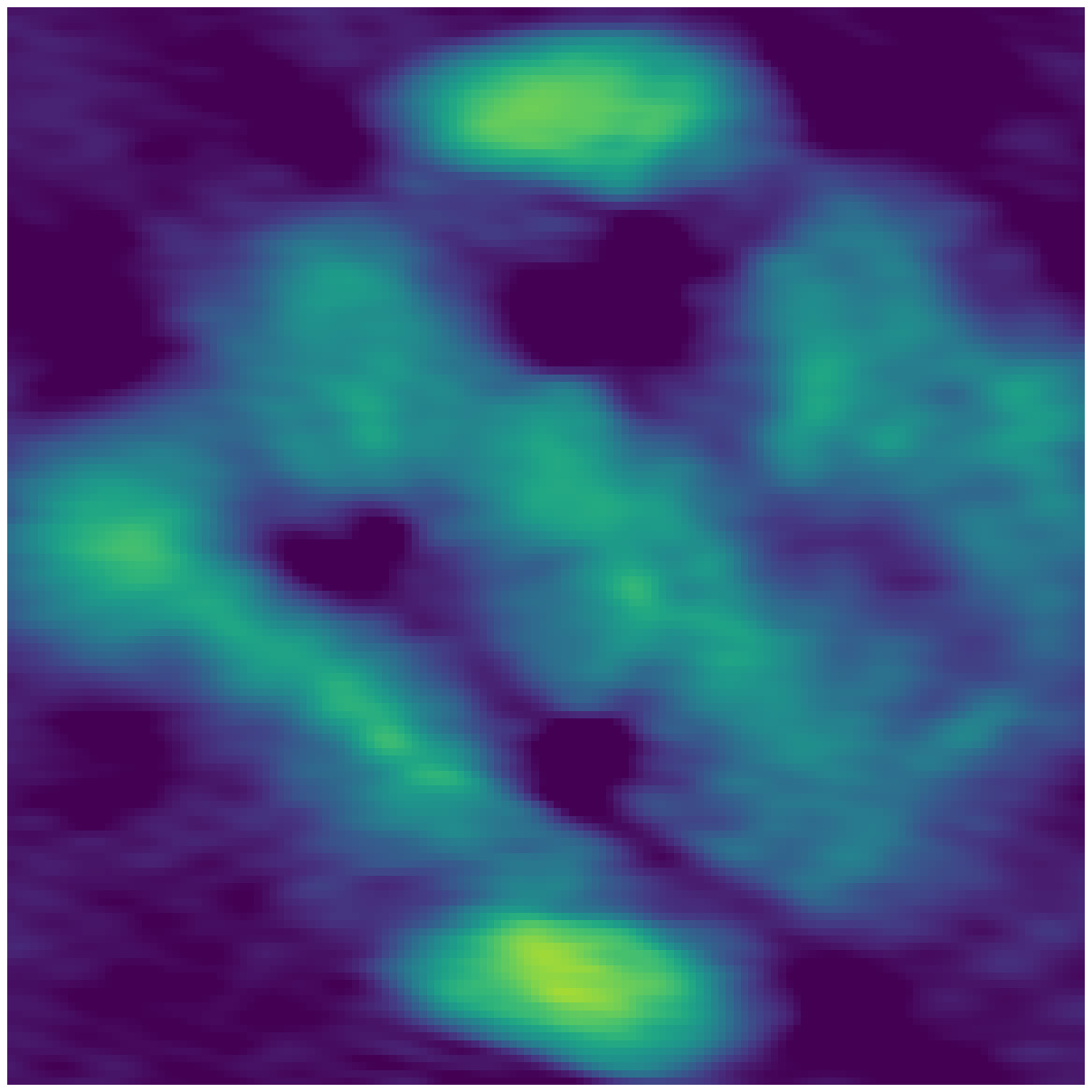}};
    \draw[-{latex[length = 0.05mm]}, thick,white] (72 + 67, 57)-- (72 +  61, 54);

    \node[inner sep=0pt, anchor = west] (ryt_xy_3) at (ryt_xy_2.east) {\scalebox{1}[-1]{\includegraphics[ width=\szsub\textwidth]{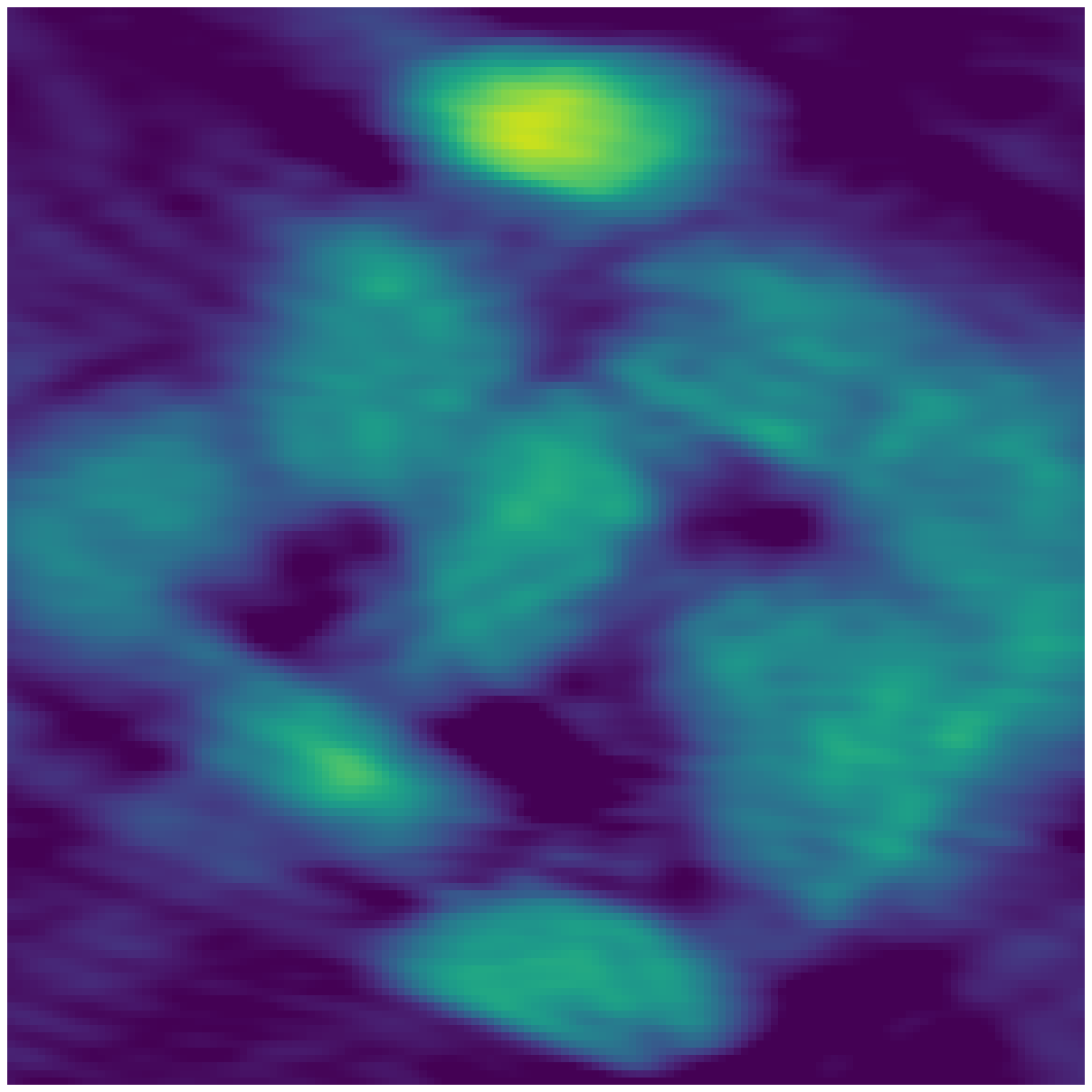}}};
    \draw[-{latex[length = 0.05mm]}, thick,white] (2*72 + 35, 22)-- (2*72 +  31, 27);
    \end{axis}
    
    \begin{axis}[at={(ryt_xy_1.south west)},anchor = north west,ylabel = BPM,
    xmin = 0,xmax = 216,ymin = 0,ymax = 72, width=\szax\textwidth,
        scale only axis,
        enlargelimits=false,
        axis line style={draw=none},
        tick style={draw=none},
        axis equal image,
        xticklabels={,,},yticklabels={,,}
        ]
    
    \node[inner sep=0pt, anchor = south west] (bpm_xy_1) at (0,0) {\includegraphics[ width=\szsub\textwidth]{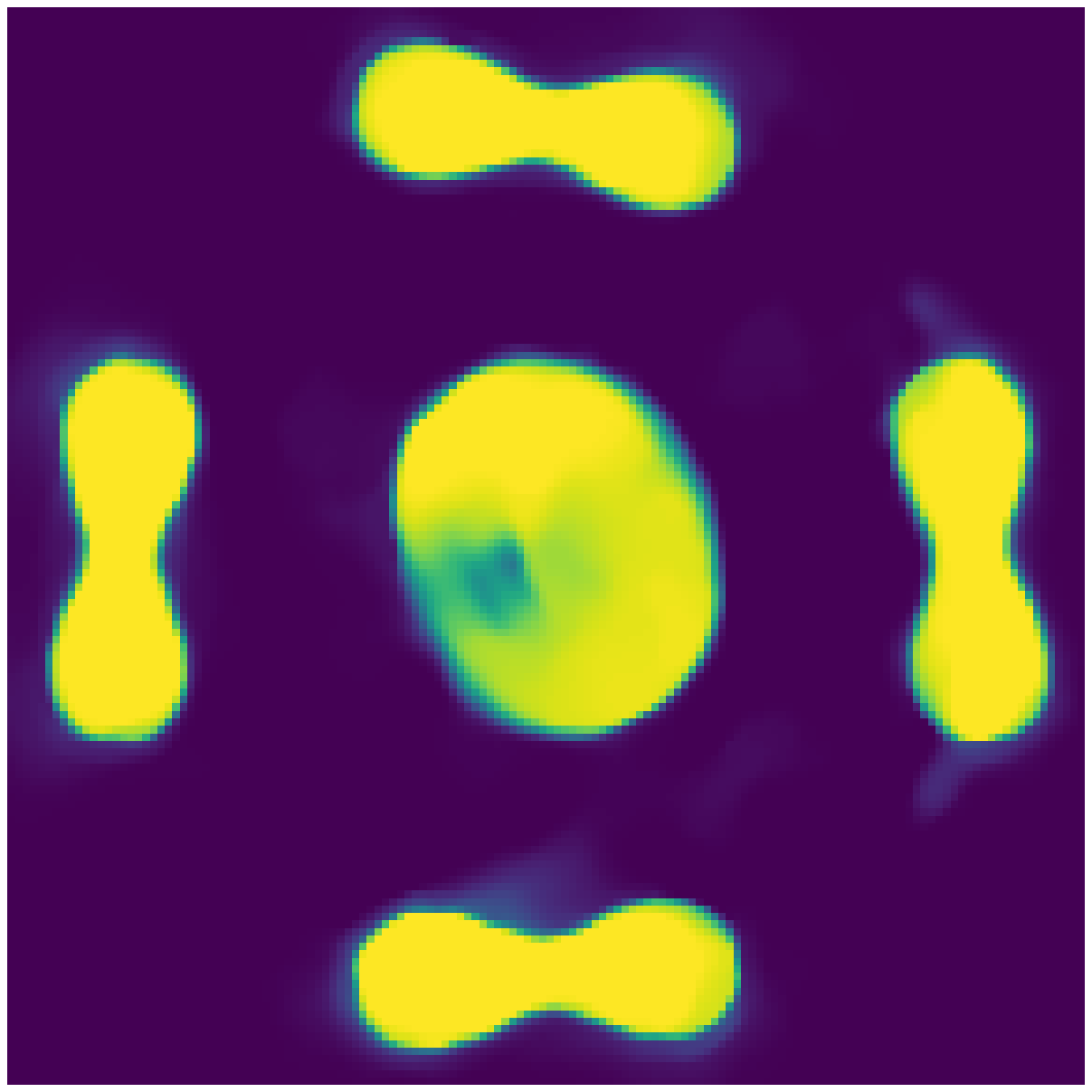}};
    
    \node[inner sep=0pt, anchor = west] (bpm_xy_2) at (bpm_xy_1.east) {\includegraphics[ width=\szsub\textwidth]{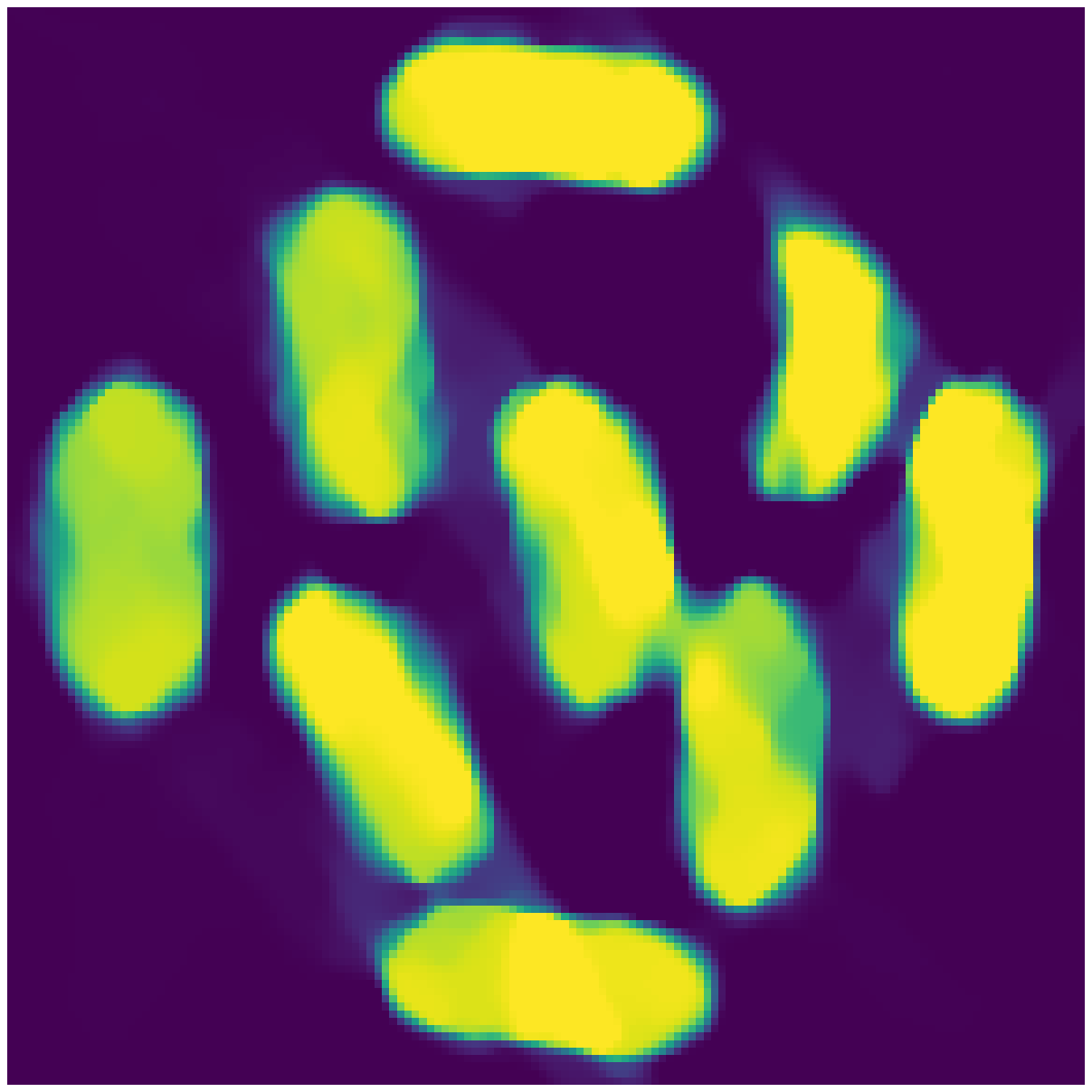}};
    \draw[-{latex[length = 0.05mm]}, thick,white] (72 + 67, 57)-- (72 +  61, 54);
    
    \node[inner sep=0pt, anchor = west] (bpm_xy_3) at (bpm_xy_2.east) {\scalebox{1}[-1]{\includegraphics[ width=\szsub\textwidth]{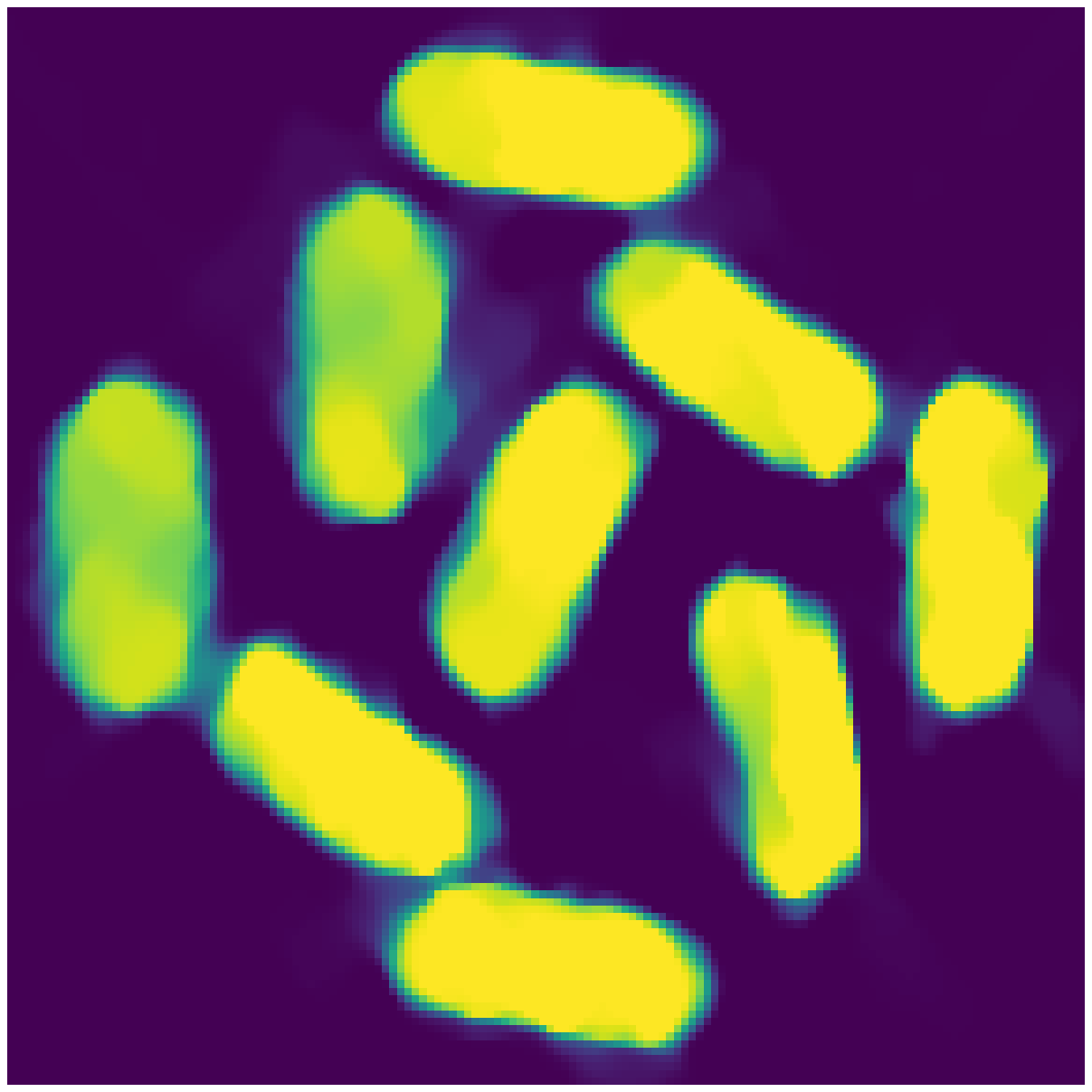}}};
    \draw[-{latex[length = 0.05mm]}, thick,white] (2*72 + 35, 22)-- (2*72 +  31, 27);
    \end{axis}
    
     \begin{axis}[at={(bpm_xy_1.south west)},anchor = north west, ylabel = LS model,
    xmin = 0,xmax = 216,ymin = 0,ymax = 72, width=\szax\textwidth,
        scale only axis,
        enlargelimits=false,
        axis line style={draw=none},
        tick style={draw=none},
        axis equal image,
        xticklabels={,,},yticklabels={,,},
        colormap/viridis,
        colorbar,
        colorbar style={
        at = {(0.007,-0.07)},
            anchor = west,
            xmin = 0, xmax = 1.05,
            point meta min = 1,
            point meta max = 1.05,
            rotate=-90,
            scale only axis,
            enlargelimits = false,
            scaled x ticks = false,
            scaled y ticks = false,
            height = 0.43\textwidth,
            samples = 200,
            width = 0.3cm,
            tick label style = {font=\normalsize, color=black,anchor = west,},
            ytick = {1,1.05}, 
            yticklabels={,},
            ylabel={1 -- 1.05},ylabel style={rotate=-90,yshift=0.4cm,xshift=-0.5cm},
        }
        ]
    
    \node[inner sep=0pt, anchor = south west] (lipp_xy_1) at (0,0) {\includegraphics[ width=\szsub\textwidth]{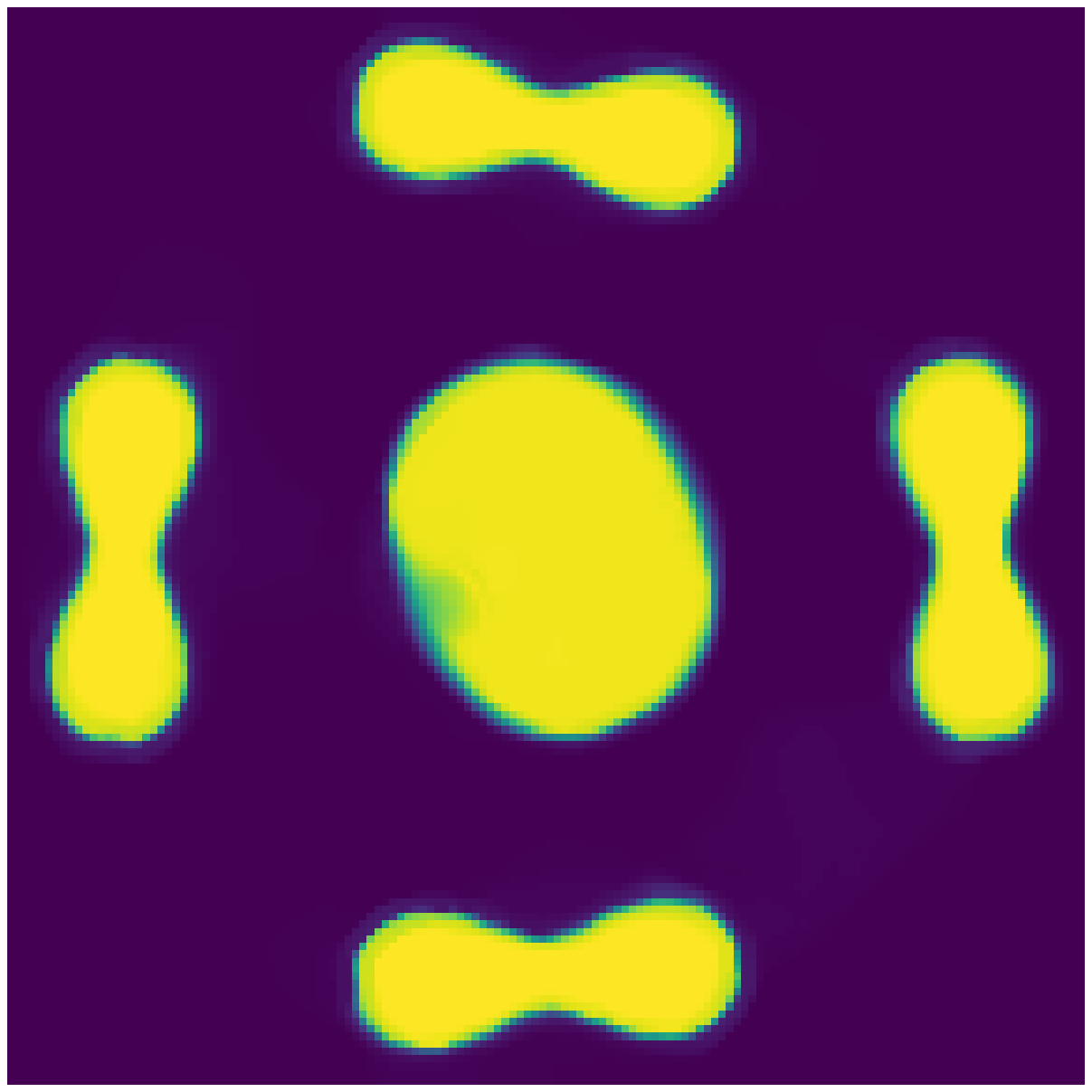}};
    
    \node[inner sep=0pt, anchor = west] (lipp_xy_2) at (lipp_xy_1.east) {\includegraphics[ width=\szsub\textwidth]{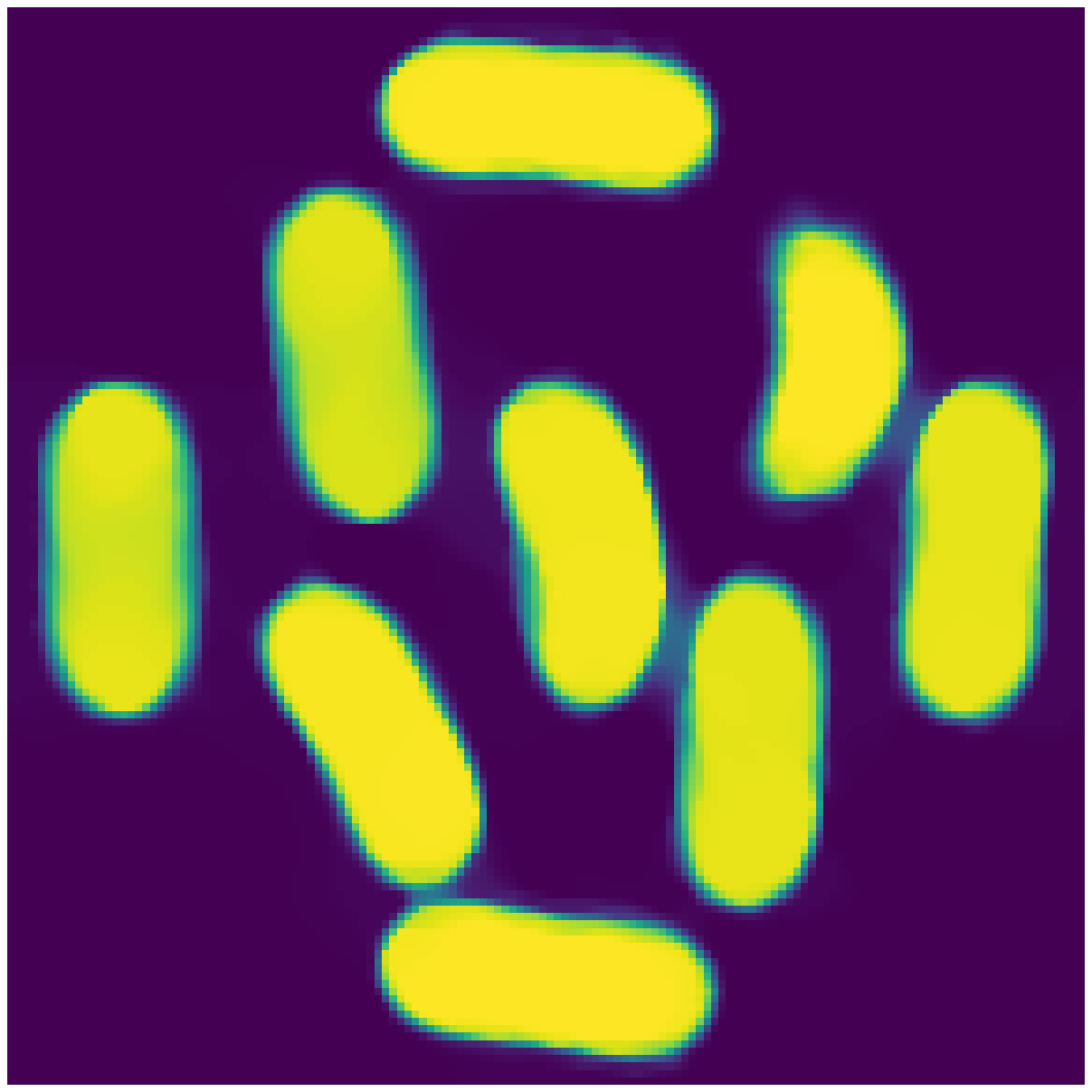}};
    \draw[-{latex[length = 0.05mm]}, thick,white] (72 + 67, 57)-- (72 +  61, 54);
    
    \node[inner sep=0pt, anchor = west] (lipp_xy_3) at (lipp_xy_2.east) {\scalebox{1}[-1]{\includegraphics[ width=\szsub\textwidth]{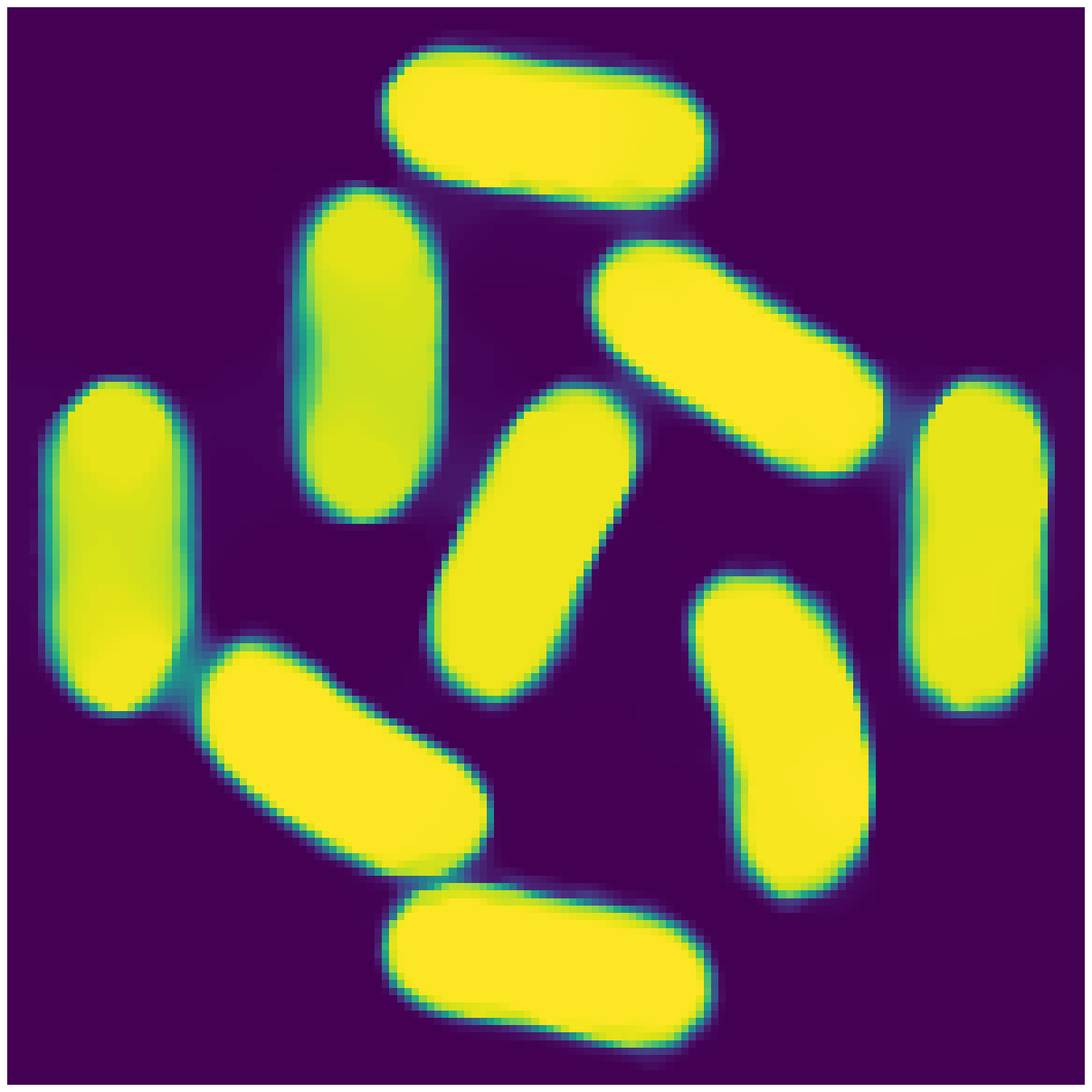}}};
    \draw[-{latex[length = 0.05mm]}, thick,white] (2*72 + 35, 22)-- (2*72 +  31, 27);
    \end{axis}
    \end{tikzpicture}
    
    \caption{Reconstructions of the simulated RBCs by Rytov, BPM, and the proposed method (LS model). 
    }
    \label{fig:sim_slice}
\end{figure}

\renewcommand{\szsub}{0.2}
\renewcommand{\szax}{4*\szsub}
\begin{figure*}[t]
    \centering
    
    \begin{tikzpicture}
     \begin{axis}[at={(0,0)},anchor = south west, ylabel = Rytov, xmin = 0,xmax = 288,ymin = 0,ymax = 72, width=\szax\textwidth,
        scale only axis,
        enlargelimits=false,
        axis equal image,
        xticklabels={,,}, yticklabels={,,},
        colormap/viridis,
        point meta min = 1.338,
        point meta max = 1.48,
        colorbar,
        colorbar style={
            at={(parent axis.east)},
            anchor = west,
            ymin = 1.338, ymax = 1.48,
            point meta min = 1.338,
            point meta max = 1.48,
            enlargelimits = false,
            scaled y ticks = false,
            height = 0.85*\pgfkeysvalueof{/pgfplots/parent axis height},
            width = 0.25cm,
            tick label style = {font=\normalsize, color=black},
            ytick = {1.338,1.48},%
            yticklabels={1.338,1.48},
        }
        ]
    \node[inner sep=0pt, anchor = south west] (ryt_xz) at (0,0) {\includegraphics[ width=\szsub\textwidth]{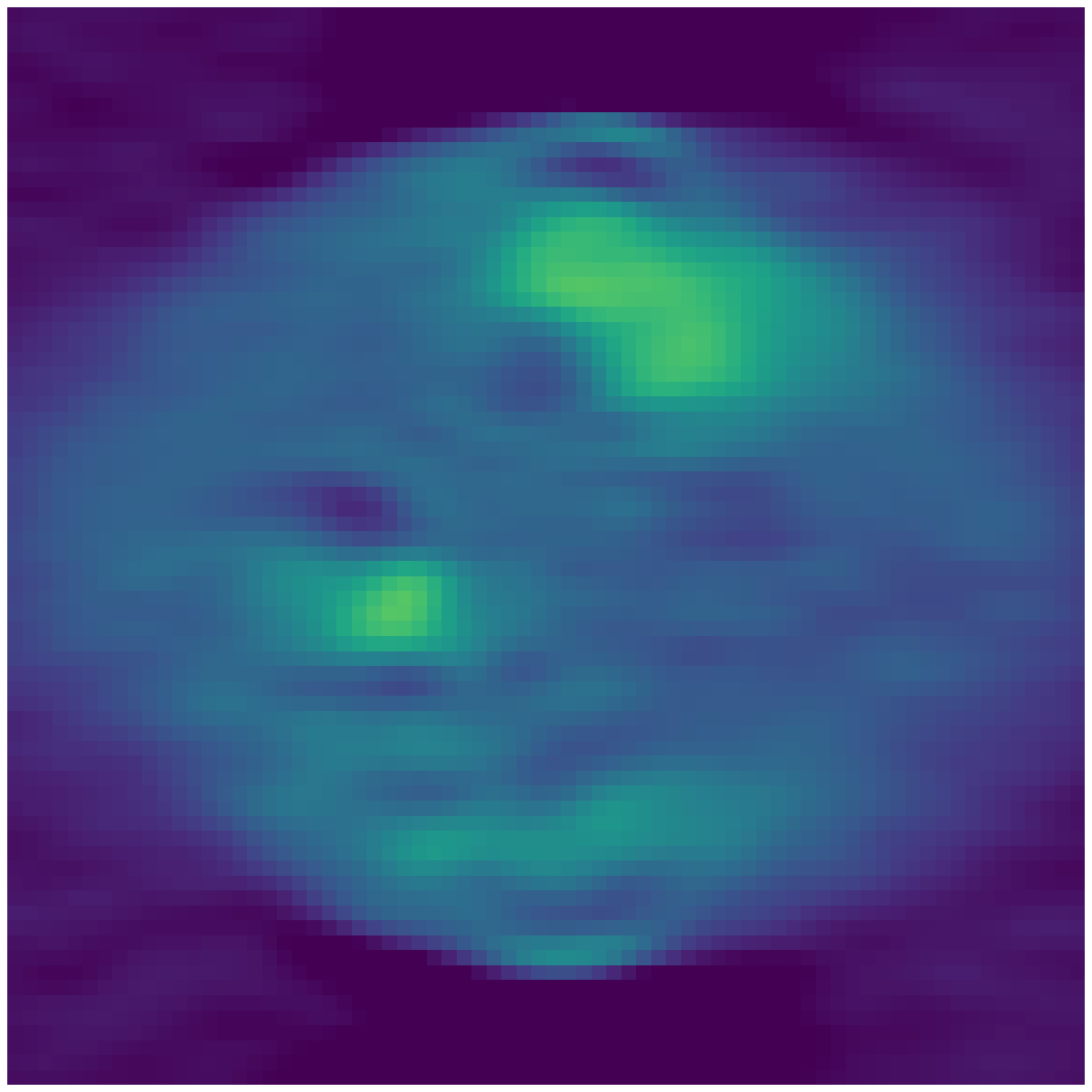}};
    \node[anchor = north east, white] at (ryt_xz.north east) {XZ};
    
    \draw[-,dotted,white,thick] (25,0)--(25,66.5) node [anchor = south,align=center,inner sep = 0] {$z_1$};
    \draw[-,dotted,white,thick] (31,0)--(31,66.5) node [anchor = south,align=center,inner sep = 0] {$z_2$};
    \draw[-,dotted,white,thick] (36,0)--(36,66.5) node [anchor = south,align=center, inner sep = 0] {$z_3$};
    
    \node[inner sep=0pt, anchor = west] (ryt_xy_1) at (ryt_xz.east) {\includegraphics[ width=\szsub\textwidth]{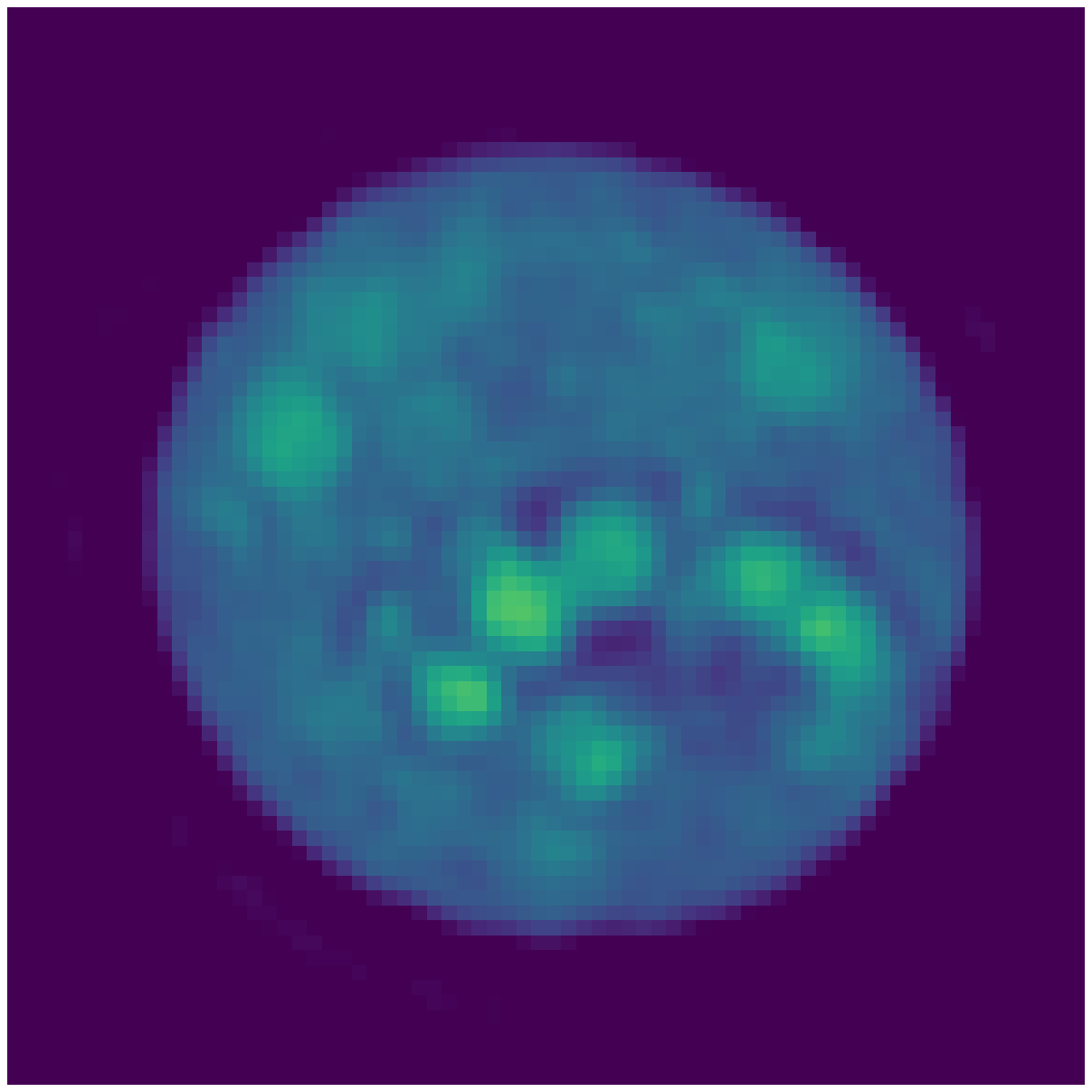}};
   \draw[-{latex[length = 0.05mm]}, thick,red] (72 + 10, 48)-- (72 +  16, 46);
    \node[anchor = north east, white] at (ryt_xy_1.north east) {$z_1 = -1.092$\micro{m}};
    
    \node[inner sep=0pt, anchor = west] (ryt_xy_2) at (ryt_xy_1.east) {\includegraphics[ width=\szsub\textwidth]{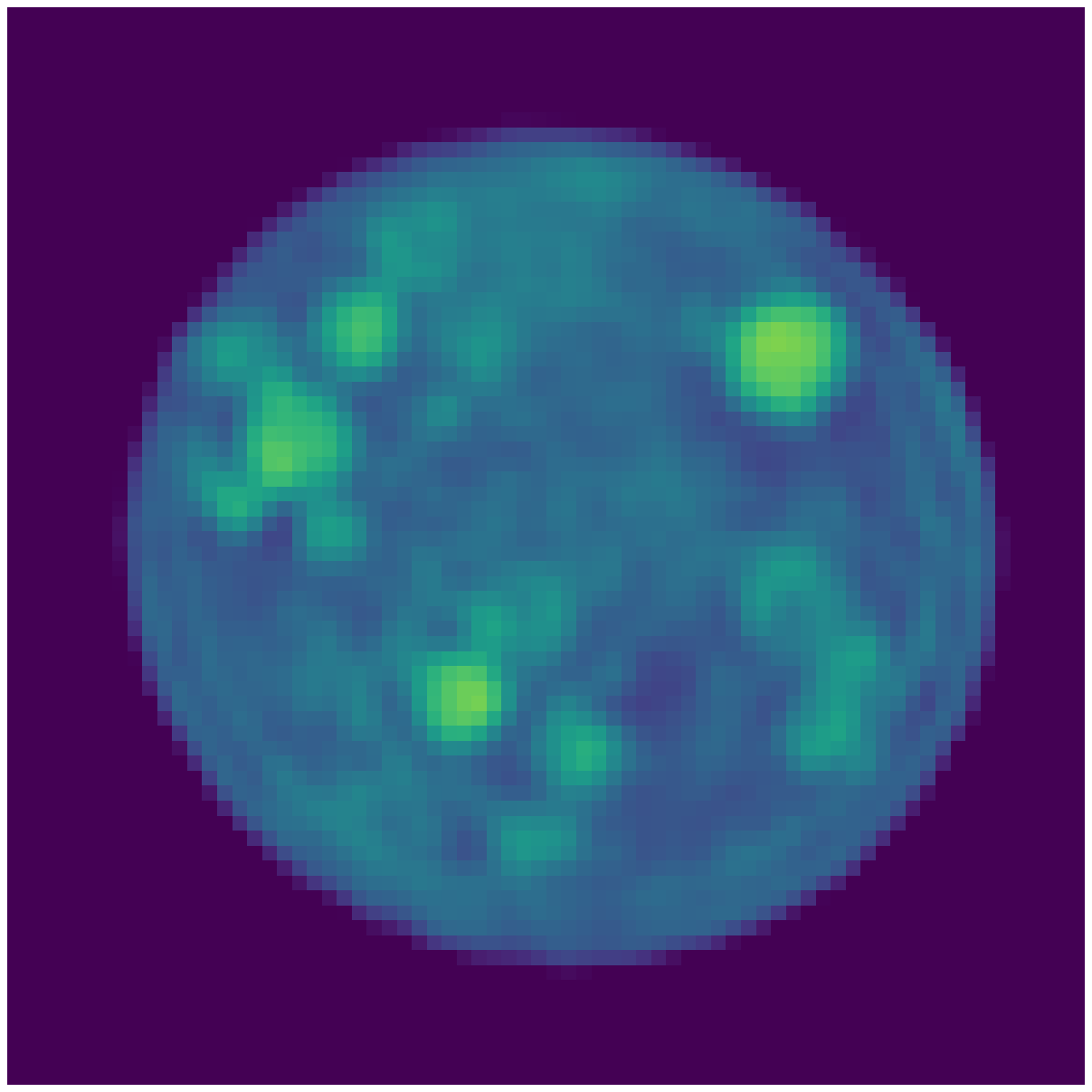}};
    \draw[-{latex[length = 0.05mm]}, thick,red] (2*72 + 6, 48)-- (2*72 +  12, 46);
    \node[anchor = north east, white] at (ryt_xy_2.north east) {$z_2 = -0.496$\micro{m}};
    
    \node[inner sep=0pt, anchor = west] (ryt_xy_3) at (ryt_xy_2.east) {\includegraphics[ width=\szsub\textwidth]{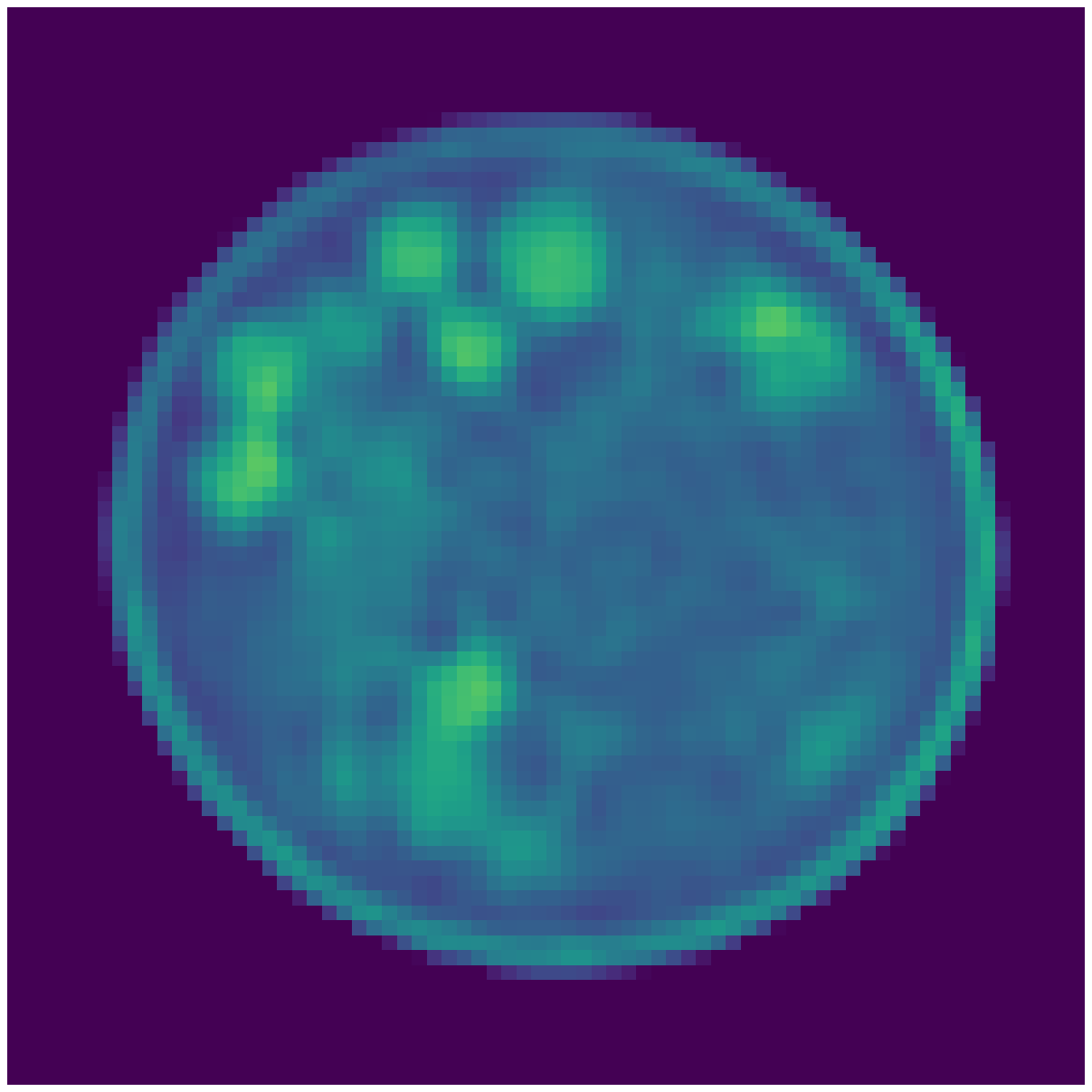}};
    \node[anchor = north east, white] at (ryt_xy_3.north east) {$z_3 = 0$\micro{m}};
    \end{axis}
   
     \begin{axis}[at={(ryt_xz.south west)},anchor = north west,ylabel = BPM,
    xmin = 0,xmax = 288,ymin = 0,ymax = 72, width=\szax\textwidth,
        scale only axis,
        enlargelimits=false,
        axis equal image,
        xticklabels={,,},yticklabels={,,}
        ]
    \node[inner sep=0pt, anchor = south west] (bpm_xz) at (0,0) {\includegraphics[ width=\szsub\textwidth]{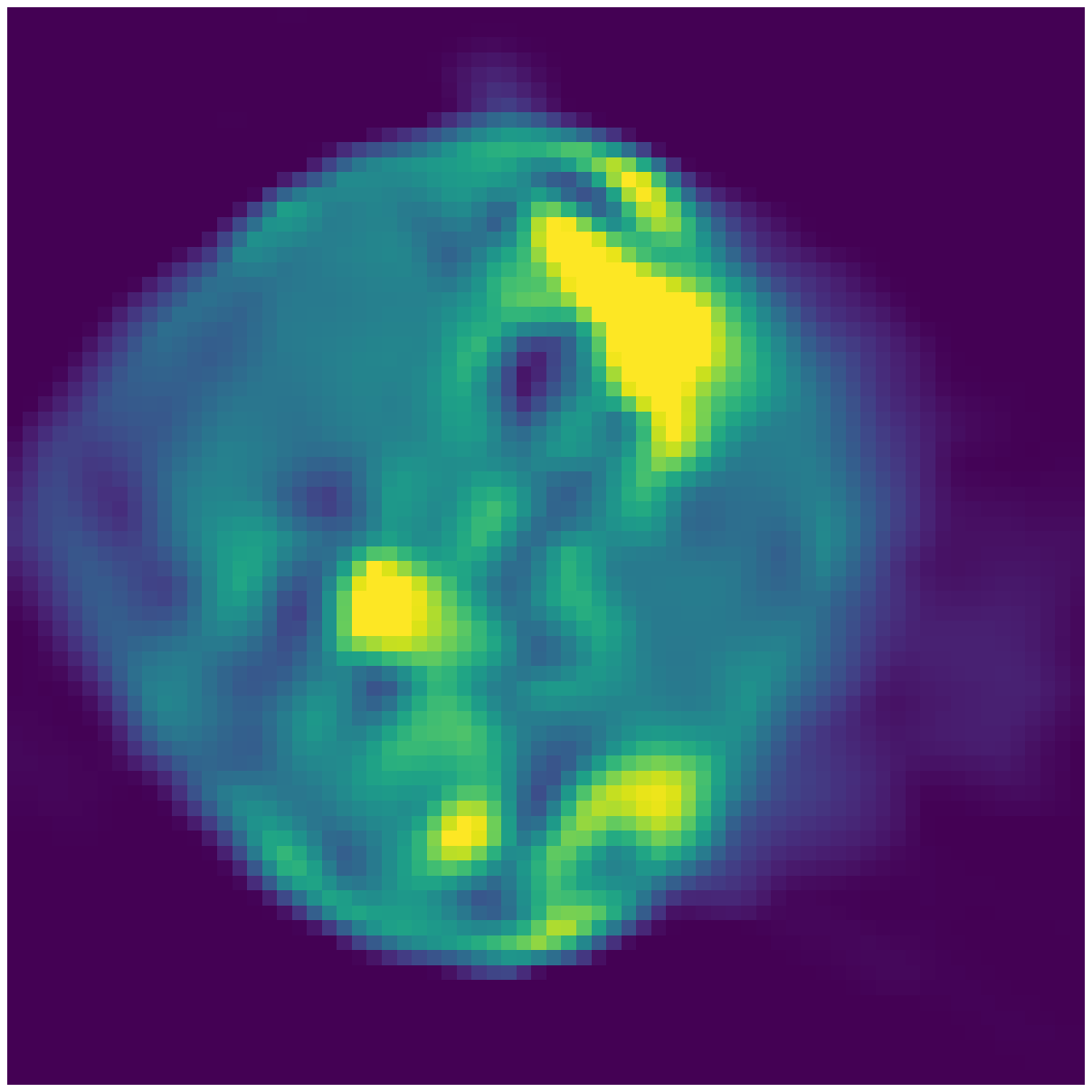}};
    
    \node[inner sep=0pt, anchor = west] (bpm_xy_1) at (bpm_xz.east) {\includegraphics[ width=\szsub\textwidth]{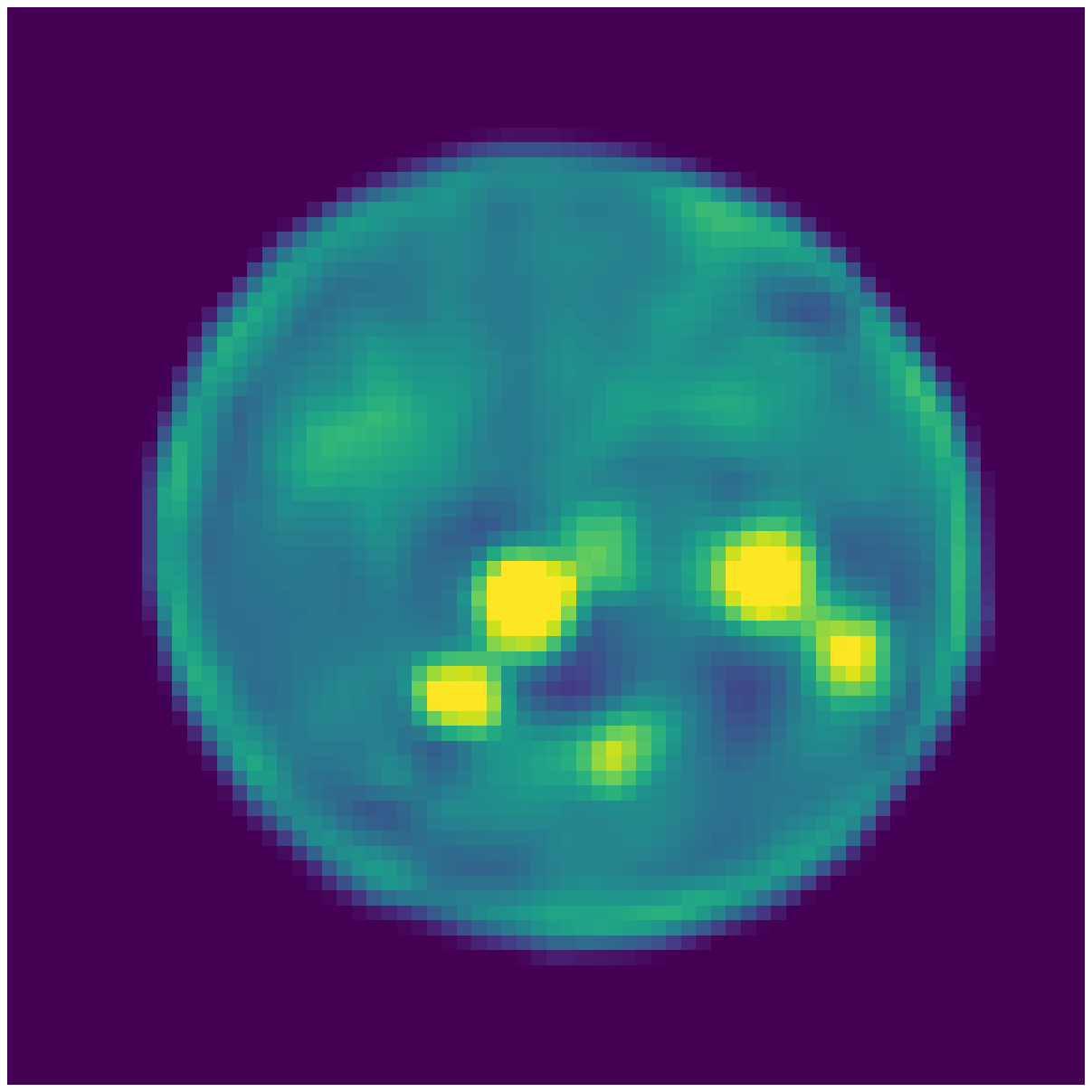}};
    \draw[-{latex[length = 0.05mm]}, thick,red] (72 + 10, 48)-- (72 +  16, 46);
    
    \node[inner sep=0pt, anchor = west] (bpm_xy_2) at (bpm_xy_1.east) {\includegraphics[ width=\szsub\textwidth]{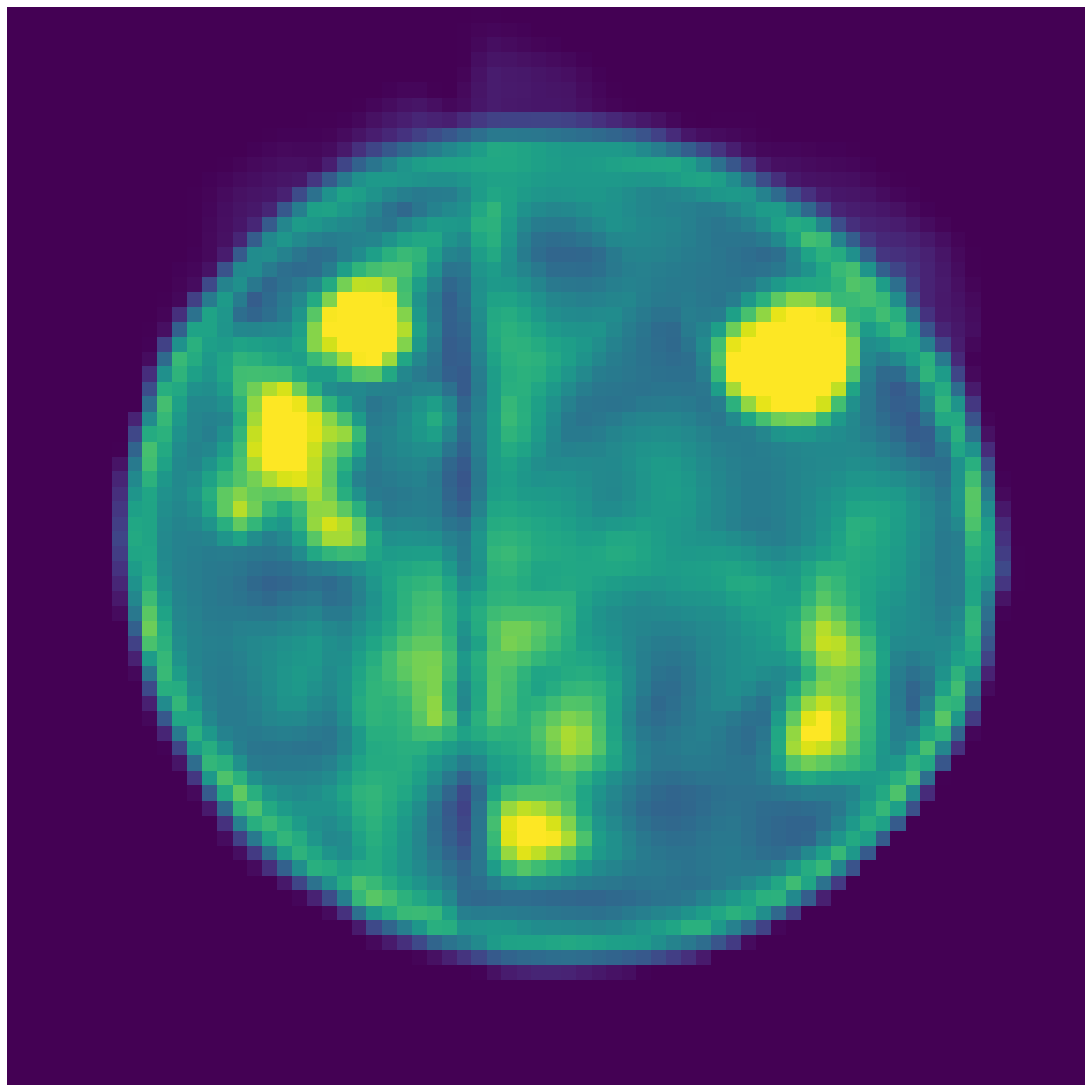}};
    \draw[-{stealth[length = 0.05mm]}, ultra thick, white] (2*72 + 41,36) -- (2*72 +  36, 40);
    \draw[-{latex[length = 0.05mm]}, thick,red] (2*72 + 6, 48)-- (2*72 +  12, 46);
    
    \node[inner sep=0pt, anchor = west] (bpm_xy_3) at (bpm_xy_2.east) {\includegraphics[ width=\szsub\textwidth]{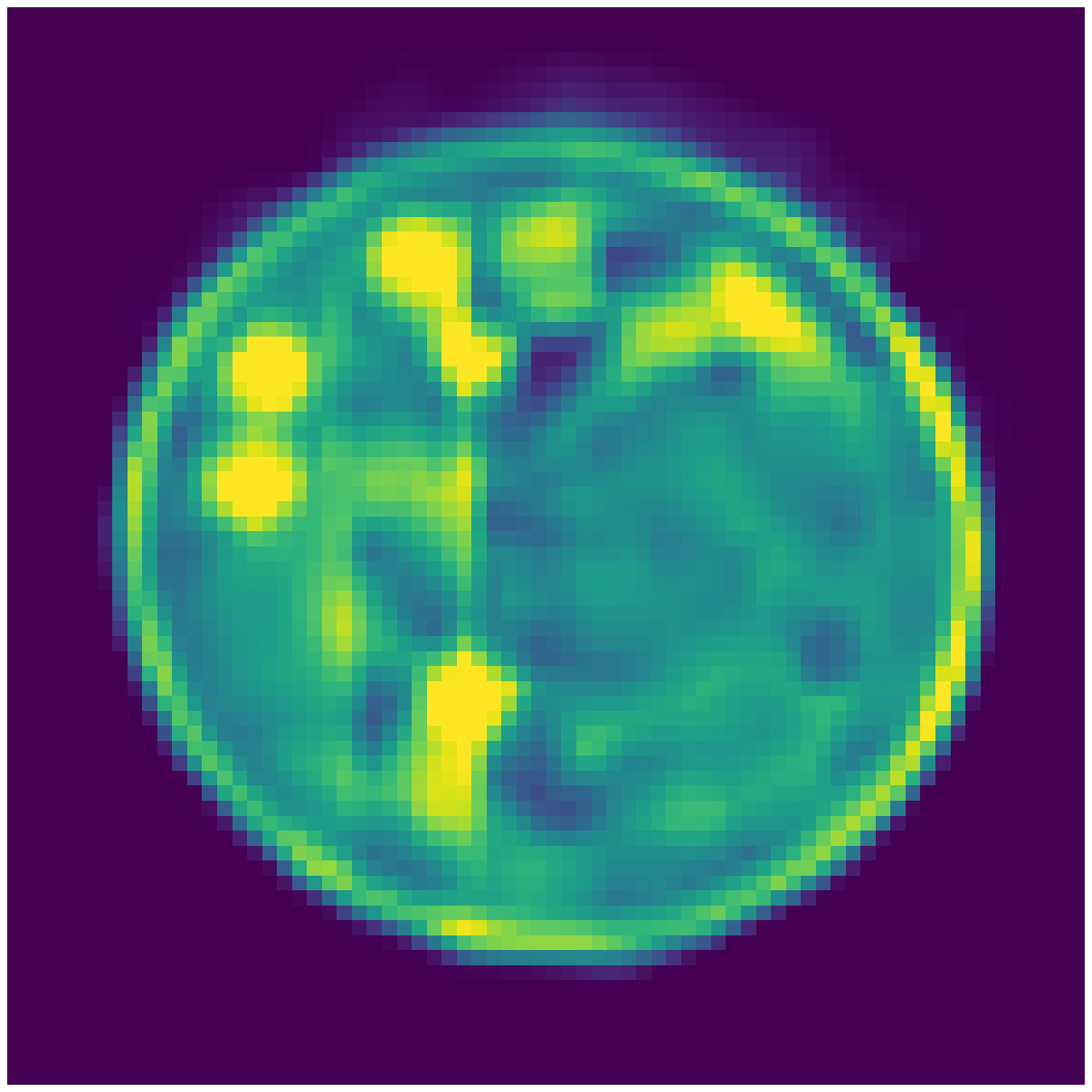}};
    \draw[-{stealth[length = 0.05mm]}, ultra thick, white] (3*72 + 41,36) -- (3*72 +  36, 40);
    \end{axis}
    
     \begin{axis}[at={(bpm_xz.south west)},anchor = north west, ylabel = LS model,
    xmin = 0,xmax = 288,ymin = 0,ymax = 72, width=\szax\textwidth,
        scale only axis,
        enlargelimits=false,
        axis equal image,
        xticklabels={,,},yticklabels={,,},
        ]
    \node[inner sep=0pt, anchor = south west] (lipp_xz) at (0,0) {\includegraphics[ width=\szsub\textwidth]{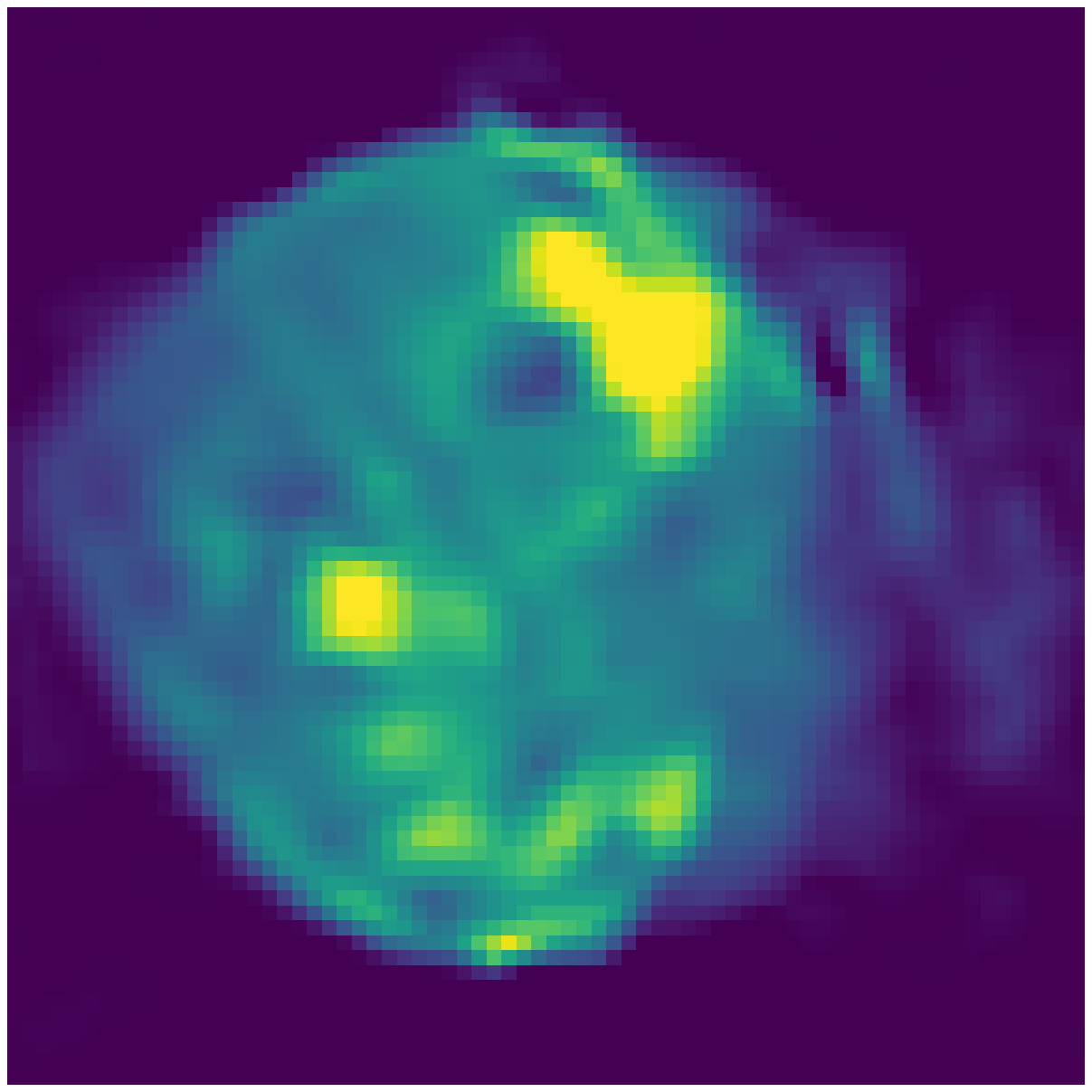}};
    
    \node[inner sep=0pt, anchor = west] (lipp_xy_1) at (lipp_xz.east) {\includegraphics[ width=\szsub\textwidth]{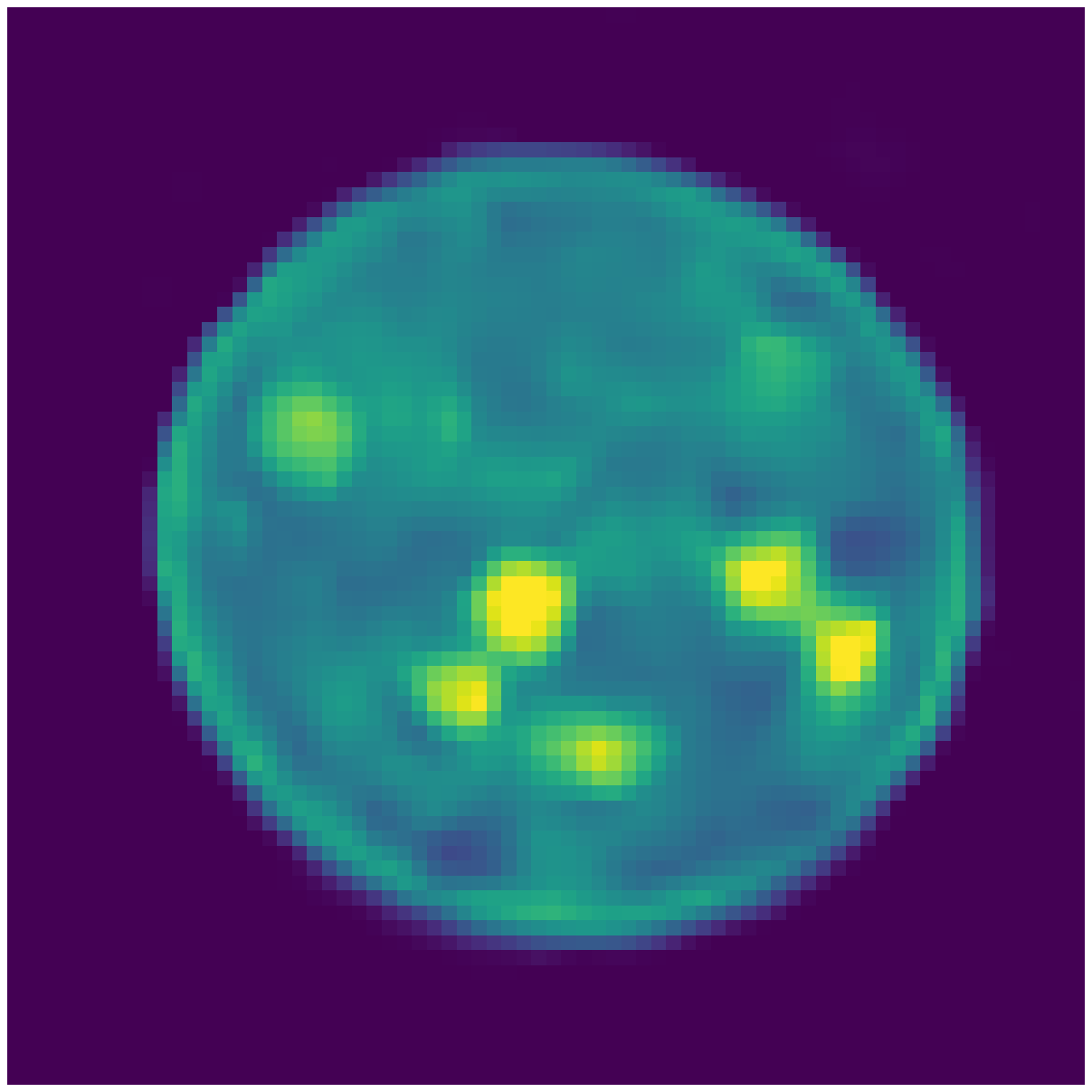}};
    \draw[-{latex[length = 0.05mm]}, thick,red] (72 + 10, 48)-- (72 +  16, 46);
    
    \node[inner sep=0pt, anchor = west] (lipp_xy_2) at (lipp_xy_1.east) {\includegraphics[ width=\szsub\textwidth]{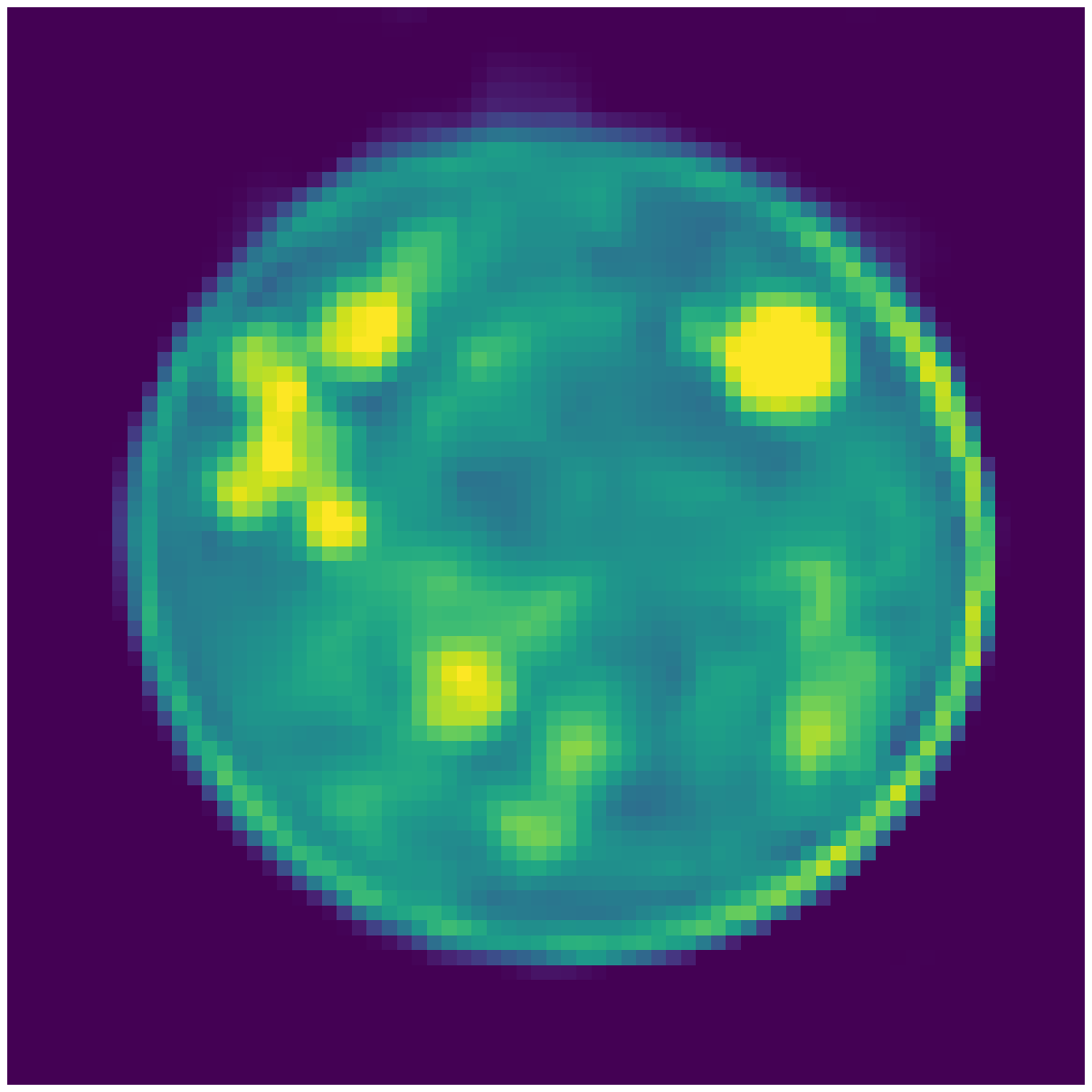}};
    \draw[-{stealth[length = 0.05mm]}, ultra thick, white] (2*72 + 41,36) -- (2*72 +  36, 40);
    \draw[-{latex[length = 0.05mm]}, thick,red] (2*72 + 6, 48)-- (2*72 +  12, 46);
    
    \node[inner sep=0pt, anchor = west] (lipp_xy_3) at (lipp_xy_2.east) {\includegraphics[ width=\szsub\textwidth]{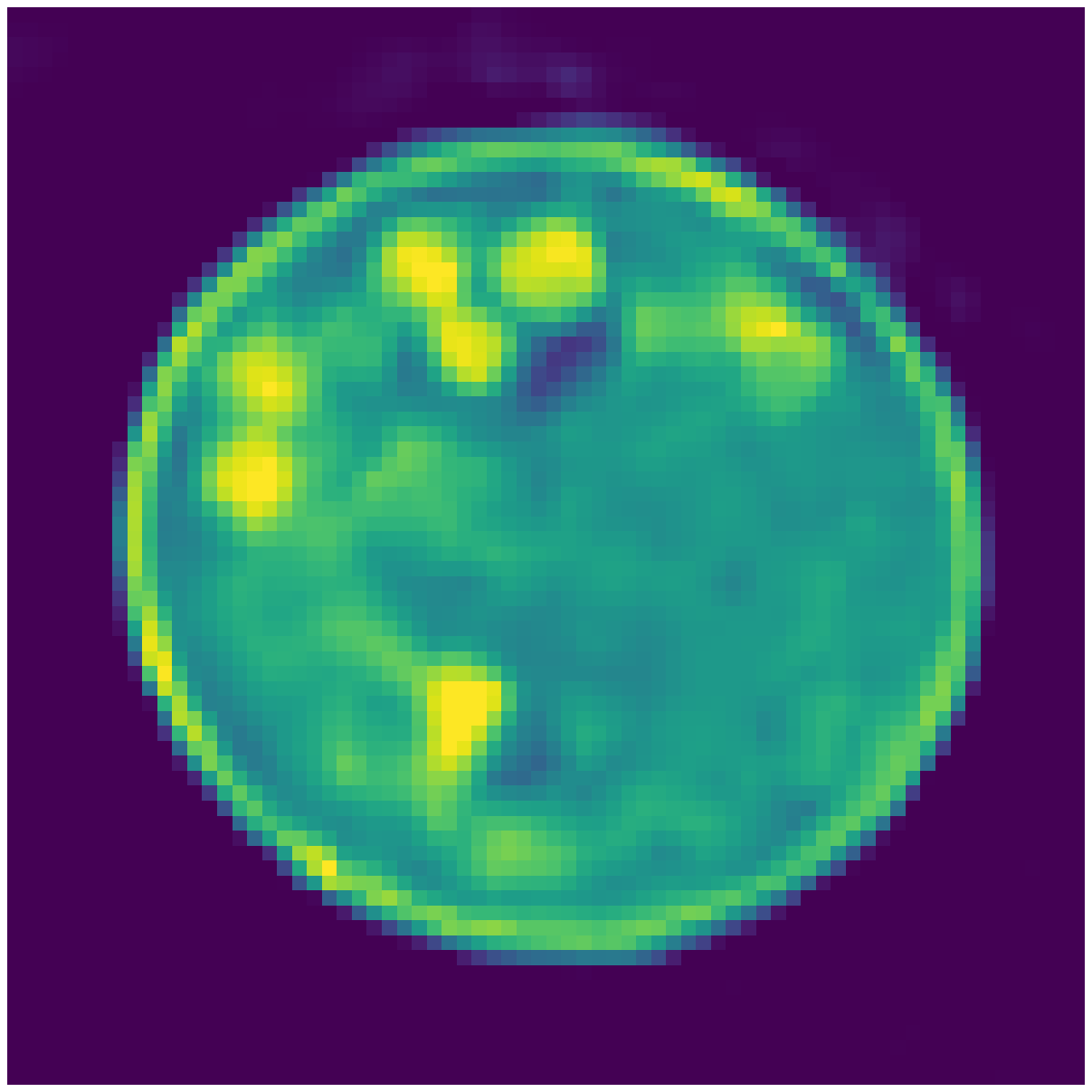}};
    \draw[-{stealth[length = 0.05mm]}, ultra thick, white] (3*72 + 41,36) -- (3*72 +  36, 40);
    \end{axis}
    \end{tikzpicture}
    \caption{
    Reconstructions of the yeast cell with Rytov, BPM, and the proposed method (LS model). 
   The first column corresponds to the central  XZ  slice  of  the  sample. Then, from left to right: XY slices at  depths $z_1 = -1.092\micro{m}$, $z_1 = -0.496\micro{m}$, and $z_1 = 0\micro{m}$.
   }
    \label{fig:real_slice}
\end{figure*}
\subsubsection{Comparisons}

We compare our LS-based reconstruction method with  the direct back-propagation  algorithm that is based on the Rytov model. In addition, we do compare it to BPM.
For each iterative method (BPM and ours),
we used TV regularization together with a nonnegativity constraint.
Finally, the regularization parameter~$\tau > 0$ was optimized through grid search in each scenario to maximize the performance with respect to the ground truth.

In Figure~\ref{fig:sim_slice}, one observes that our method faithfully recovers RBCs at several orientations. In comparison with the considered baselines, we observe that the LS model allows to recover more accurately the RBCs shape (and RI) as pointed out by the white arrows.
In Table~\ref{tab:relerr}, we present the relative error of the RBCs reconstructions.
As expected, the more sophisticated LS model obtains the lowest relative error.

\subsection{Real Data}\label{sec:ReconsReal}

\subsubsection{Acquisition Setup}

We acquired real data using the experimental tomographic setup described in~\cite{ayoub2019method}.
The sample is a yeast cell immersed in water~($\nb=1.338$) and is illuminated by tilted incident waves with wavelength $\lambda=532$\nano\meter.
As in our simulation setup, we acquired 61 views within a cone of illumination whose half-angle is~$35^\degree$.
The measurements lie on a plane that is centered and perpendicular to the optical axis.
The complex fields with and without the sample were acquired for each view,
thus providing the total and incident field, respectively.
The pixel size is $99$\nano\meter.

The  reconstructions are performed on a grid of the same resolution than that of the measurements.
We used the Hessian-Schatten-norm regularization as we found it more suitable for this type of sample.
Finally, we model $\Pd$ as the composition of a linear filtering by an ideal pupil function (binary disk in Fourier domain with radius $2\mathrm{NA} /\lambda$, $\mathrm{NA}=1.45$) and a free-space propagation to the center of the sample.
\begin{figure}[t]
\begin{center}
\begin{tikzpicture}
    \node[inner sep=0] (ryt) at (0,0) {\includegraphics[width = 0.16\textwidth,trim={5cm 3cm 5cm 3cm},clip]{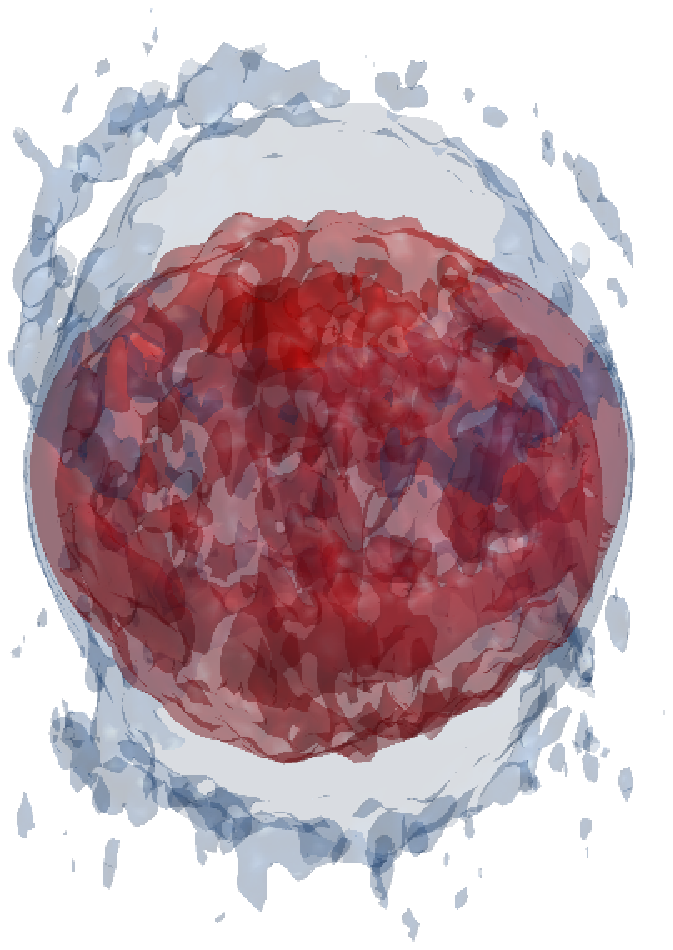}};
    \node[anchor=south] at (ryt.north) {Rytov};
    
    \node[inner sep=0,anchor = west] (bpm) at (ryt.east) {\includegraphics[width = 0.16\textwidth,trim={5cm 3cm 5cm 3cm},clip]{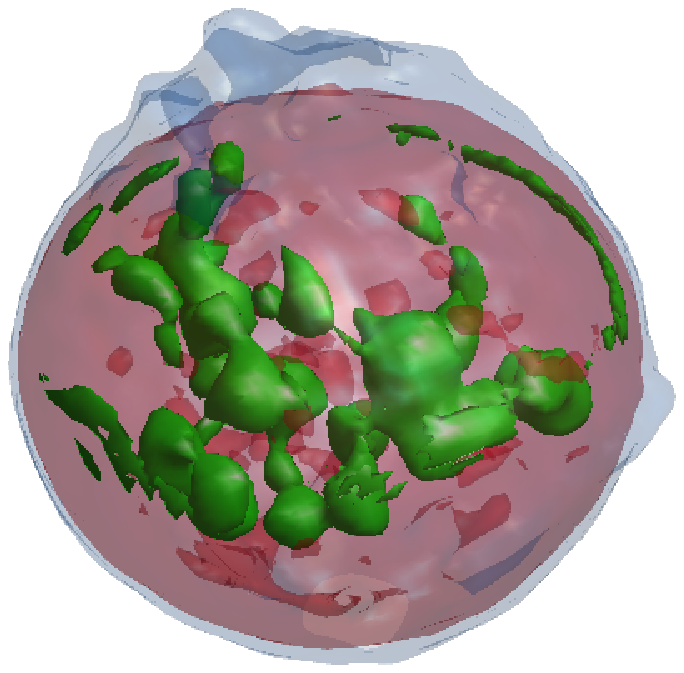}};
    \node[anchor=south] at (bpm.north) {BPM};
    
    \node[inner sep=0,anchor=west] (lipp) at (bpm.east) {\includegraphics[width = 0.16\textwidth,trim={5cm 3cm 5cm 3cm},clip]{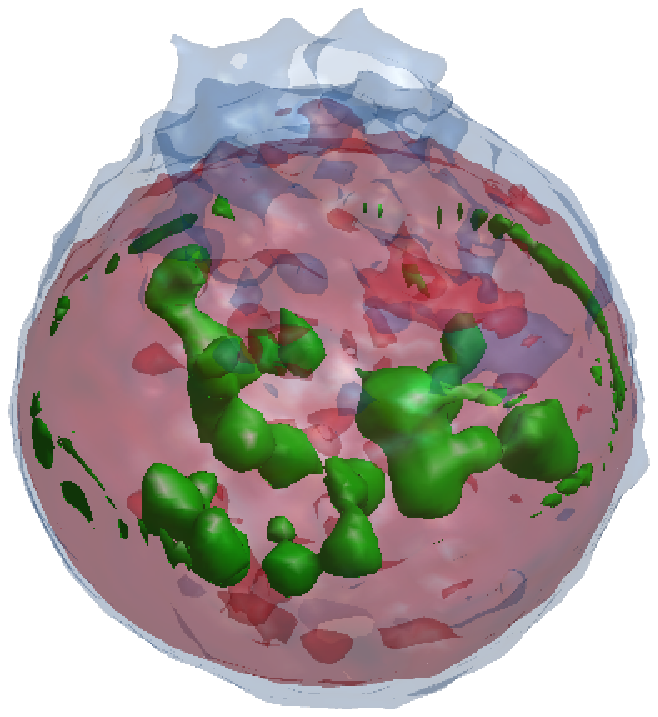}};
    \node[anchor=south] at (lipp.north) {LS model};
\end{tikzpicture}
    \caption{Iso-surface color renderings of the reconstructions of the yeast.
    The isovalues are $1.35$, $1.38$, and $1.46$ for the blue, red, and green color channels, respectively.
    }
    \label{fig:realiso}
    \end{center}
\end{figure}
\subsubsection{Reconstruction Results and Discussion}

The reconstructed volumes obtained with the Rytov method, the BPM, and the proposed approach are presented in Figure~\ref{fig:real_slice}.
Once again, nonlinear models clearly outperform the (linear) Rytov reconstruction.
Moreover, the reconstruction of the RI obtained by the LS model does not suffer from the artefacts indicated in BPM slices  $z_2, z_3$ with thick white arrows.
Also, the areas with higher RI are better resolved ($z_1, z_2$, thin red arrows) when the LS model is deployed.
Finally, one can appreciate in Figure~\ref{fig:realiso} that the inner areas with higher RI (green) are more resolved for the LS model than for BPM.

\section{Conclusion}

Three-dimensional optical diffraction tomography reconstruction is a challenging inverse problem. Its success depends on the accuracy of the implementation of the physical model. In this work, we proposed an accurate and efficient implementation of the forward model that is based on the exact Lippmann-Schwinger model. To that end, we tackled important difficulties that are related to the discretization of the model, the computational and memory burden, as well as the calibration of the incident field. Finally, we showed on both simulated and real data that the use of the  proposed model improves the quality/faithfulness of the reconstructions.

\section{Acknowledgments}
This research was supported by the European Research Council (ERC) under the European Union’s Horizon 2020 research and innovation programme, Grant Agreement No.\ 692726 GlobalBioIm: Global integrative framework for computational bio-imaging.

The authors would like to thank Ferr\'eol Soulez and Harshit Gupta for fruitful discussions.

\section{Preliminary Lemmas}

\begin{lemma}[Smoothness of a function and decay of its Fourier transform in $\R^3$]\label{lemma:DecayFourier}
  Let $v \in L_2(\R^3)$ have $(q-1)$ continuous derivatives in $L_2(\R^3)$ for some $q\geq 1$ and a $q$th derivative of bounded variations. Then, 
  \begin{equation}
      |\hat{v}(\omegad)| \leq \frac{C_1}{\|\omegad\|^{q+1}}\quad \forall \omegad \text{ s.t. } \|\omegad\| \geq C_2,
  \end{equation}
  where $C_1$ and $C_2$ are positive constants.
\end{lemma}
\begin{proof}
It is an extension of the well known result in one-dimension, see for instance~\cite[Theorems 6.1 and 6.2]{briggs1995dft}.
\end{proof}

\begin{lemma}[DFT aliasing for compactly supported functions in $\R^3$]\label{lemma:aliasing}
 Let $v \in L_2([-L/2,L/2]^3)$ be compactly supported, have $(q-1)$ continuous derivatives in $L_2(\R^3)$ for some $q\geq 3$, and a $q$th derivative of bounded variations. Let $\mathbf{v} \in \R^{N}$ ($N=n^3$) be a sampled version of $v$ with sampling step $h=L/n$. Finally,  denote by $\delta = 2\pi /(hn)$ the frequency sampling step of $\hat{\mathbf{v}}$, the DFT of $\mathbf{v}$.
Then, for all 
$\qd \in [\![\frac{-n}{2}+1;\frac{n}{2}]\!]^3$
\begin{equation}
     \left|\hat{v}(\delta \qd) -
     h^3 \widehat{\mathbf{v}}[\qd]\right| \leq C h^{q+1}
\end{equation}
for a positive constant $C>0$.
\end{lemma}

\begin{proof}
From Poisson's summation formula and the compact support of $v$, we have that
\begin{equation}\label{eq:Poisson}
    \sum_{\kd \in [\![\frac{-n}{2}+1;\frac{n}{2}]\!]^3} \mkern-15mu   \mathbf{v}[\kd] \mathrm{e}^{- \ii h \kd^T \omegad} = \frac{1}{h^3}\sum_{\md \in \Z^3} \hat{v}(\omegad + 2\pi \md/h).
\end{equation}
Setting $\omegad = \delta \qd  = 2 \pi \qd /(hn) $ in~\eqref{eq:Poisson}, one recognizes that the left-hand side is the DFT of $\mathbf{v}$.
Hence, we obtain that
\begin{equation}\label{eq:Poisson2}
    \hat{v}(\delta \qd) = 
     h^3 \widehat{\mathbf{v}}[\qd] - \sum_{\substack{\md \in \Z^3 \\ \md \neq \mathbf{0}}} \hat{v}\left( \delta \qd + 2\pi \md/h \right).
\end{equation}
Then, from Lemma~\ref{lemma:DecayFourier}, we obtain that there exists $C>0$ such that
\begin{align}
     \left|\hat{v}(\delta \qd) -
     h^3 \widehat{\mathbf{v}}[\qd]\right| & \leq  \sum_{\substack{\md\in \Z^3 \\ \md \neq \mathbf{0}}} \frac{C}{\| \delta \qd + 2\pi \md/h \|^{q+1}} \notag \\
     & \leq \frac{C h^{q+1}}{(2\pi)^{q+1} }\sum_{\substack{\md\in \Z^3 \\ \md \neq \mathbf{0}}} \frac{1}{\| \qd/n +  \md \|^{q+1}}\label{eq:ProofLemmaA2-1}
\end{align}
Let us now study the convergence of the  series in~\eqref{eq:ProofLemmaA2-1}. Using the fact that $\| \cdot \|_2 \leq \|\cdot\|_1 \leq \sqrt{N} \|\cdot\|_2$, we obtain that
\begin{equation}
    \sum_{\substack{\md \in \Z^3 \\ \md \neq \mathbf{0}}} \frac{1}{\| \qd/n +  \md \|^{q+1}_2} \leq \sum_{\substack{\md\in \Z^3 \\ \md \neq \mathbf{0}}} \frac{\sqrt{N}}{\| \qd/n +  \md \|^{q+1}_1}.\label{eq:ProofLemmaA2-2}
\end{equation}
Then, for $\qd  \in [\![\frac{-n}{2}+1;\frac{n}{2}]\!]^3$  and $m \in \N$  we introduce the set
\begin{equation}
    S_{\qd}^m = \left\lbrace \md \in \Z^3 : m \leq \|\qd/n + \md \|_1< m +1 \right\rbrace.
\end{equation}
Using the fact that $\qd \in [\![\frac{-n}{2}+1;\frac{n}{2}]\!]^3 \Rightarrow \qd/n \in (-1/2,1/2]^3$, we have that
\begin{equation}
    \|\md\|_1 - 3/2 \leq  \|\qd/n + \md \|_1 \leq \|\md\|_1 + 3/2,
\end{equation}
which implies that
\begin{align}
    |S_{\qd}^m| &  \leq \sum_{m'=m-2}^{m+2} |S_{\bm{0}}^{m'}| \notag \\
    & \leq 5 |S_{\bm{0}}^{m+2}| = 5\left( 4(m+2)^2 + 2\right),\label{eq:ProofLemmaA2-3}
\end{align}
where $|\cdot|$ stands for the cardinality of the set. 
Using the inequality~\eqref{eq:ProofLemmaA2-3}, we can bound the right-hand side of~\eqref{eq:ProofLemmaA2-2} as
\begin{align}
   \sum_{\substack{\md\in \Z^3 \\ \md \neq \mathbf{0}}} \frac{\sqrt{N}}{\| \qd/n +  \md \|^{q+1}_1} & \leq \sum_{m=1}^{+\infty} \frac{\sqrt{N} |S_{\qd}^m|}{m^{q+1}} \notag \\
   & \leq  \sum_{m=1}^{+\infty} \frac{5\sqrt{N}  \left( 4(m+2)^2 + 2\right)}{m^{q+1}},
\end{align}
which is a convergent series when $q\geq 3$. This completes the proof.
\end{proof}

\section{Proof of Theorem~\ref{th:DiscG}} \label{proof:thDiscG}

From the Fourier-convolution theorem, we have that
\begin{align}
    (g_\mathrm{t} \ast v)(\xd) & =  \int_\Omega g_{\mathrm{t}}(\xd - \zd)v(\zd) \, \dd{\zd} \notag \\
    & = \frac{1}{(2\pi)^3} \int_{\R^3}  \widehat{g_{\mathrm{t}}}(\omegad) \hat{v}(\omegad) \mathrm{e}^{ \ii \omegad^T \xd} \, \dd{\omegad}. \label{eq:intFourier}
\end{align}

Let $n = \in 2 \N \setminus \{0\}$  and  $h=L/n$ be the spatial sampling step of the volume $\Omega$ in each dimension. It follows that the frequency domain that is associated to the DFT is $\widehat{\Omega}=[-\pi /h, \pi /h]^3$. Then,  the padding factor $p \in \N_{>0}$  enlarges the spatial domain to $[-pL/2, pL/2]^3$, resulting in the frequency sampling step  $\delta = 2\pi /(hnp)=2 \pi /(Lp)$, so that $\widehat{\Omega}$ is sampled using $np$ equally spaced points in each dimension.

We are now equipped to discretize the integral in~\eqref{eq:intFourier}. To that end, we use a trapezoidal quadrature rule on $\widehat{\Omega}$ and write that
\begin{equation} \label{eq:quadrature}
     (g_\mathrm{t} \ast v)(\xd) \approx  \frac{\delta^3}{(2\pi)^3}  \mkern-10mu \sum_{\qd \in  [\![\frac{-np}{2};\frac{np}{2}]\!]^3} \mkern-10mu w_{\qd} \,   \widehat{g_{\mathrm{t}}}\left(\delta\qd\right) \hat{v}\left(\delta\qd\right)  \mathrm{e}^{ \ii \delta \qd^T \xd}.
\end{equation}
There, the weights $w_{\qd}$ are equal to $1$, $1/2$, $1/4$, and $1/8$ when $\qd$ belongs to the interior, the interior of the faces, the interior of the edges, and the corners of the cube $[\![\frac{-np}{2};\frac{np}{2}]\!]^3$, respectively. 

The approximation we made in~\eqref{eq:quadrature} generates two error terms. 
\begin{enumerate}
    \item The error $\varepsilon^\mathrm{tp}$ that is due to the trapezoidal quadrature rule used to approximate the integral over the domain $\widehat{\Omega}$. This error is well documented in the  literature~\cite{ralston2001first}. For integrand that are twice differentiable, such as $\omegad \mapsto  \widehat{g_{\mathrm{t}}}(\omegad) \hat{v}(\omegad) \mathrm{e}^{ \ii \omegad^T \xd} $, we have that 
    \begin{equation}
        |\varepsilon^\mathrm{tp}| \leq C \delta^2 = C \left( \frac{2\pi}{Lp} \right)^2
    \end{equation}
    for a positive constant $C>0$.
    \item The error $\varepsilon^\mathrm{tr}$ that is due to the truncation of the integral in~\eqref{eq:intFourier} to the domain $\widehat{\Omega}$, bounded as
    \begin{align}
        |\varepsilon^\mathrm{tr}| & =  \frac{1}{(2\pi)^3} \left|\int_{\R^3 \setminus \widehat{\Omega}}  \widehat{g_{\mathrm{t}}}(\omegad) \hat{v}(\omegad) \mathrm{e}^{ \ii \omegad^T \xd} \, \dd{\omegad} \right| \notag \\
        & \leq  \frac{1}{(2\pi)^3} \int_{\R^3 \setminus \widehat{\Omega}}  \left| \widehat{g_{\mathrm{t}}}(\omegad) \hat{v}(\omegad) \mathrm{e}^{ \ii \omegad^T \xd} \right|  \, \dd{\omegad} \notag  \\
        & \leq  \frac{C}{(2\pi)^3} \int_{\R^3 \setminus \widehat{\Omega}}  \frac{2}{(\|\omegad\| - \kb ) \|\omegad\|^{q+2}} \, \dd{\omegad},\label{eq:truncErrBound}
    \end{align}
    for a constant $C>0$.
\end{enumerate}    
    
    The last inequality in~\eqref{eq:truncErrBound} has been established in two steps. First, the assumption that $\kb < \pi /h$  implies that $\forall \omegad \in \R^3 \setminus \widehat{\Omega}$, $\|\omegad\| > \kb$. Then, one gets from~\eqref{eq:GreenL} that, $\forall \omegad \in \R^3 \setminus \widehat{\Omega}$,
    \begin{equation}
        |\widehat{g_{\mathrm{t}}}(\omegad)| \leq \frac{2}{(\|\omegad\|-\kb)\|\omegad\|}.
    \end{equation}
    Second, Lemma~\ref{lemma:DecayFourier}, along with the fact that $v$ has ($q-1$) continuous derivatives with a $q$th derivative of bounded variations, implies that its Fourier transform decays as 
    \begin{equation}
        |\hat{v}(\omegad) |\leq \frac{C}{\|\omegad\|^{q+1}}
    \end{equation}
    for a constant $C>0$. Combining these two bounds with $| \mathrm{e}^{ \ii \omegad^T \xd} |= 1$ finally leads to~\eqref{eq:truncErrBound}. 
    
    A further refinement of the bound~\eqref{eq:truncErrBound} is needed to recover the statement of Theorem~\ref{th:DiscG}.
    Denoting by $\mathcal{B}_{\pi/h}^2 = \{ \omegad \in \R^3 : \|\omegad \| \leq \pi /h \}$ the $\ell_2$-ball of radius~$\pi/h$, one sees that the integral in~\eqref{eq:truncErrBound} is upper-bounded by the integration of the same integrand over the larger domain $\R^3 \setminus \mathcal{B}_{\pi/h}^2$. This bound is easier to evaluate using spherical coordinates, as in 
    \begin{align}
         |\varepsilon^\mathrm{tr}| & \leq \frac{2C}{(2\pi)^3} \int_{\R^3 \setminus \mathcal{B}_{\pi/h}^2}  \frac{1}{(\|\omegad\| - \kb ) \|\omegad\|^{q+2}} \, \dd{\omegad} \notag \\
         & = \frac{2C}{(2\pi)^3} \int_0^{2\pi} \mkern-10mu \int_0^{\pi} \mkern-10mu \int_{\pi/h}^{+\infty} \frac{r^2 \sin(\theta)}{(r-\kb)r^{q+2}} \, \dd{r} \, \dd{\theta} \, \dd{\phi} \notag \\
         & = \frac{C}{\pi^2} \int_{\pi/h}^{+\infty} \frac{1}{(r-\kb)r^{q}} \, \dd{r} . \label{eq:intProof1} 
    \end{align}
    To evaluate~\eqref{eq:intProof1}, we use the partial fraction decomposition
    \begin{equation}
        \frac{1}{(r-\kb)r^{q}} = \frac{1}{\kb^q(r-\kb)} - \sum_{m=0}^{q-1} \frac{1}{\kb^{q-m}r^{m+1}}.
    \end{equation}
    Hence, we have that
    \begin{align}
     \mkern-10mu   |\varepsilon^\mathrm{tr}| & \leq \frac{C}{\pi^2} \bigg( \frac{1}{\kb^q}  \log(r-\kb) \Big|_{r=\frac{\pi}{h}}^{+\infty}  -  \frac{1}{\kb^q}  \log(r) \Big|_{r=\frac{\pi}{h}}^{+\infty}  \notag \\
        & \qquad  \qquad \qquad - \sum_{m=1}^{q-1} \frac{1}{\kb^{q-m}} \left(- \frac{1}{m r^m} \right) \Bigg|_{r=\frac{\pi}{h}}^{+\infty}  \bigg) \notag \\
        & = \frac{-C}{\kb^q \pi^2} \left( \log\left(1 - \frac{\kb h}{\pi} \right)\mkern-5mu + \mkern-5mu \sum_{m=1}^{q-1} \frac{1}{m} \left(\frac{\kb h}{\pi} \right)^m  \right) \label{eq:intProof2}  \\
        & = \frac{C}{\kb^q \pi^2} \sum_{m=q}^{+\infty} \frac{1}{m} \left(\frac{\kb h}{\pi} \right)^m  \label{eq:intProof3}  \\
        & = \frac{C}{\kb^q \pi^2} \left(\frac{\kb h}{\pi} \right)^q \sum_{m=0}^{+\infty} \left(\frac{\kb h}{\pi} \right)^m \frac{1}{m+q} . \label{eq:intProof4}
     \end{align}
     To obtain~\eqref{eq:intProof3} from~\eqref{eq:intProof2}, we used the fact that $\kb h/\pi < 1$ together with $\log(1-x) = (- \sum_{m=1}^{+\infty} x^m / m)$ for $|x| < 1$. Finally, we get the bound $C^\mathrm{tr}/n^q$ from the convergence of the series in~\eqref{eq:intProof4} and $h=L/n$.

Let us focus on aliasing. As opposed to $\widehat{g_\mathrm{t}}$ for which we have access to  an explicit expression in~\eqref{eq:GreenL}--\eqref{eq:GreenL2}, the samples $\hat{v}(\delta \qd)$ in~\eqref{eq:quadrature} have to be approximated by the DFT coefficients of a $p$-times zero-padded version of the sampled signal $\mathbf{v}\in \C^N$, denoted $\mathbf{v}_p \in \C^{Np^3}$, and defined by, $\forall \kd \in [\![\frac{-np}{2}+1;\frac{np}{2}]\!]^3$,
\begin{equation}\label{eq:padV}
    \mathbf{v}_p[{\kd}] = \left\lbrace \begin{array}{ll}
         \mathbf{v}[{\kd}]  = v(h\kd),  & \kd \in  [\![\frac{-n}{2}+1;\frac{n}{2} ]\!]^3 \\
         0, & \text{ otherwise} .
    \end{array} \right.
\end{equation}
We then replace $\hat{v}(\delta \qd)$ in~\eqref{eq:quadrature} by $h^3\widehat{\mathbf{v}_p}[\qd]$ and obtain that
\begin{equation} \label{eq:quadrature2}
     (g_\mathrm{t} \ast v)(\xd) \approx  \frac{1}{(np)^3}  \mkern-10mu \sum_{\qd \in  [\![\frac{-np}{2};\frac{np}{2}]\!]^3} \mkern-10mu w_{\qd} \,   \widehat{g_{\mathrm{t}}}\left(\delta\qd\right) \widehat{\mathbf{v}_p}[\qd]  \mathrm{e}^{ \ii \delta \qd^T \xd}.
\end{equation}
This approximation  introduces an error term $\varepsilon^\mathrm{al}$ that is due to aliasing. More precisely, we have that
\begin{align}
    |\varepsilon^\mathrm{al}| & \leq  \frac{\delta^3}{(2\pi)^3}  \mkern-10mu \sum_{\qd \in  [\![\frac{-np}{2};\frac{np}{2}]\!]^3} \mkern-10mu w_{\qd} \left|\widehat{g_{\mathrm{t}}}\left(\delta\qd\right) \right| \left|\hat{v}(\delta \qd) -
     h^3 \widehat{\mathbf{v}_p}[\qd]\right| \notag \\
     & \leq   \frac{\delta^3 }{(2\pi)^3} \mkern-10mu \sum_{\qd \in  [\![\frac{-np}{2};\frac{np}{2}]\!]^3} \mkern-10mu w_{\qd} \left|\widehat{g_{\mathrm{t}}}\left(\delta\qd\right) \right| C h^{q+1} \label{eq:aliasingErr-1}\\
     & \leq \frac{\delta^3 C h^{q+1} }{(2\pi)^3} (np)^3 \|\widehat{g}_\mathrm{t}\|_{\infty}  \notag \\
     & = C \|\widehat{g}_\mathrm{t}\|_{\infty} h^{q-2} = \frac{C \|\widehat{g}_\mathrm{t}\|_{\infty} L^{q-2}}{n^{q-2}},
\end{align}
where~\eqref{eq:aliasingErr-1} comes from Lemma~\ref{lemma:aliasing}.

To complete the proof, it remains to recognize an inverse DFT within~\eqref{eq:quadrature2}.
Let $\{\qd_i\}_{i=1}^8$ denotes the eight corners of the cube $[\![\frac{-np}{2};\frac{np}{2}]\!]^3$.
Then, because $\widehat{g_\mathrm{t}}$ is radially symmetric (see~\eqref{eq:GreenL} and~\eqref{eq:GreenL2}), and by periodicity of $\widehat{\mathbf{v}_p}$, we have that
\begin{equation}
    \widehat{g_{\mathrm{t}}}\left(\delta\qd_i\right) \widehat{\mathbf{v}_p}[\qd_i] = \widehat{g_{\mathrm{t}}}\left(\delta\qd_1\right) \widehat{\mathbf{v}_p}[\qd_1], \; \forall i \in \{2,\ldots,8\}.
\end{equation}
Hence we can factorize the corresponding terms in~\eqref{eq:quadrature2} as 
\begin{equation}
    \sum_{\qd \in \{\qd_i\}^8_{i=1}} \frac18  \widehat{g_{\mathrm{t}}}\left(\delta\qd_i\right) \widehat{\mathbf{v}_p}[\qd_i]  = \widehat{g_{\mathrm{t}}}\left(\delta\qd_1\right) \widehat{\mathbf{v}_p}[\qd_1].
\end{equation}
Finally, using the same arguments for points within the faces and edges of the cube $[\![\frac{-np}{2};\frac{np}{2}]\!]^3$, and sampling~\eqref{eq:quadrature2} at points $h\kd$, $\kd \in [\![\frac{-n}{2}+1;\frac{n}{2}]\!]^3$, we obtain that
\begin{equation} \label{eq:quadrature3}
    (\Gd \mathbf{v})[\kd] =  \frac{1}{(np)^3}  \mkern-10mu \sum_{\qd \in  [\![\frac{-np}{2}+1;\frac{np}{2}]\!]^3} \mkern-10mu   \widehat{\mathbf{g}_{\mathrm{t}}}[\qd] \widehat{\mathbf{v}_p}[\qd]  \mathrm{e}^{ \frac{2 \ii \pi}{np} \qd^T \kd},
\end{equation}
where $\widehat{\mathbf{g}_{\mathrm{t}}} = (\widehat{g_{\mathrm{t}}}(\delta \qd))_{\qd \in [\![\frac{-np}{2}+1;\frac{np}{2}]\!]^3}$. We recognize an inverse DFT, which completes the proof.

\section{Proof of Proposition~\ref{th:ReducedMem}} \label{proof:thReducedMem}
First, let us introduce the notation $\Omega_n = [\![\frac{-n}{2}+1;\frac{n}{2}]\!]^3$. Then, we have that, for all $\kd \in \Omega_n$,
\begin{align}
    \mkern-15mu & \left( \mathbf{F}^{-1} (\widehat{\mathbf{g}_{\mathrm{t}}} \odot \widehat{\mathbf{v}_{\mathrm{p}}}) \right)[\kd]  \notag \\ \mkern-15mu
    &= \mkern-5mu \frac{1}{(np)^3}  \sum_{\qd \in \Omega_{np}}  \widehat{\mathbf{g}_{\mathrm{t}}}[\qd] \widehat{\mathbf{v}_p}[\qd] \, \mathrm{e}^{ \frac{2 \ii \pi}{np} \qd^T \kd} \notag \\ \mkern-15mu
     &= \mkern-5mu\frac{1}{(np)^3}  \sum_{\qd \in \Omega_{np}}   \widehat{\mathbf{g}_{\mathrm{t}}}[\qd]  \sum_{\tilde{\qd} \in \Omega_{np}} \mathbf{v}_p[\tilde{\qd}] \,  \mathrm{e}^{ \frac{-2 \ii \pi}{np} \tilde{\qd}^T \mkern-2mu \qd} \, \mathrm{e}^{ \frac{2 \ii \pi}{np} \qd^T \kd} \notag \\ \mkern-15mu 
    & = \mkern-5mu\frac{1}{(np)^3}   \sum_{\tilde{\qd} \in \Omega_{2n}}  \mathbf{v}_2[\tilde{\qd}]  \sum_{\qd \in \Omega_{np}}    \widehat{\mathbf{g}_{\mathrm{t}}}[\qd]   \,  \mathrm{e}^{ \frac{2 \ii \pi}{np} (\kd - \tilde{\qd})^T \qd}  \label{eq:proofthReducedMem-1} \notag  \\ \mkern-15mu
     & = \mkern-5mu\frac{1}{(np)^3}  \sum_{\tilde{\qd} \in \Omega_{2n}}  \mkern-13mu \mathbf{v}_2[\tilde{\qd}] \mkern-10mu \sum_{\substack{\bm{s} \in [\![0;\frac{p}{2}-1]\!]^3 \\ \qd \in \Omega_{2n}}} \mkern-10mu  \widehat{\mathbf{g}_{\mathrm{t}}}[{\textstyle \frac{p}{2}}\qd \mkern-3mu-\mkern-3mu\bm{s}]   \,  \mathrm{e}^{ \frac{2 \ii \pi}{np} (\kd - \tilde{\qd})^T ({\textstyle \frac{p}{2}}\qd-\bm{s})} \mkern-30mu \notag  \\ \mkern-15mu
     & = \mkern-5mu \frac{8}{p^3} \mkern-10mu \sum_{\tilde{\qd} \in \Omega_{2n}} \mkern-10mu  \mathbf{v}_2[\tilde{\qd}]   \mkern-20mu   \sum_{\bm{s} \in [\![0;\frac{p}{2}-1]\!]^3} \mkern-23mu  \mathbf{F}^{-1}(\widehat{\mathbf{g}_{\mathrm{t}}}[{\textstyle \frac{p}{2}}\cdot \mkern-5mu - \mkern-3mu\bm{s}] )[\kd\mkern-2mu - \mkern-2mu \tilde{\qd}] \,  \mathrm{e}^{ \frac{-2 \ii \pi}{np} (\kd - \tilde{\qd})^T \mkern-5mu  \bm{s}}   \mkern-10mu
\end{align}
where we have used the fact that $\mathrm{supp}(\mathbf{v}_p) =\mathrm{supp}(\mathbf{v}) \subseteq \Omega_n \subseteq \Omega_{2n}$.
Hence, we have shown that $\big(\mathbf{F}^{-1}(\widehat{\mathbf{g}_{\mathrm{t}}} \odot \widehat{\mathbf{v}_{\mathrm{p}}})\big) \big|_{\Omega_n}$ can be obtained as the valid part of the discrete convolution between $\mathbf{v}_2$, defined as $\mathbf{v}$ padded with $p=2$, and a modified truncated Green function given by, $\forall \kd \in \Omega_{2n}$,
\begin{equation}
    \mathbf{g}_{\mathrm{t}}^\mathrm{m}[\kd]=  \frac{8}{p^3} \sum_{\bm{s} \in [\![0;{\textstyle \frac{p}{2}}-1]\!]^3} \mkern-10mu \mathbf{F}^{-1}(\widehat{\mathbf{g}_{\mathrm{t}}}[{\textstyle \frac{p}{2}}\cdot-\bm{s}] )[\kd] \,  \mathrm{e}^{ \frac{-2 \ii \pi}{np} \kd^T \bm{s}} ,
\end{equation}
which completes the proof.

\bibliographystyle{IEEEtran}
\bibliography{ODT}

\end{document}